\documentclass[12pt,reqno]{amsart}
\usepackage{graphicx}
\usepackage{amssymb,amsmath}
\usepackage{amsthm}
\usepackage{color}
\usepackage[pdf]{pstricks}
\usepackage{hyperref}
\usepackage{empheq}
\usepackage{pgfplots}
\usepackage{subcaption}
\usepackage{extarrows}
\usepackage{calc}

\usepackage{arydshln}

\newtheorem{remark}{Remark}
\newtheorem{lemma}{Lemma}

\newtheorem{prop}{Proposition}
\newtheorem{theorem}{Theorem}

\numberwithin{equation}{section}

\newcommand{\beq}{\begin{equation}}
\newcommand{\eeq}{\end{equation}}
\newcommand{\ben}{\begin{eqnarray}}
\newcommand{\een}{\end{eqnarray}}
\newcommand{\beno}{\begin{eqnarray*}}
\newcommand{\eeno}{\end{eqnarray*}}

\renewcommand{\theequation}{\thesection.\arabic{equation}}
\allowdisplaybreaks

\newdimen\eqjot \eqjot = 1\jot
\def\openupeq{\openup \the\eqjot}

\begin{document}
	\title[Rational solutions for algebraic solitons in MTM]{Rational solutions for algebraic solitons \\ in the massive Thirring model}
	
\author{Zhen Zhao}
\address[Z. Zhao]{School of Mathematics and Statistics, Ningbo University, Ningbo 315211, People's Republic of China and Department of Mathematics and Statistics, McMaster University, Hamilton, Ontario, Canada, L8S 4K1}
\email{zhaozhen00728@163.com}

\author{Cheng He}
\address[C. He]{School of Mathematics and Statistics, Ningbo University, Ningbo 315211, People's Republic of China}
\email{1811071003@nbu.edu.cn}
	
\author{Baofeng Feng}
\address[B. Feng]{School of Mathematical and Statistical Sciences, The University of Texas Rio Grande Valley, Edinburg, Texas, USA 78539}
\email{baofeng.feng@utrgv.edu}
	
\author{Dmitry E. Pelinovsky}
\address[D. E. Pelinovsky]{Department of Mathematics and Statistics, McMaster University, Hamilton, Ontario, Canada, L8S 4K1}
\email{pelinod@mcmaster.ca}
	
\begin{abstract}
An algebraic soliton of the massive Thirring model (MTM) is expressed by the simplest rational solution of the MTM with the spatial decay of $\mathcal{O}(x^{-1})$. The  corresponding potential is related to a simple embedded eigenvalue in the Kaup--Newell spectral problem. This work focuses on the hierarchy of rational solutions of the MTM, in which the $N$-th member of the hierarchy describes a nonlinear superposition of $N$ algebraic solitons with identical masses and corresponds to an embedded eigenvalue of algebraic multiplicity $N$. We show that the hierarchy of rational solutions can be constructed by using the double-Wronskian determinants. The novelty of this work is a rigorous proof that each solution is defined by a polynomial of degree $N^2$ with $2N$ arbitrary parameters, which admits $\frac{N (N-1)}{2}$ poles in the upper half-plane and $\frac{N(N+1)}{2}$ poles in the lower half-plane. Assuming that the leading-order polynomials have exactly $N$ real roots, we show that the $N$-th member of the hierarchy describes the slow scattering of $N$ algebraic solitons on the time scale $\mathcal{O}(\sqrt{t})$.
\end{abstract}   
	
	\maketitle
	
{\it Corresponding author: Dmitry E. Pelinovsky; e-mail: pelinod@mcmaster.ca}
	
\section{Introduction}

\subsection{Motivations}

The massive Thirring model (MTM) in laboratory coordinates was proposed in quantum field theory as a relativistically invariant nonlinear Dirac equation \cite{Thirring-AOP-1958}. Integrability of the MTM in one spatial dimension was shown in \cite{Mikhailov-JETP-1976} and developed in  \cite{Kaup-ANC-1977,Kuznetsov-TMP-1977,Orfanidis-PRD-1976}. The MTM is the only integrable case of the coupled-mode theory used for the homogenization of the Gross--Pitaevskii equation with a periodic potential \cite{Pel11}. 

Stability of solitary waves in the time evolution of the nonlinear Dirac equations in one spatial dimension is a challenging problem due to the lack of Lyapunov functional provided by the mass, momentum, and energy \cite{BC19}. Nevertheless, 
orbital stability of exponentially decaying solitons was proven in the MTM due to its integrability, by using the higher-order energy \cite{PS14} and
the B\"{a}cklund--Darboux transformation \cite{CPS16}. More recently, 
asymptotic stabillity of exponentially decaying solitons was shown in \cite{Cheng23} 
based on the development of the inverse scattering transform in \cite{PS19,PS18}. 

{\em The algebraic soliton} arises at the limiting point in the family of the exponentially decaying solitons, at which the soliton mass is maximal and the spatial tails are algebraically decaying. It was suggested in \cite{KPR06} that the algebraic solitons are stable with respect to perturbations, 
by using analysis of a simple embedded eigenvalue in the Kaup--Newell spectral problem related to the Lax system for the MTM. This conjecture was never proven, with some partial results on the orbital stability of algebraic solitons obtained in \cite{Wu18} for a similar model of the derivative NLS equation. 

Further progress in the study of algebraic solitons was achieved recently in \cite{Han-2024}, where the second-order rational solution to the MTM was constructed based on the bilinear (Hirota) method implemented for the MTM in \cite{Chen-SAPM-2023}. The second-order solution describes 
{\em the double algebraic soliton}, that is, two algebraic solitons with identical masses which scatter at the slow time scale of $\mathcal{O}(\sqrt{t})$. This solution suggests the existence of a hierarchy of rational solutions to the MTM with multiple algebraic solitons, similar to the hierarchy of rational solutions associated with the Zakharov--Shabat spectral problem existing up to arbitrary higher order, see \cite{bilman1,bilman2,bilman3,bilman4}. The second-order rational solution of the MTM was recovered in \cite{LiPelin2025} by using the double-pole solutions of the IST for the MTM, where it was shown that the double algebraic soliton is related to a double embedded eigenvalue in the Kaup--Newell spectral problem predicted in \cite{KPR06}. 

The main purpose of this work is to construct the hierarchy of rational solutions to the MTM, where the $N$-th member of the hierarchy describes a nonlinear superposition of $N$ algebraic solitons with identical masses and corresponds to an embedded eigenvalue of algebraic multiplicity $N$ in the Kaup--Newell spectral problem. Based on the recent development in \cite{Wang1,Wang2} for similar integrable equations, we construct the hierarchy of rational solutions by using the double-Wronskian determinants associated with the Jordan block for a multiple embedded eigenvalue. 

Compared to the previous works on rational solutions for integrable models, the novelty of this work is a rigorous proof that the $N$-the member of the hierarchy of rational solutions is defined by a polynomial of degree $N^2$ with $2N$ arbitrary parameters, which admits $\frac{N (N-1)}{2}$ poles in the upper half-plane and $\frac{N(N+1)}{2}$ poles in the lower half-plane. Under the assumption that the leading-order polynomials have exactly $N$ real roots, we show that the corresponding rational solution describes the slow scattering of $N$ algebraic solitons on the time scale $\mathcal{O}(\sqrt{t})$.

\subsection{Main results}

We write the integrable MTM system in the normalized form: 
\begin{equation}
\label{MTM}
\begin{cases}
{\rm i} (u_t+u_x)+v = |v|^2 u, \\
{\rm i} (v_t-v_x)+u = |u|^2 v,
\end{cases}
\end{equation}
where $(x,t) \in \mathbb{R}^2$ and $(u,v) \in \mathbb{C}^2$. The initial-value 
problem for the MTM system (\ref{MTM}) is known to be well-posed in  $L^2(\mathbb{R})$ \cite{Candy,Huh}, where the mass $M(u,v)$ is conserved in time:
\begin{equation}
\label{mass-intro}
M(u,v) = \int_{\mathbb{R}} (|u(x,t)|^2 + |v(x,t)|^2) dx.
\end{equation}
The nonlinear system (\ref{MTM}) is a compatibility condition for the Lax system of linear equations \cite{Mikhailov-JETP-1976}:
\begin{align}\label{Lax-pair}
\partial_x     \vec{\phi} +  L(u,v,\zeta)     \vec{\phi} = 0, \quad \partial_t     \vec{\phi} + M(u,v,\zeta)     \vec{\phi} = 0,
\end{align}
where $\zeta \in \mathbb{C}$ is the spectral parameter,
$ \vec{\phi} \in \mathbb{C}^2$ is the wave function, and
the $2 \times 2$ matrices $L(u,v,\zeta)$ and $M(u,v,\zeta)$ are given by
\begin{align*}
L=\frac{{\rm i}}{4}\left(|u|^2-|v|^2\right)\sigma_3 + \frac{{\rm i}}{2}\zeta\left(
\begin{array}{cc}
0 & \bar{v} \\
v & 0 \\
\end{array}
\right)+\frac{{\rm i}}{2\zeta}\left(
\begin{array}{cc}
0 & \bar{u} \\
u & 0 \\
\end{array}
\right)+\frac{{\rm i}}{4}\left(\zeta^2-\zeta^{-2}\right)\sigma_3
\end{align*}
and
\begin{align*}
M = -\frac{{\rm i}}{4}\left(|u|^2+|v|^2\right)\sigma_3 + \frac{{\rm i}}{2}\zeta\left(
\begin{array}{cc}
0 & \bar{v} \\
v & 0 \\
\end{array}
\right)-\frac{i}{2\zeta}\left(
\begin{array}{cc}
0 & \bar{u} \\
u & 0 \\
\end{array}
\right)+\frac{{\rm i}}{4}\left(\zeta^2+\zeta^{-2}\right)\sigma_3,
\end{align*}
with $\sigma_3 = {\rm diag}(1,-1)$ being Pauli's matrix. 
By using the variables 
\begin{equation}
\label{MTM-bil}
u=\frac{g}{\bar{f}}, \quad v=\frac{h}{f},
\end{equation}
the system (\ref{MTM}) can be rewritten in the bilinear form \cite{Chen-SAPM-2023}:
\begin{equation}
\label{bilinear}
\begin{cases}
{\rm i} (D_t+D_x)g\cdot f + h \bar{f} = 0, \\
{\rm i} (D_t-D_x)h\cdot\bar{f} + g f = 0, \\
{\rm i} (D_t+D_x)f\cdot\bar{f} - |h|^2 = 0, \\
{\rm i} (D_t-D_x)\bar{f}\cdot f - |g|^2 = 0. 
\end{cases}
\end{equation}

Our first result is a proof of the double-Wronskian solutions to the system of bilinear equations (\ref{bilinear}). Although similar solutions appear in the previous works, see Appendix A in \cite{Wang1}, we justify the validity of the 
explicit solutions independently and relate them to solutions of the two-component 
KP hierarchy.

To set up the double-Wronskian solutions, we use the characteristic coordinates
$(\xi,\eta) \in \mathbb{R}^2$ instead of the physical coordinates $(x,t) \in \mathbb{R}^2$:
\begin{equation}
\begin{cases}
t = 2(\xi + \eta),\\
x = 2(\xi - \eta),
\end{cases} \quad \Rightarrow \quad 
\begin{cases}
\partial_{\xi} = 2 (\partial_t + \partial_x),\\
\partial_{\eta} = 2 (\partial_t - \partial_x).
\end{cases}
\label{transformation}
\end{equation}
The bilinear equations (\ref{bilinear}) transform into the equivalent form:
\begin{equation}
\label{MTM-4}
\begin{cases}
{\rm i}D_\xi g \cdot f+2h \bar{f}=0, \\
{\rm i} D_\eta h \cdot \bar{f}+2g f=0, \\
{\rm i} D_\xi f \cdot \bar{f}-2 h\bar{h}=0, \\
{\rm i} D_\eta \bar{f} \cdot f - 2g\bar{g}=0.
\end{cases}
\end{equation}
The Lax pair (\ref{Lax-pair}) transforms into 
\begin{equation}
\label{Lax-1}
\partial_{\xi} \vec{\phi}  -{\rm i} |v|^2 \sigma_3 \vec{\phi} + 2{\rm i} \zeta \left(
\begin{array}{cc}
0 & \bar{v} \\
v & 0 \\
\end{array}
\right) \vec{\phi} + {\rm i} \zeta^2 \sigma_3 \vec{\phi} = 0
\end{equation}
and
\begin{equation}
\label{Lax-2}
\partial_{\eta} \vec{\phi} -{\rm i}|u|^2 \sigma_3 \vec{\phi} - 2{\rm i} \zeta^{-1} \left(
\begin{array}{cc}
0 & \bar{u} \\
u & 0 \\
\end{array}
\right) \vec{\phi} + {\rm i} \zeta^{-2} \sigma_3 \vec{\phi} = 0.
\end{equation}

The following theorem represents the double-Wronskian solutions of the bilinear equations (\ref{MTM-4}), which represent solutions of the MTM system (\ref{MTM})  via (\ref{MTM-bil}) and (\ref{transformation}).

\begin{theorem}
	\label{theorem-Wronskian}
	Fix $N \in \mathbb{N}$. Let $A \in \mathbb{C}^{2N \times 2N}$ be an invertible matrix which can be factorized by $S \in \mathbb{C}^{2N \times 2N}$ in the form
	\begin{equation}
	\label{rel-1}
	A = -S \bar{S}.
	\end{equation}
Define two vectors $\phi, \psi \in \mathbb{C}^{2N}$ from solutions of the linear equations 
\begin{equation}
\label{eigen}
\begin{cases}
\partial_{\xi} \phi = {\rm i} A \phi, \\
\partial_{\eta} \phi = {\rm i} A^{-1} \phi, 
\end{cases} 
\quad 
\mbox{\rm and} 
\qquad 
\begin{cases}
\partial_{\xi} \psi = -{\rm i} A \psi, \\
\partial_{\eta} \psi = -{\rm i} A^{-1} \psi,
\end{cases}
\end{equation}	
subject to the relation
\begin{equation}
\label{rel-2}
\psi = S \bar{\phi}.
\end{equation}
Then, the following double-Wronskian functions 
	\begin{equation}
	\label{Wronskian}
	\begin{cases}
	f = |\widetilde{N} ; \widehat{N-1}|, \\
	\bar{f} = C |\widetilde{N} ; \widetilde{N}|,
	\end{cases} \qquad 
	\begin{cases}
	g = |\widehat{N} ; \widetilde{N-1}|, \\
	\bar{g} = {\rm i}C |\overline{N} ; \widehat{N}|, 
	\end{cases} \qquad 
	\begin{cases}
	h = {\rm i} C^{-1} |\widehat{N} ; \widehat{N-2}| \\
	\bar{h} = C \bar{C}^{-1} |\widetilde{N-1} ; \widehat{N}|
	\end{cases}
	\end{equation}
satisfy the bilinear equations (\ref{MTM-4}) with $C = (-{\rm i})^N/|S|$.
\end{theorem}

\begin{remark}
	In Theorem \ref{theorem-Wronskian}, we use the following notations for double-Wronskian determinants of $(2N) \times (2N)$ matrices:
\begin{align*}
|\widehat{N-1} ; \widehat{N-1}| &= |\phi,\phi',\dots,\phi^{(N-1)};\psi,\psi',\dots,\psi^{(N-1)}|, \\
|\widetilde{N} ; \widetilde{N}| &= |\phi',\phi'',\dots,\phi^{(N)};\psi',\psi'',\dots,\psi^{(N)}|, \\
|\overline{N+1} ; \overline{N+1}| &= |\phi'',\phi''',\dots,\phi^{(N)};\psi'',\psi''',\dots,\psi^{(N+1)}|, 
\end{align*}
where $|A| = \det(A)$ and the prime stands for the derivative with respect to $\xi$. Similarly, in the proof of Theorem \ref{theorem-Wronskian}, we extend the 
definitions for the following one-column modications of the double-Wronskian determinants:
\begin{align*}
|0,\overline{N}; \widehat{N-1}| &= |\phi,\phi'',\dots,\phi^{(N)};\psi,\psi',\dots,\psi^{(N-1)}|, \\
|\widetilde{N};-1,\widetilde{N-1}|&= |\phi',\phi'',\dots,\phi^{(N)};\partial_{\xi}^{-1} \psi,\psi',\dots,\psi^{(N-1)}|.
\end{align*}
\end{remark}

\begin{remark}
	\label{rem-tau}
It is customary \cite{HirotaOhtaSatsuma1988,Kakei1988} to represent the double-Wronskian determinants by using the following tau-function:
\begin{align}
\label{double-Wr}
\tau^l_{n,m} &:= |\phi^{(n)},\phi^{(n+1)},\dots,\phi^{(n+N+l-1)};\psi^{(m)},\psi^{(m+1)},\dots,\psi^{(m+N-l-1)}|.
\end{align}
The double-Wronskian solutions (\ref{Wronskian}) are expressed by using the tau functions as
	\begin{equation}
	\label{Wronskian-tau}
	\begin{cases}
	f = \tau^0_{1,0}, \\
	\bar{f} = \frac{(-{\rm i})^N}{|S|} \tau^0_{1,1},
	\end{cases} \qquad 
	\begin{cases}
	g = \tau^1_{0,1}, \\
	\bar{g} = -\frac{(-{\rm i})^{N+1}}{|S|}  \tau^{-1}_{2,0}, 
	\end{cases} \qquad 
	\begin{cases}
	h = {\rm i}^{N+1} |S|  \tau^1_{0,0} \\
	\bar{h} = (-1)^N \tau^{-1}_{1,0}
	\end{cases}
	\end{equation}
\end{remark}

\begin{remark}
In the simplest case $N = 1$, the double-Wronskian solutions in Theorem \ref{theorem-Wronskian} recover the family of the exponentially decaying 
solitons in the explicit form:
\begin{equation}
u = \frac{\sin \gamma e^{{\rm i}t \cos \gamma}}{\cosh(x \sin \gamma + \frac{{\rm i}\gamma}{2})}, \quad 
v = \frac{\sin \gamma e^{{\rm i}t \cos \gamma}}{\cosh(x \sin \gamma - \frac{{\rm i}\gamma}{2})},
\label{1-soliton-expr}
\end{equation}
where the parameter $\gamma \in (0,\pi)$ is arbitrary. This solution is 
well-known and stability of exponential solitons was studied
in \cite{PS14}, \cite{CPS16}, and \cite{Cheng23}. A general family 
of the exponential solitons has four parameters, but 
two parameters are given by the translations in $(x,t)$ 
and the speed parameter can be added by using the Lorentz transformation of the MTM, see details in \cite{Han-2024}. The limit $\gamma \to \pi$ in (\ref{1-soliton-expr}) recovers the algebraic soliton: 
\begin{equation}
\gamma = \pi : \quad u = \frac{2 e^{-{\rm i}t}}{1 + 2{\rm i}x}, \quad 
v = \frac{2 e^{-{\rm i}t}}{1-2{\rm i}x}.
\label{1-soliton-alg-expr}
\end{equation}
which can also be obtained directly from Theorem \ref{theorem-Wronskian} in the case  $N = 1$. 
\end{remark}

Our second result presents a hierarchy of rational solutions to the MTM, 
which generalizes the algebraic soliton (\ref{1-soliton-alg-expr}) for any $N \in \mathbb{N}$. In the case $N = 2$, this hierarchy includes the second-order 
rational solution which corresponds to the double algebraic soliton obtained in \cite{Han-2024} and \cite{LiPelin2025}.

To present the rational solutions, we denote the $(2N)\times(2N)$  identity matrix by $I$ and the $(2N)\times(2N)$  nilpotent matrix by $L$, 
\begin{equation}
\label{matrix-L}
L = \left( \begin{matrix} 0 & 0 & 0 & \dots & 0 & 0 \\
1 & 0 & 0 & \dots & 0 & 0\\
0 & 1 & 0 &  \dots & 0 & 0 \\
\vdots & \vdots &  \vdots & \ddots & \vdots &  \vdots \\
0 & 0 & 0 & \dots & 0 & 0 \\
0 & 0 & 0 & \dots & 1 & 0 
\end{matrix} \right).
\end{equation}
The $j$-th power of $L$ has ones at the $j$-th lower diagonal for $j = 1,2,\dots,2N-1$ such that $L^{2N} = 0$ is the $(2N)\times(2N)$ zero matrix.

The following theorem defines the hierarchy of rational solutions to the MTM. 

\begin{theorem}
	\label{theorem-rat-sol}
	Let matrix $A$ and $A^{-1}$ be defined by 
	\begin{equation}
	\label{choice-A}
	A = -I + L, \qquad A^{-1} = -I - L - L^2 - \dots - L^{2N-1}.
	\end{equation} 
Let $\phi \in \mathbb{C}^{2N}$ be defined by a general solution 
to the left system (\ref{eigen}) with $(2N)$ complex coefficients and $\psi = S \bar{\phi}$. The double-Wronskian solutions (\ref{Wronskian}) generate the rational solutions to the MTM system (\ref{MTM}) in the form:
	\begin{equation}
	\label{rat-sol}
	u(x,t) = \frac{Q_N(x,t)}{\bar{P}_N(x,t)} e^{-{\rm i}t}, \quad 
	v(u,x) = \frac{R_N(x,t)}{P_N(x,t)} e^{-{\rm i}t}, 
	\end{equation}
	where $P_N$ is a polynomial of degree $N^2$ in $x$ and $Q_N$, $R_N$ are polynomials of degree $N^2 - 1$ in $x$. The solution (\ref{rat-sol}) depends on $(2N)$ real parameters and is bounded for all $(x,t) \in \mathbb{R}^2$. 
\end{theorem}

\begin{remark}
	In the proof of Theorem \ref{theorem-rat-sol}, we show that the asymptotic behavior of the solution (\ref{rat-sol}) is given by 
	\begin{equation}
	\label{rat-sol-asympt}
	u(x,t) \sim -\frac{{\rm i}N}{x} e^{-{\rm i} t}, \quad v(x,t) \sim \frac{{\rm i}N}{x} e^{-{\rm i}t}, \quad \mbox{\rm as} \;\; |x| \to \infty,
	\end{equation}
	see (\ref{P_N}), (\ref{Q_N}), and (\ref{R_N}).
	The behavior (\ref{rat-sol-asympt}) is in agreement with \cite[Lemma 6.4]{KPR06} about the existence of a pair of embedded eigenvalues 
	$\zeta = \pm {\rm i}$ in the linear system (\ref{Lax-pair}) of algebraic multiplicity $N$, for which $\lim\limits_{|x| \to \infty} |x| |u(x,t)| =  \lim\limits_{|x| \to \infty} |x| |v(x,t)| > N - \frac{1}{2}$ is required. 
\end{remark}

Our third and final result is about the relation of the higher-order rational solution in Theorem \ref{theorem-rat-sol} to the dynamics of $N$ copies of  algebraic solitons (\ref{1-soliton-alg-expr})  with the slow scattering on the time scale $\mathcal{O}(\sqrt{|t|})$ as $|t| \to \infty$. We show in Lemma \ref{lem-prinicipal-polyn} that the principal part of the polynomial $P_N(x,t)$ in (\ref{rat-sol}) can be written in the form
\begin{equation}
\label{polyno-principal}
p_N(x,t) = a_N^{(0)} x^{N^2} + a_N^{(1)} x^{N^2-4} t^2 + a_N^{(2)} x^{N^2-8} t^4 + 
\ldots + a_N^{(J)} x^{N^2-4J} t^{2J},
\end{equation}
where $J$ is the largest integer such that $N^2 - 4J \geq 0$ 
and $\{ a_N^{(j)}\}_{j=0}^J$ are real-valued coefficients computed from the explicit determinants, see (\ref{p-double}) below. It follows that $J = \frac{N^2}{4}$ if $N$ is even and $J = \frac{N^2-1}{4}$ if $N$ is odd. 
If $x = \upsilon \sqrt{|t|}$, then $p_N(x,t) = |t|^{\frac{N^2}{2}} \hat{p}_N(\upsilon)$, where 
\begin{equation}
\label{polyno-upsilon}
\hat{p}_N(\upsilon) = a_N^{(0)} \upsilon^{N^2} + a_N^{(1)} \upsilon^{N^2-4} + a_N^{(2)} \upsilon^{N^2-8}  + 
\ldots + a_N^{(J)} \upsilon^{N^2-4J}.
\end{equation}
The following theorem describes the corresponding result. 

\begin{theorem}
	\label{theorem-polyn}
Assume that $\hat{p}_N$ in (\ref{polyno-upsilon}) admits exactly $N$ real roots. Then, $P_N$ in (\ref{rat-sol}) admits $\frac{N(N-1)}{2}$ roots in the upper half-plane of $x$ and $\frac{N(N+1)}{2}$ roots in the lower half-plane of $x$ for large $|t|$ and the mass integral is quantized as 
\begin{equation}
\label{mass-quantization}
\mathcal{M}_N(u,v) = \int_{\mathbb{R}} \frac{|Q_N(x,t)|^2 + |R_N(x,t)|^2}{|P_N(x,t)|^2} dx = 4 \pi N.
\end{equation}
\end{theorem}

\begin{remark}
	The real roots of the real-valued polynomial $\hat{p}_N$ in the assumption of Theorem \ref{theorem-polyn}	determine the slow dynamics of individual algebraic solitons, positions of which in $x$ change with the slow time scale $\mathcal{O}(\sqrt{|t|})$ as $|t| \to \infty$. The mass quantization rule (\ref{mass-quantization}) in the conclusion of Theorem \ref{theorem-polyn} gives the total mass of $N$ identical algebraic solitons. The number $(2N)$ of real parameters proven in Theorem \ref{theorem-rat-sol} suggests that the arbitrary parameters of the rational solutions correspond to translations of $N$ individual algebraic solitons in the space-time. 
\end{remark}

\begin{remark}
	The only conjecture left open in this work is the assumptions of $N$ real roots of the polynomial $\hat{p}_N$ defined by (\ref{polyno-upsilon}). The conjecture has been checked numerically for $N = 1,\dots,6$. All other conclusions in Theorems \ref{theorem-rat-sol} and \ref{theorem-polyn} are obtained by the analysis of the double-Wronskian solutions of Theorem \ref{theorem-Wronskian} for the matrix $A$ defined in (\ref{choice-A}).
\end{remark}

\subsection{Comparison with previous works}

Rigorous study of rational solutions of integrable systems started with 
the series of works \cite{YY1,YY2,YY3}, where fundamental patterns of rogue waves 
were constructed by using roots of the limiting polynomials related to 
Yablonskii–-Vorob\'ev and Okamoto polynomial hierarchies (see also the book \cite{YY}). The former hierarchy arises for the Zakharov--Shabat spectral problem 
(e.g., for the NLS equation) and the latter one arises for the $3 \times 3$ spectral problem (e.g., for the Manakov system of the coupled NLS equations). 
Rogue waves are algebraically decaying in both space and time variables and their 
existence is related to the modulation instability of the background. 

In the context of the MTM, rogue waves and associated rational solutions 
on the nonzero background were constructed and analyzed in the previous works 
\cite{Degasperis}, \cite{He}, \cite{Ye}, and \cite{Feng2}. However, analysis of algebraic solitons on zero background is much harder for the Kaup--Newell spectral problem, and it has been an open problem for many years. Particular second-order rational solutions for the double algebraic solitons were obtained in \cite{Guo-SAPM-2013,Xu-JPA-2011} for the derivative NLS equation and 
in \cite{Han-2024,LiPelin2025} for the MTM system. A more systematic approach 
on constructing a hierarchy of rational solutions for the derivative NLS equations
and their close relatives was developed recently in \cite{Wang1,Wang2} but this work is the first one, where these solutions are rigorously analyzed in the context of the particular MTM system.  

Among further developments, which can be implied by our work, is the proof 
of the assumption in Theorem \ref{theorem-polyn} that the polynomial 
$\hat{p}_N$ given by (\ref{polyno-upsilon}) admits exactly $N$ real roots. This might be related to the special polynomial hierarchies for the Kaup--Newell spectral problem, which has not been analyzed in \cite{YY1,YY2,YY3} or \cite{Feng2}.

Another interesting direction is to consider the rational solutions in 
Theorem \ref{theorem-rat-sol} in the limit of $N \to \infty$ to 
study the limiting universal pattern of the rational solutions in the MTM system. 
The picture is likely to be very different from what has been studied in integrable equations related to the Zakharov--Shabat spectral problem, e.g. 
in \cite{bilman1,bilman2,bilman3,bilman4}.

\subsection{Methodology and organization of the paper}

Section \ref{sec-2} contains the proof of Theorem \ref{theorem-Wronskian}. We adopt the construction of \cite{Wang1,Wang2} of the double-Wronskian solutions 
and we verify all bilinear equations directly by using the fundamental properties of determinants in Lemmas \ref{lem-1} and \ref{lem-2}. We give examples of how 
the double-Wronskian solutions recover the exponential and algebraic solitons. We also show the relation between the double-Wronskian solutions of the MTM system and the double-Wronskian solutions of the two-component KP hierarchy. 

Section \ref{sec-3} gives the proof of Theorem \ref{theorem-rat-sol} with Lemmas \ref{lem-S}, \ref{lem-param}, \ref{lem-leading-order}, and \ref{lem-bounded}. Although 
the construction of the rational solutions is straightforward by taking the 
matrix $A$ in the Jordan block form (\ref{choice-A}) for eigenvalues $\zeta = \pm {\rm i}$ of algebraic multiplicity $N$, we rigorously prove in Lemma \ref{lem-leading-order} that the coefficients at the highest powers of the coresponding polynomials $P_N$, $Q_N$, and $R_N$ given by certain numerical determinants are nonzero. The key computations rely on an inductive method which reduces the numerical determinants to a factorized form by using successive two-column eliminations, described in Appendices \ref{app-A} and \ref{app-B}.

Section \ref{sec-4} presents the proof of Theorem \ref{theorem-polyn} with Lemmas \ref{lemma-polyn} and \ref{lemma-mass}. The corresponding analysis is based on the leading-order representation 
of the polynomial $P_N$ in the form 
\begin{equation}
	\label{polynom-dom}
P_N(x,t) = a_N^{(0)} \left( x + \frac{{\rm i}}{2} \right)^{N^2} + \mathcal{O}(x^{N^2-2}),
\end{equation}
which is proven in Appendix \ref{app-C} with a modification of the two-column elimination algorithm. The coefficient $a_N^{(0)}$ in (\ref{polyno-principal}), (\ref{polyno-upsilon}), and (\ref{polynom-dom}) is the same. This allows us to relate the assumption of Theorem \ref{theorem-polyn} on the real roots of $\hat{p}_N$ to the numbers of complex roots 
of $P_N(\cdot,t)$ in $\mathbb{C}^+$ and $\mathbb{C}^-$ for large $|t|$. From this information and the conservation of the mass $M(u,v)$ in time $t \in \mathbb{R}$, the quantization formula (\ref{mass-quantization}) is obtained by using an argument principle, similar to the case $N = 2$ considered in \cite{Han-2024}.

Section \ref{sec-5} presents examples of the rational solutions of Theorem \ref{theorem-rat-sol} for $N = 1,\dots,6$, from which the assumptions of Theorem \ref{theorem-polyn} are verified. We also display 
the slow scattering dynamics of $N$ identical algebraic solitons by 
using the solution surfaces computed from (\ref{rat-sol}) and by comparing the dynamics with roots of $p_N$ given by (\ref{polyno-principal}).

\section{Proof of Theorem \ref{theorem-Wronskian}}
\label{sec-2}

Let $A \in \mathbb{C}^{2N \times 2N}$ be an invertible matrix for a fixed $N \in \mathbb{N}$. We define two vectors $\phi, \psi \in \mathbb{C}^{2N}$ from solutions of the linear equations (\ref{eigen}). Furthermore, we impose the factorization (\ref{rel-1}) with an invertible matrix $S \in \mathbb{C}^{2N \times 2N}$ and the relation (\ref{rel-2}), namely 
$\psi = S \bar{\phi}$. 

For the proof of Theorem \ref{theorem-Wronskian}, we recall the following two lemmas. Lemma \ref{lem-1} is equivalent to Liouville's theorem for a system of linear differential equations. Lemma \ref{lem-2} is equivalent to the Pl\"ucker relation for determinants. 

\begin{lemma}
	\label{lem-1} Let $A \in \mathbb{C}^{n \times n}$ and $\{ x_1, x_2, \dots, x_n \} \in \mathbb{C}^n$ for some $n \in \mathbb{N}$. Then 
	\begin{equation}
	\label{trace-identity}
	{\rm tr}(A) |x_1,x_2,\dots,x_n| = |A x_1, x_2, \dots, x_n| + 
	 |x_1, A x_2, \dots, x_n| + \dots + |x_1,x_2,\dots, Ax_n|.
	\end{equation}
\end{lemma}

\begin{lemma}
	\label{lem-2} Let $M \in \mathbb{C}^{n \times n-2}$ and $\{ a,b,c,d \} \in \mathbb{C}^n$ for some $n \in \mathbb{N}$. The Pl\"ucker relation is 
	\begin{equation}
	\label{det-identity}
|M,a,b| |M,c,d| - |M,a,c| |M,b,d| + |M,a,d| |M,b,c| = 0.
	\end{equation}
\end{lemma}

By using notations in Remark \ref{rem-tau}, the following lemma establishes the complex-conjugate symmetry for the tau-function $\tau_{n,m}^l$ defined by (\ref{double-Wr}). Since $C = (-{\rm i})^N/|S|$, this verifies the complex-conjugate symmetry of the double-Wronskian solutions (\ref{Wronskian}) rewritten as (\ref{Wronskian-tau}).

\begin{lemma}
 \label{lem-3}    
Let $\tau^l_{n,m}$ be defined by (\ref{double-Wr}) with  $\psi = S \bar{\phi}$. Then, $\tau^l_{n,m}$ and $\tau^{-l}_{m+1,n}$ are related by the complex-conjugate symmetry:
\begin{align}
\tau^l_{n,m} = \frac{{\rm i}^{N-l}}{|S|} \overline{\tau^{-l}_{m+1,n}}\,. 
\label{complex-red}
\end{align}
\end{lemma}

\begin{proof}
    By using (\ref{rel-2}), we have 
\begin{align*}
&\overline{\tau^l_{n,m}} = |\bar{\phi}^{(n)},\bar{\phi}^{(n+1)},\dots,\bar{\phi}^{(n+N+l-1)};\bar{\psi}^{(m)},\bar{\psi}^{(m+1)},\dots,\bar{\psi}^{(m+N-l-1)}|\\
&=(-1)^{N-l} 
|\bar{\psi}^{(m)},\bar{\psi}^{(m+1)},\dots,\bar{\psi}^{(m+N-l-1)}; \bar{\phi}^{(n)},\bar{\phi}^{(n+1)},\dots,\bar{\phi}^{(n+N+l-1)}| \\
&= (-1)^{N-l} 
|\bar{S}{\phi}^{(m)},\bar{S}{\phi}^{(m+1)},\dots,\bar{S}{\phi}^{(m+N-l-1)}; \bar{\phi}^{(n)},\bar{\phi}^{(n+1)},\dots,\bar{\phi}^{(n+N+l-1)}|\\
&=\frac{(-1)^{N-l}}{|S|} 
|-A{\phi}^{(m)},-A{\phi}^{(m+1)},\dots,-A{\phi}^{(m+N-l-1)}; S\bar{\phi}^{(n)},S\bar{\phi}^{(n+1)},\dots,
S\bar{\phi}^{(n+N+l-1)}| \\
&=\frac{(-{\rm i})^{N-l}}{|S|} 
|{\phi}^{(m+1)},{\phi}^{(m+2)},\dots,{\phi}^{(m+N-l)}; {\psi}^{(n)},{\psi}^{(n+1)},\dots,
{\psi}^{(n+N+l-1)}|\\
&= \frac{(-{\rm i})^{N-l}}{|S|}\tau^{-l}_{m+1,n}\,,
\end{align*}
which completes the proof of (\ref{complex-red}).
\end{proof}

It remains to prove that the double-Wronskian solutions (\ref{Wronskian}) satisfy the bilinear equations (\ref{MTM-4}). These four equations are verified next, where we recall that the prime stands for the derivative of vectors $\phi$ and $\psi$ with respect to $\xi$.

\vspace{0.25cm}

\underline{Validity of ${\rm i} D_\eta(\bar{f} \cdot f) - 2g\bar{g}=0$.} 

By using expression for $f$ and $\bar{f}$ in (\ref{Wronskian}), we get 
\begin{align*}
{\rm i} D_\eta(\bar{f} \cdot f) &= {\rm i} (\bar{f}_{\eta} f - \bar{f} f_{\eta}) \\
&= {\rm i} C |\widetilde{N} ; \widetilde{N}| \left( |0,\overline{N}; \widehat{N-1}| + |\widetilde{N};-1,\widetilde{N-1}| \right) - {\rm i} C |\widetilde{N} ; \widehat{N-1}| \left( |0,\overline{N}; \widetilde{N}| + |\widetilde{N};0,\overline{N}| \right) \\
&= 2{\rm i} C |\widetilde{N} ; \widetilde{N}| |0,\overline{N}; \widehat{N-1}| - 2 {\rm i} C |\widetilde{N} ; \widehat{N-1}|  |0,\overline{N}; \widetilde{N}|. 
\end{align*}
To get the second equality, we have used 
\begin{align*}
{\rm tr}(A^{-1})  |\widetilde{N} ; \widehat{N-1}| &= {\rm i} \left( |0,\overline{N}; \widehat{N-1}| - |\widetilde{N};-1,\widetilde{N-1}| \right), \\
{\rm tr}(A^{-1}) |\widetilde{N} ; \widetilde{N}| &= {\rm i} \left( |0,\overline{N}; \widetilde{N}| - |\widetilde{N};0,\overline{N}| \right),
\end{align*}
which follow from the identity (\ref{trace-identity}) in Lemma \ref{lem-1} 
with $A^{-1}$. Combining with $-2g\bar{g}$ from (\ref{Wronskian}), we get 
\begin{align*}
{\rm i} D_\eta(\bar{f} \cdot f)- 2g\bar{g} = 2 {\rm i} C \left( 
 |\widetilde{N} ; \widetilde{N}| |0,\overline{N}; \widehat{N-1}| - |\widetilde{N} ; \widehat{N-1}|  |0,\overline{N}; \widetilde{N}| - 
|\widehat{N} ; \widetilde{N-1}| |\overline{N} ; \widehat{N}| \right).
\end{align*}
To show that the expression in brackets is identically zero, we use identity (\ref{det-identity}) of Lemma \ref{lem-2} with $M := |\overline{N};\widetilde{N-1}|$, $a = \phi'$ in the first column, 
$b = \psi^{(N)}$ in the last column, $c = \phi$ in the first column, and $d = \psi$ in the $(N+1)$-th column. The identity (\ref{det-identity}) holds after rearrangement of the columns since the order of vector $a,b,c,d$ appear to be the same in each determinant.

\vspace{0.25cm}

\underline{Validity of ${\rm i} D_\xi (g \cdot f)+2h \bar{f}=0$.}\\
By using expression for $g$ and $f$ in \eqref{Wronskian}, we have
\begin{align*}
{\rm i} D_\xi(g\cdot f)=&{\rm i}(g_\xi f-gf_\xi)\\
=& {\rm i}|\widetilde{N} ; \widehat{N-1}|\left(|\widehat{N-1},N+1;\widetilde{N-1}|+|\widehat{N};\widetilde{N-2},N|\right)\\
&-{\rm i}|\widehat{N} ; \widetilde{N-1}|\left(|\widetilde{N-1},N+1;\widehat{N-1}|+|\widetilde{N};\widehat{N-2},N|\right)\\
=&2{\rm i} \left(|\widetilde{N};\widehat{N-1}| |\widehat{N};\widetilde{N-2},N| - |\widehat{N};\widetilde{N-1}| |\widetilde{N};\widehat{N-2},N| \right).
\end{align*}
To get the second equality, we have used 
\begin{align*}
{\rm tr}(A)  |\widehat{N} ; \widetilde{N-1}| &=- {\rm i} \left( |\widehat{N-1},N+1;\widetilde{N-1}|-|\widehat{N};\widetilde{N-2},N| \right), \\
{\rm tr}(A) |\widetilde{N} ; \widehat{N-1}| &=- {\rm i} \left( |\widetilde{N-1},N+1;\widehat{N-1}|-|\widetilde{N};\widehat{N-2},N| \right),
\end{align*}
which follows from the identity (\ref{trace-identity}). 
Together with $2h\bar{f}$, we have
\begin{align*}
& {\rm i} D_\xi (g \cdot f)+2h \bar{f} \\
& \qquad =2\mathrm{i}\left(|\widetilde{N};\widehat{N-1}| |\widehat{N};\widetilde{N-2},N| - |\widetilde{N};\widehat{N-2},N| |\widehat{N};\widetilde{N-1}| + |\widehat{N};\widehat{N-2}| |\widetilde{N};\widetilde{N}|  \right).
\end{align*}
To show that the expression in brackets is identically zero, we use identity (\ref{det-identity}) with $M:=(\widetilde{N};\widetilde{N-2})$, $a=\psi$ in the $(N+1)$-th column, $b=\psi^{(N-1)}$ in the last column, $c=\phi$ in the first column, and $d=\psi^{(N)}$ in the last column.

\vspace{0.25cm}

\underline{Validity of $\mathrm{i} D_\eta(h \cdot \bar{f})+2g f=0$.}\\
By using expression for $h$ and $\bar{f}$ in \eqref{Wronskian}, we obtain
\begin{align*}
\mathrm{i}D_\eta(h\cdot \bar{f})&=\mathrm{i}(h_\eta\bar{f}-h\bar{f}_\eta)\\
&=|\widetilde{N} ; \widetilde{N}|(|-1, \widetilde{N} ; \widehat{N-2}|+|\widehat{N} ;-1, \widetilde{N-2}|)-|\widehat{N} ; \widehat{N-2}|(|0, \overline{N} ; \widetilde{N}|+|\widetilde{N} ; 0, \overline{N}|)\\
&=2|\widetilde{N} ; \widetilde{N}||\widehat{N} ;-1, \widetilde{N-2}|-2|\widehat{N} ; \widehat{N-2}||\widetilde{N} ; 0, \overline{N}|
\end{align*}
To get the second equality, we have used 
\begin{align*}
{\rm tr}(A^{-1})  |\widehat{N} ; \widehat{N-2}| &= {\rm i} \left( |-1, \widetilde{N} ; \widehat{N-2}|-|\widehat{N} ;-1, \widetilde{N-2}|)\right), \\
{\rm tr}(A^{-1}) |\widetilde{N} ; \widetilde{N}| &= {\rm i} \left( |0,\overline{N}; \widetilde{N}| - |\widetilde{N};0,\overline{N}| \right),
\end{align*}
which follow from the identity (\ref{trace-identity}) with $A^{-1}$.
Combining with $2gf$, we get
\begin{align*}
\mathrm{i} D_\eta(h\cdot \bar{f})+2gf = 2 \left(|\widehat{N} ;-1, \widetilde{N-2}||\widetilde{N} ; \widetilde{N}| - |\widehat{N} ; \widehat{N-2}| |\widetilde{N} ; 0, \overline{N}| + |\widetilde{N} ; \widehat{N-1}| |\widehat{N} ; \widetilde{N-1}|\right).
\end{align*}
To show that the expression in brackets is identically zero, we 
can not use identity (\ref{det-identity}) directly. However, we can write 
\begin{align*}
|\widehat{N} ;-1, \widetilde{N-2}| &= |\partial_{\xi}^{-1} \phi', \partial_{\xi}^{-1} \phi'',\dots,\partial_{\xi}^{-1} \phi^{(N+1)};\partial_{\xi}^{-1} \psi, \partial_{\xi}^{-1} \psi'',\dots,\partial_{\xi}^{-1} \psi^{(N-1)}| \\
&= |-\partial_{\eta} \phi', -\partial_{\eta} \phi'',\dots,-\partial_{\eta} \phi^{(N+1)};-\partial_{\eta} \psi, -\partial_{\eta} \psi'',\dots,-\partial_{\eta} \psi^{(N-1)}| \\
&= |-{\rm i} A^{-1} \phi', -{\rm i} A^{-1} \phi'',\dots,-{\rm i} A^{-1} \phi^{(N+1)};{\rm i} A^{-1} \psi, {\rm i} A^{-1} \psi'',\dots,{\rm i} A^{-1} \psi^{(N-1)}| \\
&= (-{\rm i})^{N+1} {\rm i}^{N-1} |A^{-1}| |\phi', \phi'',\dots,\phi^{(N+1)};\psi,\psi'',\dots,\psi^{(N-1)}|  \\
&= -|A^{-1}| |\widetilde{N+1};0,\overline{N-1}|
\end{align*}
and similarly,
\begin{align*}
|\widehat{N} ; \widehat{N-2}| &= -|A^{-1}| |\widetilde{N+1} ; \widetilde{N-1}|, \\
|\widehat{N} ; \widetilde{N-1}| &= -|A^{-1}| |\widetilde{N+1} ; \overline{N}|.
\end{align*}
Hence, we rewrite the formula in the equivalent way:
\begin{align*}
& \qquad \mathrm{i} D_\eta(h\cdot \bar{f})+2gf \\
& = -2 |A^{-1}|
\left( |\widetilde{N+1} ; 0, \overline{N-1}||\widetilde{N} ; \widetilde{N}|-|\widetilde{N+1} ; \widetilde{N-1}||\widetilde{N} ; 0, \overline{N}|+|\widetilde{N+1} ; \overline{N}||\widetilde{N} ; \widehat{N-1}| \right).
\end{align*}
We can now use identity (\ref{det-identity}) with $M:=(\widetilde{N},\overline{N-1})$, $a=\phi^{(N+1)}$ in the $(N+1)$-th column, $b=\psi$ in the $(N+2)$-th column, $c=\psi^{\prime}$ in the $(N+1)$-th column, and  $d=\psi^{(N)}$ in the last column. This yields zero in the brackets. 

\vspace{0.25cm}

\underline{Validity of $\mathrm{i} D_\xi (f \cdot \bar{f})-2 h\bar{h}=0$.}\\
By using expression for $f$ and $\bar{f}$ in \eqref{Wronskian}, we find
\begin{align*}
\mathrm{i} D_\xi (f \cdot \bar{f})=&\mathrm{i}(f_\xi\bar{f}-f\bar{f}_\xi)\\
=&\mathrm{i}C|\widetilde{N} ; \widetilde{N}|\left(|\widetilde{N-1},N+1;\widehat{N-1}|+|\widetilde{N};\widehat{N-2},N|\right)\\
&-\mathrm{i}C|\widetilde{N};\widehat{N-1}|\left(|\widetilde{N-1},N+1;\widetilde{N}|+|\widetilde{N},\widetilde{N-1},N+1|\right)\\
=&2\mathrm{i}C\left(|\widetilde{N};\widetilde{N}||\widetilde{N-1},N+1;\widehat{N-1}|-|\widetilde{N};\widehat{N-1}||\widetilde{N-1},N+1;\widetilde{N}|\right).
\end{align*}
To get the second equality, we have used 
\begin{align*}
{\rm tr}(A)  |\widetilde{N} ; \widetilde{N}| &=- \mathrm{i} \left( |\widetilde{N-1},N+1;\widetilde{N}|-|\widetilde{N};\widetilde{N-1},N+1| \right), \\
{\rm tr}(A) |\widetilde{N} ; \widehat{N-1}| &=- \mathrm{i} \left( |\widetilde{N-1},N+1;\widehat{N-1}|-|\widetilde{N};\widehat{N-2},N| \right),
\end{align*}
which follow from the identity (\ref{trace-identity}). 
Together with the term $-2h\bar{h}$, we have
\begin{align*}
\mathrm{i} D_\xi (f \cdot \bar{f})-2 h\bar{h} &= 2\mathrm{i} C\left(|\widetilde{N};\widetilde{N}||\widetilde{N-1},N+1;\widehat{N-1}|-|\widetilde{N};\widehat{N-1}||\widetilde{N-1},N+1;\widetilde{N}|\right) \\
& \quad -2\mathrm{i}\bar{C}^{-1}|\widehat{N};\widehat{N-2}||\widetilde{N-1};\widehat{N}|
\end{align*}
In order to use the identity (\ref{det-identity}), we need to rewrite the last term in the equivalent way. Since 
\begin{align*}
|\widehat{N};\widehat{N-2}|  &= |\phi, \phi', \dots, \phi^{(N)};\psi, \psi', \dots, \psi^{(N-2)}| \\
&= |A^{-1}| (-{\rm i})^{N+1} {\rm i}^{N-1} |{\rm i}A \phi, {\rm i}A\phi', \dots, {\rm i}A \phi^{(N)}; -iA\psi, -{\rm i}A \psi', \dots, -{\rm i}A \psi^{(N-2)}| \\
&= - |A^{-1}| |\widetilde{N+1};\widetilde{N-1}|,
\end{align*}
we use $|A^{-1}| = (|S| |\bar{S}|)^{-1}$ and $C = (-{\rm i})^N/|S|$ to rewrite 
\begin{align*}
& \mathrm{i} D_\xi (f \cdot \bar{f})-2 h\bar{h} \\ &  = 2\mathrm{i} C\left(|\widetilde{N};\widetilde{N}||\widetilde{N-1},N+1;\widehat{N-1}|-|\widetilde{N};\widehat{N-1}||\widetilde{N-1},N+1;\widetilde{N}| + |\widetilde{N+1};\widetilde{N-1}||\widetilde{N-1};\widehat{N}| \right).
\end{align*}
We can now use identity (\ref{det-identity}) with  $M:=(\widetilde{N-1};\widetilde{N-1})$, $a:=\phi^{(N)}$ in the $N$-th column, $b:=\psi^{(N)}$ in the last column, $c:=\phi^{(N+1)}$ in the $N$-th column, and $d:=\psi$ in the $(N+1)$-th column. This yields zero in the brackets. 

\vspace{0.25cm}

All four bilinear equations (\ref{MTM-4}) are satisfied. The proof of Theorem \ref{theorem-Wronskian} is complete. 

\subsection{Examples of solitons via double-Wronskian solutions}

For $N = 1$, we define in (\ref{rel-1}):
\begin{equation}
\label{1-exp}
        A =
        \begin{bmatrix}
            e^{{\rm i} \gamma} & 0\\
            0 & e^{-{\rm i}\gamma}
        \end{bmatrix}
        \quad \rm{and} \quad
        S = 
        \begin{bmatrix}
            0 & e^{\frac{{\rm i} \gamma}{2}}\\
            -e^{\frac{-{\rm i} \gamma}{2}} & 0
        \end{bmatrix},
\end{equation}
where $\gamma \in (0,\pi)$ is an arbitrary parameter. Using (\ref{eigen}) and (\ref{rel-2}), we obtain
\begin{equation}
\label{1-exp-eigen}
        \phi =  \begin{bmatrix}
        c_1 e^{{\rm i} e^{{\rm i}\gamma} \xi + {\rm i} e^{-{\rm i}\gamma} \eta} \\
        c_2 e^{{\rm i} e^{-{\rm i}\gamma} \xi + {\rm i} e^{{\rm i}\gamma} \eta}
\end{bmatrix}
 \quad \rm{and} \quad
         \psi =  \begin{bmatrix}
 \bar{c}_2 e^{\frac{{\rm i}\gamma}{2} - {\rm i} e^{{\rm i} \gamma} \xi -{\rm i} e^{-{\rm i}\gamma} \eta} \\
 - \bar{c}_1 e^{-\frac{{\rm i}\gamma}{2} - {\rm i} e^{-{\rm i}\gamma} \xi - {\rm i} e^{{\rm i}\gamma} \eta},
 \end{bmatrix}
\end{equation}
where $c_1$ and $c_2$ are arbitrary complex coefficients.

\begin{remark}
If $(u,v) = (0,0)$, solutions of the linear system (\ref{Lax-1})--(\ref{Lax-2}) are given explicitly as 
	\begin{equation}
	\label{sol-Phi}
	\vec{\phi} = e^{-\mathrm{i}(\zeta^2\xi+\zeta^{-2}\eta) \sigma_3} \vec{d},
	\end{equation}
	where $\vec{d} = (d_1,d_2) \in \mathbb{C}^2$ contain two arbitrary complex coefficients. By writing $\vec{\phi} = (\psi_0,\phi_0)^T$, we obtain
	\begin{equation}
	\label{sol-Phi-component}
	\psi_0 = e^{-\mathrm{i}(\zeta^2\xi+\zeta^{-2}\eta)} d_1, \qquad 
	\phi_0 = e^{\mathrm{i}(\zeta^2\xi+\zeta^{-2}\eta)} d_2.
	\end{equation}
If $\zeta \in \mathbb{C}$ is an eigenvalue in the first quadrant of the complex plane, so are $\bar{\zeta}$, $-\zeta$, and $-\bar{\zeta}$ in the other three quadrants. Hence, we also obtain another solution of the linear system (\ref{Lax-1})--(\ref{Lax-2}) with $(u,v) = (0,0)$:
	\begin{equation}
	\label{sol-Phi-component-new}
	\tilde{\psi}_0 = e^{-\mathrm{i}(\bar{\zeta}^2\xi+\bar{\zeta}^{-2}\eta)} \tilde{d}_1, \qquad 
	\tilde{\phi}_0 = e^{\mathrm{i}(\bar{\zeta}^2\xi+\bar{\zeta}^{-2}\eta)} \tilde{d}_2,
	\end{equation}
where $\tilde{d}_1$ and $\tilde{d}_2$ are also two arbitrary complex coefficients. 
The solution (\ref{1-exp-eigen}) agrees with the solutions (\ref{sol-Phi-component}) and (\ref{sol-Phi-component-new}) for $\zeta = e^{\frac{{\rm i} \gamma}{2}}$.
\end{remark}

Substituting (\ref{1-exp-eigen}) into (\ref{Wronskian}) with 
$N = 1$ and $C = -{\rm i}/|S| = -{\rm i}$, we get
\begin{equation*}
\left\{ \begin{array}{l} 
f = |\phi';\psi|,  \\
g = |\phi,\phi'|, \\
h =  -|\phi,\phi'|, 
\end{array} \right. \quad \Rightarrow \quad 
\left\{ \begin{array}{l} 
f = -{\rm i}\left( |c_1|^2 e^{\frac{{\rm i}\gamma}{2} -2 (\xi - \eta) \sin\gamma} + |c_2|^2 e^{-\frac{{\rm i}\gamma}{2} +2 (\xi - \eta) \sin\gamma} \right), \\
g = 2c_1 c_2 \sin \gamma e^{2{\rm i} (\xi + \eta) \cos \gamma}, \\
h = -2 c_1 c_2 \sin \gamma e^{2{\rm i} (\xi + \eta) \cos \gamma}.
\end{array} \right.
\end{equation*}
This solution with $c_1 = 1$ and $c_2 = {\rm i}$ generates the exponential soliton in the form (\ref{1-soliton-expr}) by virtue of (\ref{MTM-bil}) and (\ref{transformation}).

Next, we recover the algebraic soliton in the form  (\ref{1-soliton-alg-expr}) 
without taking the limit $\gamma \to \pi$ in (\ref{1-soliton-expr}). For $N = 1$, we define in (\ref{rel-1}):
\begin{equation}
A =
\begin{bmatrix}
-1 & 0\\
1 & -1
\end{bmatrix}
\quad \rm{and} \quad
S = 
\begin{bmatrix}
1 & 0 \\
-\frac{1}{2} & 1
\end{bmatrix}.
\label{Jordan-block-2}
\end{equation}
Using (\ref{eigen}) and (\ref{rel-2}), we obtain
\begin{equation*}
\phi =  c_1 \begin{bmatrix} 1 \\ {\rm i}(\xi-\eta) \end{bmatrix} e^{-{\rm i} (\xi + \eta)} 
+ c_2 \begin{bmatrix} 0 \\ 1 \end{bmatrix} e^{-{\rm i} (\xi + \eta)} 
\;\; \rm{and} \;\;
\psi = \bar{c}_1 \begin{bmatrix} 1 \\ - {\rm i}(\xi-\eta) - \frac{1}{2} \end{bmatrix}  e^{{\rm i} (\xi + \eta)} + \bar{c}_2 \begin{bmatrix} 0 \\ 1 \end{bmatrix}  e^{{\rm i} (\xi + \eta)}.
\end{equation*}
Substituting these expression into (\ref{Wronskian}) with $N = 1$ and $C = -{\rm i}/|S| = -{\rm i}$, we obtain
\begin{equation*}
\left\{ \begin{array}{l} 
f = |\phi';\psi|, \\
g = |\phi,\phi'|, \\
h = -|\phi,\phi'|, 
\end{array} \right. \quad \Rightarrow \quad 
\left\{ \begin{array}{l} 
f = {\rm i} (\bar{c}_1 c_2 - c_1 \bar{c}_2) - 2 (\xi - \eta) |c_1|^2 - \frac{{\rm i}}{2} |c_1|^2 , \\
g = {\rm i} c_1^2 e^{-2{\rm i}(\xi + \eta)}, \\
h = -{\rm i} c_1^2 e^{-2{\rm i}(\xi + \eta)}.
\end{array} \right.
\end{equation*}
This solution with $c_1 = 1$ and $c_2 = 0$ generates the algebraic soliton in the form (\ref{1-soliton-alg-expr}) by virtue of (\ref{MTM-bil}) and (\ref{transformation}). 

\begin{remark}
The solution for $\phi$ and $\psi$ obtained from (\ref{Jordan-block-2}) agrees with the solution (\ref{sol-Phi-component}) and its derivative with respect to $\zeta^2$ evaluated at $\zeta = \pm {\rm i}$. This corresponds to the fact that the algebraic soliton is related to the embedded eigenvalues 
$\zeta = \pm {\rm i}$ of the linear system (\ref{Lax-pair}), see \cite{KPR06,LiPelin2025}. 
To recover the algebraic soliton directly, we take the $2 \times 2$ matrix $A$ as the lower triangular Jordan block for eigenvalue $\zeta^2 = -1$.
\end{remark}

\subsection{Relation to double-Wronskian solutions of the two-component KP hierarchy}

We introduce tau-functions of the two-component KP hierarchy by using the double-Wronskian determinants in the form (\ref{double-Wr}), where $\phi=\phi(x_1,x_{-1})$, $\psi=\psi(y_1,y_{-1})$, with $\phi^{(n)}$ and $\psi^{(m)}$ representing the $n$-th and $m$-th derivatives respective to $x_1$ and $y_1$, respectively. We impose the following relations:
\begin{align*}
\partial_{x_{1}} \phi^{(n)}=\phi^{(n+1)},\qquad \partial_{y_{1}}\psi^{(m)}=\psi^{(m+1)}  
\end{align*}
and 
\begin{align*}
\partial_{x_{-1}} \phi^{(n)}=\phi^{(n-1)},\qquad \partial_{y_{-1}}\psi^{(m)}=\psi^{(m-1)}.  
\end{align*}
The double-Wronskian functions (\ref{double-Wr}) were introduced in \cite{HirotaOhtaSatsuma1988,Kakei1988} to represent soliton solutions of many integrable equations. By using the Pl\"ucker  relation in Lemma \ref{lem-2}, we show that the double-Wronskian functions in (\ref{double-Wr}) satisfy the following bilinear equations:
\begin{align}
& D_{x_1}\tau^1_{n,m+1} \cdot \tau^0_{n+1,m} = \tau^1_{n,m} \tau^0_{n+1,m+1}, \label{2KP-bilinear1}\\
& D_{x_{-1}}\tau^1_{n,m} \cdot \tau^0_{n,m} = \tau^1_{n-1,m} \tau^0_{n+1,m}, 
\label{2KP-bilinear2}\\
& D_{y_1}\tau^0_{n+1,m} \cdot \tau^0_{n,m} = \tau^1_{n,m} \tau^{-1}_{n+1,m},  \label{2KP-bilinear3}\\
& D_{y_{-1}}\tau^0_{n+1,m} \cdot \tau^0_{n,m} = -\tau^1_{n,m+1} \tau^{-1}_{n+1,m-1}\,.
\label{2KP-bilinear4}
\end{align}
For brevity, we only prove the last two equations (\ref{2KP-bilinear3}) and (\ref{2KP-bilinear4}). Since
\begin{align}
\tau^0_{n,m} &= |\phi^{(n)},\dots,\phi^{(n+N-1)};\psi^{(m)},\psi^{(m+1)},\dots,\psi^{(m+N-1)}|, \label{tau0} \\
\tau^0_{n+1,m} &= |\phi^{(n+1)},\dots,\phi^{(n+N)};\psi^{(m)},\psi^{(m+1)},\dots,\psi^{(m+N-1)}|, 
\label{tau1}
\end{align}
it follows that
\begin{align}
\partial_{y_1} \tau^0_{n,m} &= |\phi^{(n)},\dots,\phi^{(n+N-1)};\psi^{(m)},\dots,\psi^{(m+N-2)}, \psi^{(m+N)}|,  \label{Dy1a}\\
\partial_{y_1} \tau^0_{n+1,m} &= |\phi^{(n+1)},\dots,\phi^{(n+N)};\psi^{(m)},  \dots,\psi^{(m+N-2)}, \psi^{(m+N)}|, \label{Dy1b} \\
\partial_{y_{-1}} \tau^0_{n,m} &= |\phi^{(n)},\dots,\phi^{(n+N-1)};\psi^{(m-1)}, \psi^{(m+1)}, \dots, \psi^{(m+N-1)}|,  \label{Dy2a}\\
\partial_{y_{-1}} \tau^0_{n+1,m} &= |\phi^{(n+1)},\dots,\phi^{(n+N)};\psi^{(m-1)}, \psi^{(m+1)}, \dots,\psi^{(m+N-1)}|. \label{Dy2b}
\end{align}
Applying the Pl\"ucker  relation to (\ref{tau0})--(\ref{tau1}) and (\ref{Dy1a})--(\ref{Dy1b}), the bilinear equation (\ref{2KP-bilinear3}) is verified. Applying the Pl\"ucker  relation to (\ref{tau0})--(\ref{tau1}) and (\ref{Dy2a})--(\ref{Dy2b}), the bilinear equation (\ref{2KP-bilinear4}) is verified.

For the purpose of getting exponential and algebraic solitons, we take
\begin{align*}
\phi = e^{{\rm i}Ax_1+{\rm i}A^{-1}x_{-1}+x_0} \vec{c}, \quad \psi=e^{{\rm i} By_1+{\rm i}B^{-1}y_{-1}+y_0} \vec{d}.
\end{align*}
where $A, B \in \mathbb{C}^{2N \times 2N}$ are invertible matrices  
and $\vec{c}, \vec{d} \in \mathbb{C}^{2N}$ are constant vectors.
Imposing the condition $B=-A$ yields
\begin{align*}
(\partial_{x_1} -\partial_{y_1})\tau^0_{n,m} = C_1 \tau^0_{n,m}, \quad
(\partial_{x_{-1}} -\partial_{y_{-1}})\tau^0_{n,m} = C_2 \tau^0_{n,m}.
\end{align*}
As a result, the variables $y_1$ and $y_{-1}$ become dummy variables and the bilinear equations (\ref{2KP-bilinear3}) and (\ref{2KP-bilinear4}) become
\begin{align}
D_{x_1}\tau^0_{n+1,m} \cdot \tau^0_{n,m} &= \tau^1_{n,m}\tau^{-1}_{n+1,m}, 
\label{2KP-bilinear5}\\
D_{x_{-1}}\tau^0_{n+1,m} \cdot \tau^0_{n,m} &= -\tau^1_{n,m+1} \tau^{-1}_{n+1,m-1}\,.
\label{2KP-bilinear6}
\end{align}
Then, we define $x_1=2 \xi$, $x_{-1}=-2 \eta$ and modify the definition such that $\phi^{(n)}$ and $\psi^{(m)}$ represent the $n$-th and $m$-th derivatives respective to $\xi$. After some calculations, the four bilinear equations 
(\ref{2KP-bilinear1}), (\ref{2KP-bilinear2}), (\ref{2KP-bilinear5}), and (\ref{2KP-bilinear6}),  become
\begin{align}
\left\{ \begin{array}{l}
D_{\xi}\tau^1_{0,1} \cdot \tau^0_{1,0} = 2\tau^1_{0,0} \tau^0_{1,1}, \\
D_{\eta}\tau^1_{0,0} \cdot \tau^0_{0,0} = -2\tau^1_{-1,0} \tau^0_{1,0}, \\
D_{\xi}\tau^0_{1,0} \cdot \tau^0_{0,0} = 2\tau^1_{0,0} \tau^{-1}_{1,0}, \\
 D_{\eta} \tau^0_{1,0} \cdot \tau^0_{0,0} = 2\tau^1_{0,1} \tau^{-1}_{1,-1}\,. 
\end{array} \right. \label{bl}
\end{align}
by imposing $n=m=0$.
Since $\tau^l_{n+1,m+1}= (-1)^l |A| \tau^l_{n,m}$, by using (\ref{complex-red}), 
the four bilinear equations to the MTM (\ref{MTM-4}) are recovered from the four  bilinear equations (\ref{bl}). In summary, we have established the correspondence 
between the double-Wronskian solutions of the MTM system and those of the two-component KP hierarchy.

\section{Proof of Theorem \ref{theorem-rat-sol}} 
\label{sec-3}

We set $A = -I + L$ as in (\ref{choice-A}), where $L$ is the nilpotent matrix of index $(2N)$ defined by (\ref{matrix-L}). This choice for $A$ generalizes 
the Jordan block (\ref{Jordan-block-2}) from $2 \times 2$ to $(2N) \times (2N)$ matrices. Vector $\phi \in \mathbb{C}^{2N}$ satisfies the first system in (\ref{eigen}), from which we derive the following equations for components of $\phi$:
\begin{equation}
\label{eq-eigen}
\partial_{\xi} \phi_j = -{\rm i} \phi_j + {\rm i} \phi_{j-1}, \quad j = 1,2,\dots,2N
\end{equation}
and 
\begin{equation}
\label{eq-eigen-eta}
\partial_{\eta} \phi_j = -{\rm i} \phi_j - {\rm i}\phi_{j-1} - {\rm i} \phi_{j-2} - \dots - {\rm i} \phi_{1}, \quad j = 1,2,\dots,2N,
\end{equation}
closed with $\phi_0 \equiv 0$. The other vector $\psi \in \mathbb{C}^{2N}$ is defined by $\psi = S \bar{\phi}$ as in (\ref{rel-2}). The following lemma gives the unique expression for the real matrix $S$ solving (\ref{rel-1}). 

\begin{lemma}
	\label{lem-S}
	Solution of the matrix equation $-S^2 = A = -I + L$ is given by 
	\begin{equation}
	\label{choice-S}
	S = I - \frac{1}{2} L - \frac{1}{2^3} L^2 - \dots - \frac{(2m-3)!!}{m! 2^m} L^m - \dots - \frac{(4N-5)!!}{(2N-1)! 2^{2N-1}} L^{2N-1}.
	\end{equation}
\end{lemma}

\begin{proof}
For the Jordan block form $J = \lambda I + L \in \mathbb{C}^{n \times n}$ with $\lambda \in \mathbb{C}$, we use the following Taylor expansion for every smooth function $f : \mathbb{C} \to \mathbb{C}$ extended to matrices as $f : \mathbb{C}^{n \times n} \to \mathbb{C}^{n \times n}$: 
\begin{align*}
f(J) &= f(\lambda) I + f'(\lambda) L + \frac{1}{2!} f''(\lambda) L^2 + \dots + \frac{1}{m!} f^{(m)}(\lambda) L^m + \dots \\
&= \left( \begin{matrix} f(\lambda) & 0 & 0 & \dots & 0 \\
f'(\lambda) & f(\lambda) & 0 & \dots & 0 \\
\frac{1}{2!} f''(\lambda) & f'(\lambda) & f(\lambda) & \dots & 0 \\
\vdots & \vdots & \vdots & \ddots & \vdots \\
\frac{1}{(n-1)!} f^{(n-1)}(\lambda) & \frac{1}{(n-2)!} f^{(n-2)}(\lambda) & 
\frac{1}{(n-3)!} f^{(n-3)}(\lambda) & \dots & f(\lambda) 
\end{matrix} \right) 
\end{align*}
	We apply this formula for $A = -I + L$ with $f(\lambda) = \sqrt{\lambda}$ and $\lambda = -1$. We get recursively
	$$
	f(-1) = {\rm i}, \quad f'(-1) = -\frac{1}{2}{\rm i}, \quad 
	f''(-1) = -\frac{1}{2^2}{\rm i}, \quad f'''(-1) = -\frac{3!!}{2^3} {\rm i}, \quad f''''(-1) = -\frac{5!!}{2^4} {\rm i}, 
	$$
	and generally, 
	$$
	f^{(m)}(-1) = -\frac{(2m-3)!!}{2^m}  {\rm i}, \quad m \in \mathbb{N}.
	$$ 
	Defining $S = -{\rm i} \sqrt{A}$ and dividing $f^{(m)}(-1)$ by $m!$ yields (\ref{choice-S}). 
\end{proof}

The fundamental solution of equations (\ref{eq-eigen}) and (\ref{eq-eigen-eta}) 
denoted as $\Phi = (\Phi_1,\Phi_2,\dots,\Phi_{2N})^T$ can be obtained
by using the generating function:
\begin{equation}
\label{der-eigen}
\Phi_j = \frac{1}{(j-1)!} \partial^{j-1}_{\zeta^2} e^{{\rm i}(\zeta^2\xi+\zeta^{-2}\eta)} |_{\zeta^2 = -1}, \quad j = 1,2,\dots,2N.
\end{equation}
The example for $N = 2$ yields
\begin{align*}
\Phi = \left( \begin{array}{c} 1 \\
{\rm i} (\xi - \eta) \\
\frac{{\rm i}^2}{2!} (\xi - \eta)^2 - {\rm i} \eta \\
\frac{{\rm i}^3}{3!} (\xi - \eta)^3 + \eta (\xi - \eta) - {\rm i} \eta 
 \end{array} \right) \; e^{-{\rm i}(\xi + \eta)}.
\end{align*}
where we have used 
\begin{equation}
\label{gen-function}
e^{{\rm i}(\zeta^2\xi+\zeta^{-2}\eta)} = e^{-{\rm i}(\xi + \eta) + {\rm i} (\zeta^2 + 1) (\xi - \eta) - {\rm i} \eta \sum_{k=2}^{\infty} (\zeta^2 + 1)^k}.
\end{equation}

The general solution of  equations (\ref{eq-eigen}) and (\ref{eq-eigen-eta}) is given by the linear combination of the fundamental solutions (\ref{der-eigen}) with $(2N)$ complex parameters:
\begin{equation}
\label{der-eigen-lin}
\phi = \sum_{k=1}^{2N}  c_k L^{k-1} \Phi.
\end{equation}
The following lemma gives the exact count of irreducible real parameters in the general rational solution of the MTM system generated from (\ref{der-eigen-lin}) by using (\ref{MTM-bil}) and (\ref{Wronskian}). This yields the assertion 
of Theorem \ref{theorem-rat-sol} on $(2N)$ arbitrary real parameters.

\begin{lemma}
    \label{lem-param}
Let $\phi$ be defined by (\ref{der-eigen}) and (\ref{der-eigen-lin}) with $c_1,c_2,\dots,c_{2N} \in \mathbb{C}$ and $\psi = S \bar{\phi}$ with $S$ defined by (\ref{choice-S}). The rational solutions obtained by (\ref{MTM-bil}) and  (\ref{Wronskian}) depend on $(2N)$ arbitrary real parameters. 
\end{lemma}

\begin{proof}
Writing $\Phi_j(\xi,\eta) = P_{j-1}(\xi,\eta) e^{-{\rm i}(\xi + \eta)}$ in (\ref{der-eigen}) with some polynomials $P_{j-1}$ in $(\xi,\eta)$ of degree $j-1$, we can rewrite (\ref{der-eigen-lin}) in the equivalent form:
\begin{align*}
\phi_j(\xi,\eta) &= \left( \sum_{k=1}^j c_k P_{j-k}(\xi,\eta) \right) e^{-{\rm i}(\xi + \eta)} \\
&= \frac{1}{(j-1)!} \partial^{j-1}_{\zeta^2} \left( \sum_{k=1}^{j} c_k(\zeta^2+1)^{k-1} \right) e^{{\rm i}(\zeta^2\xi+\zeta^{-2}\eta)} \Biggr|_{\zeta^2 = -1}
\end{align*}
where we have used the Leibniz rule for $\partial_{\zeta^2}^{\,j-1}$ applied to
$\big(\sum_{k=1}^{2N}c_k(\zeta^2+1)^{k-1}\big)  e^{{\rm i}(\zeta^2\xi+\zeta^{-2}\eta)}$ and the identity
\begin{equation}\nonumber
\partial_{\zeta^2}^{\,m}(\zeta^2+1)^r\Big|_{\zeta^2=-1}=
\begin{cases}
r!, & m=r,\\
0, & m\neq r.
\end{cases}
\end{equation}
Furthermore, the expression for $\phi_j$ can be rewritten in the form:
\begin{align}
\phi_j(\xi,\eta) &= \frac{1}{(j-1)!} \partial^{j-1}_{\zeta^2} 
\left( \sum_{k=1}^j a_k (\zeta^2+1)^{k-1} \right) e^{{\rm i}\sum_{k=1}^j b_k (\zeta^2+1)^{k-1}} e^{\mathrm{i}(\zeta^2\xi+\zeta^{-2}\eta)} \Bigg|_{\zeta^2 = -1},
\label{phi-represent}
\end{align}
where $(4N)$ real coefficients $a_1,a_2,\dots,a_{2N}$ and $b_1,b_2,\dots,b_{2N}$ are uniquely computed from $(2N)$ complex coefficients $c_1,c_2,\dots,c_{2N}$ by using the recursive relations:
\begin{align*}
\left\{ \begin{array}{l} 
c_1 = a_1 e^{{\rm i} b_1}, \\
c_2 = (a_2 + {\rm i} a_1 b_2) e^{{\rm i} b_1}, \\
c_3 = (a_3 + {\rm i} a_2 b_2 + {\rm i} a_1 b_3 -\frac{1}{2} a_1 b_2^2) e^{{\rm i} b_1}.
\end{array} \right.
\end{align*}
The unique solution for $a_j, b_j \in \mathbb{R}$ at each $j = 1,2,\dots,2N$ is found from a linear equation with given $c_j \in \mathbb{C}$ and 
uniquely defined $\{ a_k, b_k\}_{k=1}^{j-1}$ by induction. 

Let $\tilde\phi_j$ denote the sequence obtained from (\ref{phi-represent}) for $a_1=1$ and $a_2=\ldots=a_{2N}=0$. The representation (\ref{phi-represent}) implies that $\phi_j$ can be written recursively as 
\begin{equation}\nonumber
	\phi_j = a_1 \tilde{\phi}_j + a_2 \tilde{\phi}_{j-1} + \ldots + a_j \tilde{\phi}_1, \quad j = 1,2,\dots,2N.
\end{equation}
Similarly, we obtain
\begin{equation}\nonumber
	\psi_j = a_1 \tilde{\psi}_j + a_2 \tilde{\psi}_{j-1} + \ldots + a_j \tilde{\psi}_1, \quad j = 1,2,\dots,2N,
\end{equation}
where $\tilde{\psi}_j = S \overline{\tilde{\phi}}_j$. Introducing 
$$
T = \left( \begin{matrix} a_1 & 0 & 0 & \dots & 0 \\
a_2 & a_1 & 0 & \dots & 0 \\ a_3 & a_2 & a_1 & \dots & 0 \\ 
\vdots & \vdots & \vdots & \ddots & \vdots \\ a_{2N} & a_{2N-1} & a_{2N-2} & \dots & a_1 \end{matrix} \right),
$$
we can write 
\begin{equation}\nonumber
\phi = T \tilde\phi,\qquad \psi = T \tilde\psi.
\end{equation}
Every double-Wronskian determinant in the solution (\ref{Wronskian}) 
can be written in terms of the double-Wronskian determinant computed from $\tilde{\phi}$ and $\tilde{\psi}$ multiplied by $|T| = a_1^{2N}$, e.g. 
\begin{align*}
f &= |\widetilde{N} ; \widehat{N-1}| \\
&= |\phi',\phi'',\dots,\phi^{(N)};\psi,\psi',\dots,\psi^{(N-1)}|, \\
&= |T \tilde{\phi}',T \tilde{\phi}'',\dots,T \tilde{\phi}^{(N)};T \tilde{\psi},T \tilde{\psi}',\dots,T \tilde{\psi}^{(N-1)}|, \\
&= |T| |\tilde{\phi}',\tilde{\phi}'',\dots,\tilde{\phi}^{(N)};\tilde{\psi}, \tilde{\psi}',\dots,\tilde{\psi}^{(N-1)}|.
\end{align*}
Since $|T| = a_1^{2N}$, parameters $a_2,a_3,\dots,a_{2N}$ do not affect the solution, whereas the parameter $a_1 \neq 0$ is canceled in the quotients (\ref{MTM-bil}) since $|T| = a_1^{2N}$ appears in both terms of the quotients. As a result, the quotients (\ref{MTM-bil}) only depend on $2N$ real parameters $b_1,b_2,\dots,b_{2N}$ which are generally irreducible.
\end{proof}

From now on, we only consider the rational solutions obtained 
from the fundamental solutions (\ref{der-eigen}) by using (\ref{MTM-bil}) and (\ref{Wronskian}). These solutions give the principal 
parts in the expansion of the polynomials $P_N(x,t)$, $Q_N(x,t)$, $R_N(x,t)$ 
defined by (\ref{rat-sol}). The folowing lemma shows that 
the highest powers of these polynomials in $x$ have  nonzero coefficients. 

\begin{lemma}
	\label{lem-leading-order}
	Let $\phi = \Phi$ be defined by (\ref{der-eigen}) and $\psi = S \bar{\Phi}$ with $S$ defined by (\ref{choice-S}). Then, the polynomials $P_N(x,t)$, $Q_N(x,t)$, $R_N(x,t)$ 
	defined by (\ref{rat-sol}) are expanded as follows:
	\begin{align}
	P_N(x,t) &= a_N^{(0)} x^{N^2} + \mathcal{O}(x^{N^2-1}), \label{P_N} \\
	Q_N(x,t) &= -{\rm i}N a_N^{(0)} x^{N^2-1} + \mathcal{O}(x^{N^2-2}), \label{Q_N} \\
	R_N(x,t) &= {\rm i}N a_N^{(0)} x^{N^2-1} + \mathcal{O}(x^{N^2-2}),  \label{R_N}
	\end{align}
	where the numerical coefficient $a_N^{(0)}$ is given by 
	\begin{align}
	\label{a_N}
	a_N^{(0)} =  \frac{(-1)^{\frac{N(N+1)}{2}}}{2^{N(N-1)}} 
		\frac{1}{1^{2N-1} 3^{2N-3} 5^{2N-5} 7^{2N-7} \dots (2N-3)^3 (2N-1)^{1}}.
	\end{align}
\end{lemma}

\begin{proof}
	By (\ref{Wronskian}) with (\ref{choice-A}), we have
	\begin{align*}
	f &= |\phi',\phi'',\dots,\phi^{(N)};\psi,\psi',\dots,\psi^{(N-1)}| \\
	&= |\phi',-\mathrm{i} \phi' + \mathrm{i} L \phi',\dots, -{\rm i}\phi^{(N-1)}+{\rm i}L \phi^{(N-1)};\psi,\mathrm{i} \psi -\mathrm{i} L \psi,\dots,{\rm i}\psi^{(N-2)}-{\rm i} L \phi^{(N-2)}| \\
	&= |\phi',L \phi',\dots, L\phi^{(N-1)};\psi, L \psi, \dots, L \psi^{(N-2)}|. 
	\end{align*}
	Continuing by induction, we reduce this expression to 
	\begin{align}
	\label{expression-f}
	f &= |\phi',L \phi',\dots, L^{N-1} \phi';\psi,L \psi,\dots,L^{N-1} \psi |.
	\end{align}
	Similarly, we reduce $g$ and $h$ to 
	\begin{align*}
	g &= |\phi,\phi',\dots,\phi^{(N)};\psi',\psi'',\dots,\psi^{(N-1)}| \\
	&= {\rm i}^{2N-1}  |\phi,L \phi,\dots, L^N \phi;\psi',L \psi',\dots,L^{N-2} \psi'|
	\end{align*}
	and 
	\begin{align*}
	h &={\rm i} C^{-1} |\phi,\phi',\dots,\phi^{(N)};\psi,\psi',\dots,\psi^{(N-2)}| \\
	&= C^{-1} {\rm i}^{2N}  |\phi,L \phi,\dots, L^N \phi;\psi,L \psi,\dots,L^{N-2} \psi|,
	\end{align*}
	where the factor ${\rm i}^{2N-1}$ is due to the two columns with $\phi^{(N-1)}$ and $\phi^{(N)}$ which are not compensated by the columns from $\psi$.
	
We observe from (\ref{der-eigen}) and (\ref{gen-function}) that $\Phi_j(\xi,\eta) = P_{j-1}(\xi,\eta) e^{-{\rm i} (\xi + \eta)}$ with 
\begin{equation}
\label{phi-poly}
P_{j-1}(\xi,\eta) = \frac{{\rm i}^{j-1}}{(j-1)!}  (\xi - \eta)^{j-1} + \eta p_{j-3}(\xi-\eta,\eta), \quad j = 1,2,\dots, 2N, 
\end{equation}
where $p_{j-3}$ is a polynomial in variables $\xi-\eta = \frac{1}{2} x$ and $\eta = \frac{1}{4} (t-x)$. The degree of polynomials $P_N$, $Q_N$, $R_N$ in $x$ can be obtained by inspecting the leading order of $f$, $g$, $h$ with the first dominant term in (\ref{phi-poly}). 

Subtituting $S\bar{\phi} = \bar{\phi} + \mathcal{O}(L\bar{\phi})$ and $\phi' = (-\mathrm{i}) \phi + \mathcal{O}(L \phi)$ into (\ref{expression-f}), see (\ref{eq-eigen}) and (\ref{choice-S}), we rewrite 
the leading-order part of the polynomial $f$ in variable 
$$
z := {\rm i} (\xi - \eta) = \frac{{\rm i}}{2} x.
$$ 
By using the first dominant term in (\ref{phi-poly}), we obtain 
\begin{align}
f &= (-{\rm i})^N \left| b_N(z),Lb_N(z),\dots,L^{N-1} b_N(z); b_N(-z),Lb_N(-z),\dots,L^{N-1} b_N(-z) \right| \notag \\
& \qquad \times \left[ 1 + \mathcal{O}(z^{-1}) \right],
\label{f-expansion}
\end{align}
where the column vector $b_N \in \mathbb{C}^{2N}$ is given by  
\begin{align*}
b_N(z) := \left( \begin{matrix} 1 \\
z  \\
\frac{1}{2!} z^2  \\
\frac{1}{3!} z^3  \\
\vdots \\
\frac{1}{(2N-1)!} z^{2N-1}
\end{matrix} \right).
\end{align*}
Due to the hierarchical structure of $b_N(\pm z)$ in powers of $z$, 
we can write the matrix 
$$
M_N(z) := \left(  b_N(z),Lb_N(z),\dots,L^{N-1} b_N(z); b_N(-z),Lb_N(-z),\dots,L^{N-1} b_N(-z) \right)
$$
in the factorized form:
\begin{align}
\label{M-transformation}
M_N(z) = D_-(z) M_N(1) D_+(z),
\end{align}
where 
\begin{align*}
D_-(z) &:= {\rm diag}(z^{-N+1},z^{-N+2}, \dots, 1; z, z^2, \dots, z^N), \\
D_+(z) &:= {\rm diag}(z^{N-1},z^{N-2}, \dots, 1; z^{N-1}, z^{N-2}, \dots, 1).
\end{align*}
Since 
$$
\sum_{j=1}^N j + \sum_{j=1}^{N-1} j = \frac{N(N+1)}{2} + \frac{N(N-1)}{2} = N^2,
$$
we obtain from the properties of determinants that 
\begin{align}
\left| M_N(z) \right| = z^{N^2}  \left| M_N(1) \right|. 
\label{f-expansion-2}
\end{align}
We show in Appendix \ref{app-A} that
\begin{align}
|M_N(1)| = \frac{2^N (-1)^N}{1^{2N-1} 3^{2N-3} 5^{2N-5} 7^{2N-7} \dots (2N-3)^3 (2N-1)^{1}}.
\label{det-expression-2}
\end{align}
By using (\ref{f-expansion}), (\ref{f-expansion-2}), and (\ref{det-expression-2}), 
we obtain (\ref{P_N}) with (\ref{a_N}) since $z = \frac{\mathrm{i}}{2} x$ and 
$$
(-{\rm i})^N {\rm i}^{N^2} (-1)^N = {\rm i}^{N(N+1)} = (-1)^{\frac{N(N+1)}{2}}.
$$

Expressions for $g$ and $h$ are not polynomials but they are given by the polynomials multiplied by $e^{-2 {\rm i}(\xi + \eta)} = e^{-{\rm i} t}$. Therefore, we can compute the leading-order parts of $g$ function and $h$ function as
\begin{align}
g &= \mathrm{i}^{3N-2} e^{-\mathrm{i} t} 
\left| \tilde{M}_N(z) \right|  \left[ 1 + \mathcal{O}(z^{-1})\right] 
\label{g-expansion}
\end{align}
and
\begin{align}
h &= {\rm i}^{2N} C^{-1} e^{-\mathrm{i} t} \left| \tilde{M}_N(z) \right|  \left[ 1 + \mathcal{O}(z^{-1})\right],
\label{h-expansion}
\end{align}
where 
$$
\tilde{M}_N(z) :=  \left( b_N(z),Lb_N(z),\dots,L^{N} b_N(z); b_N(-z),Lb_N(-z),\dots,L^{N-2} b_N(-z) \right). 
$$
Again, we can factorize the matrix as 
\begin{align*}
 \tilde{M}_N(z) = \tilde{D}_-(z) \tilde{M}_N(1) \tilde{D}_+(z),
\end{align*}
where 
\begin{align*}
\tilde{D}_-(z) &= {\rm diag}(z^{-N},z^{-N+1}, \dots, 1; z, z^2, \dots, z^{N-1}), \\
\tilde{D}_+(z) &= {\rm diag}(z^{N},z^{N-1}, \dots, 1; z^{N}, z^{N-1}, \dots, z^2).
\end{align*}
Since 
$$
\sum_{j=1}^{N-1} j + \sum_{j=2}^{N} j = \frac{N(N-1)}{2} + \frac{N(N+1)}{2} - 1 = N^2 - 1,
$$
we obtain from the properties of determinants that 
\begin{align}
 \left| \tilde{M}_N(z) \right| = z^{N^2-1} \left| \tilde{M}_N(1) \right|. 
\label{g-expansion-2}
\end{align}
We show in Appendix \ref{app-B} that  
\begin{align}
\left| \tilde{M}_N(1) \right| = 
\frac{2^{N-1} (-1)^{N-1} N}{1^{2N-1} 3^{2N-3} 5^{2N-5} 7^{2N-7} \dots (2N-3)^3 (2N-1)^{1}}.
\label{det-expression-3}
\end{align}
By using (\ref{g-expansion}), (\ref{g-expansion-2}), and (\ref{det-expression-3}), 
we obtain (\ref{Q_N}) with the same $a_N^{(0)}$ given by (\ref{a_N}) since $z = \frac{\mathrm{i}}{2} x$ and 
$$
{\rm i}^{3N-2} {\rm i}^{N^2-1} (-1)^{N-1} = (-{\rm i}) {\rm i}^{N(N+3)} (-1)^N = (-{\rm i}) (-1)^{\frac{N(N+5)}{2}} = (-{\rm i}) (-1)^{\frac{N(N+1)}{2}}.
$$
Finally, we have $|S| = 1$ so that $C = (-{\rm i})^N$. Therefore, 
${\rm i}^{2N} C^{-1} = {\rm i}^{3N} = -{\rm i}^{3N-2}$ in (\ref{h-expansion}) so that (\ref{R_N}) follows from (\ref{Q_N}).
\end{proof}

The following lemma ensures that the rational solution given by (\ref{rat-sol}) is bounded for all $(x,t) \in \mathbb{R}^2$, in agreement with the last assertion 
of Theorem \ref{theorem-rat-sol}. 

\begin{lemma}
    \label{lem-bounded}
    Let $P_N$, $Q_N$, and $R_N$ be polynomials defined in Lemma \ref{lem-leading-order}. The rational solution $(u,v)$ in (\ref{rat-sol}) is bounded for all $(x,t) \in \mathbb{R}^2$. 
\end{lemma}

\begin{proof}
If $P_N(x,t) \neq 0$ for all $(x,t) \in \mathbb{R}^2$, then zeros of the polynomial $P_N$ of degree $N^2$ in $x$ are bounded away from the real axis for every $t \in \mathbb{R}$ so that the rational solution $(u,v)$ is bounded for all $(x,t) \in \mathbb{R}^2$. 

Assume now that $P_N$ has a zero at $(x,t) = (x_0,t_0) \in \mathbb{R}^2$ of multiplicity $m$. Without loss of generality, we can fix $\eta = \eta_0$ and consider the behavior of 
$P_N = (\xi - \xi_0)^m \tilde{P}_N$, where $\tilde{P}_N$ is a polynomial of degree $N^2 - m$ and $\tilde{P}_N(\xi_0,\eta_0) \neq 0$. Due to reality of $(\xi_0,\eta_0)$, both $f = P_N$ and $\bar{f} = \bar{P}_N$ have a zero at $(\xi_0,\eta_0)$ of the same multiplicity $m$ so that ${\rm i} D_{\xi} (f \cdot \bar{f})$ in the third bilinear equation of the system (\ref{MTM-4}) has a zero at $(\xi_0,\eta_0)$ of multiplicity $2m$. Hence $h = R_N e^{-\mathrm{i}t}$ has a zero at $(\xi_0,\eta_0)$ of multiplicity $m$. The first bilinear equation in the system (\ref{MTM-4}) implies that $g = Q_N e^{-\mathrm{i} t}$ has also a zero at $(\xi_0,\eta_0)$ of multiplicity $m$. Thus, $Q_N = (\xi - \xi_0)^m \tilde{Q}_N$ and $R_n = (\xi - \xi_0)^m \tilde{R}_N$ 
with polynomial $\tilde{Q}_N,\tilde{R}_N$ satisfying $\tilde{Q}_N(\xi_0,\eta_0) \neq 0$ and $\tilde{R}_N(\xi_0,\eta_0) \neq 0$. Hence $\xi_0$ is a removable singularity of the rational functions $Q_N/\bar{P}_N$ and $R_N/P_N$ so that the rational solution $(u,v)$ is bounded for all $(x,t) \in \mathbb{R}^2$. The same analysis holds for fixed $\xi = \xi_0$ with respect to $\eta$ but we use the second and fourth bilinear equations in the system (\ref{MTM-4}) to prove that $\eta_0$ is a removable singularity of the rational functions. 
\end{proof}

\begin{remark}
    We conjecture that $P_N(x,t) \neq 0$ for all $(x,t) \in \mathbb{R}^2$. As Lemma \ref{lem-bounded} shows, this property is not important for regularity of the solution given by (\ref{rat-sol}) since the bilinear equations (\ref{MTM-4}) ensure that even if $P_N$ vanishes at some $(x_0,t_0) \in \mathbb{R}^2$, then it is a removable singularity of the rational solution $(u,v)$. 
\end{remark}

Lemmas \ref{lem-param}, \ref{lem-leading-order}, and \ref{lem-bounded} complete the proof of Theorem \ref{theorem-rat-sol}.

\section{Proof of Theorem \ref{theorem-polyn}} 
\label{sec-4}

We start with the principal part of the polynomial $P_N(x,t)$ denoted as $p_N(x,t)$. The following lemma guarantees that $p_N$ has the representation (\ref{polyno-principal}) with real-valued coefficients $\{ a_N^{(j)}\}_{j=0}^J$ for $J = \frac{N^2}{2}$ if $N$ is even and $J = \frac{N^2-1}{2}$ if $N$ is odd. 

\begin{lemma}
	\label{lem-prinicipal-polyn}
	Let $\phi = \Phi$ be defined by (\ref{der-eigen}) and $\psi = S \bar{\Phi}$ with $S$ defined by (\ref{choice-S}). By using the scaling transformation 
	$x \mapsto \lambda x$ and $t \mapsto \lambda^2 t$ with $\lambda > 0$, 
	we have 
	\begin{equation}
	\label{lambda-leading}
	P_N(\lambda x, \lambda^2 t) = \lambda^{N^2} p_N(x,t) + \mathcal{O}(\lambda^{N^2-1}), \quad \mbox{\rm as} \;\; \lambda \to \infty,
	\end{equation}
where $p_N(x,t)$ is given by (\ref{polyno-principal}) with real-valued coefficients. 
\end{lemma}

\begin{proof}
	We recall from (\ref{der-eigen}) and (\ref{gen-function}) that $\Phi_j(\xi,\eta) = P_{j-1}(\xi,\eta)  e^{-{\rm i}(\xi + \eta)}$ with 
	$$
P_{j-1}(\xi,\eta) =  \frac{{\rm i}^{j-1}}{(j-1)!} (\xi - \eta)^{j-1} -  \frac{{\rm i}^{j-2}}{(j-3)!}  \eta (\xi - \eta)^{j-3} + \eta p_{j-4}(\xi-\eta,\eta), \quad j = 1,2,\dots,2N,
	$$
	where $p_{j-4}(\xi-\eta,\eta)$ is a polynomial of degree $j-4$ in variables $\xi - \eta = \frac{1}{2} x$ and $\eta = \frac{1}{4} (t - x)$. Extracting the leading-order part after the scaling transformation $x \mapsto \lambda x$ and $t \mapsto \lambda^2 t$, we obtain 
	$$
P_{j-1} = \frac{{\rm i}^{j-1} \lambda^{j-1}}{2^{j-1} (j-1)!} \left( x^{j-1} + {\rm i} (j-1) (j-2)  x^{j-3} t \right) + \mathcal{O}(\lambda^{j-2}),\quad j = 1,2,\dots,2N.
	$$
	Similarly, we obtain $(S \overline{\Phi}(\xi,\eta))_j = Q_{j-1}(\xi,\eta) e^{{\rm i}(\xi + \eta)}$ with 
	$$
Q_{j-1} = \frac{(-{\rm i})^{j-1} \lambda^{j-1}}{2^{j-1} (j-1)!} \left( x^{j-1} - {\rm i} (j-1) (j-2)  x^{j-3} t \right) + \mathcal{O}(\lambda^{j-2}),\quad j = 1,2,\dots,2N,
$$
since $S$ does not modify the principal terms of $\overline{\Phi}_j$ due to 
the representation (\ref{choice-S}). 

Let us now denote $\tau := {\rm i} t$ in the principal terms of $P_{j-1}$ and $Q_{j-1}$. By using the representation (\ref{expression-f}) for $f = P_N(x,t)$ and accounting for the principal terms in $\lambda$ in the expansion (\ref{lambda-leading}) with $a_N^{(0)} \neq 0$ in (\ref{P_N}) and (\ref{a_N}), we obtain the definition of $p_N$ as 
\begin{align}
p_N = (-{\rm i})^N \left| b_N, Lb_N,\dots,L^{N-1} b_N; c_N,L c_N,\dots,L^{N-1} c_N \right|,
\label{p-double}
\end{align}
where the column vectors $b_N,c_N \in \mathbb{C}^{2N}$ are given by  
\begin{align*}
b_N = \left( \begin{matrix} 1 \\
\frac{{\rm i}}{2} x \\
\frac{{\rm i}^2}{2^2 2!} (x^2 + 2 \tau)\\
\frac{{\rm i}^3}{2^3 3!} (x^3 + 6 x \tau) \\
\vdots \\
\frac{{\rm i}^{2N-1}}{2^{2N-1} (2N-1)!} (x^{2N-1} + (2N-2) (2N-1) x^{2N-3} \tau)
\end{matrix} \right)
\end{align*}
and
\begin{align*}
c_N = \left( \begin{matrix} 1 \\
-\frac{{\rm i}}{2} x \\
\frac{{\rm i}^2}{2^2 2!} (x^2 - 2 \tau)\\
-\frac{{\rm i}^3}{2^3 3!} (x^3 - 6 x \tau) \\
\vdots \\
-\frac{{\rm i}^{2N-1}}{2^{2N-1} (2N-1)!} (x^{2N-1} - (2N-2) (2N-1) x^{2N-3} \tau)
\end{matrix} \right). 
\end{align*}
The purely imaginary coefficient ${\rm i}$ in the alternating rows of $b_N$ and $c_N$ can be removed from $b_N$ and $c_N$ by the transformation (\ref{M-transformation}) with $D_{\pm}(z)$ replaced by $D_{\pm}({\rm i})$. Since the leading-order term in (\ref{P_N}) has the real-valued coefficient $a_N^{(0)} \neq 0$ in (\ref{a_N}), then $p_N$ defined by (\ref{p-double}) is a polynomial in variables $x$ and $\tau$ with the real-valued coefficients. 

Since $\tau = {\rm i} t$, we only need to show that the polynomial $p_N$ is even with respect to $\tau$ so that 
the coefficients of the polynomial $p_N$ remain real-valued after we substitute back $\tau = {\rm i} t$. In view of the scaling transformation $x \mapsto \lambda x$ and $t \mapsto \lambda^2 t$, the evenness of $p_N$ also imply the representation (\ref{polyno-principal}) in powers of $x^{-4} t^2$ relative to the leading-order power $x^{N^2}$ with real-valued coefficients.

Let us now show that $p_N$ is even in $\tau$. We claim that 
\begin{equation}
\label{parity-p}
p_N(-x,t) = (-1)^N p_N(x,t),
\end{equation}
which coincides with the parity of the principal term $a_N^{(0)} x^{N^2}$. Indeed, 
$p_N$ is a polynomial in $\tau$ and since $\tau$ is balanced with $x^2$ in the definition of $b_N$ and $c_N$, $p_N$ is a polynomial in powers $x^{-2} t$ relative to the principal term $a_N^{(0)} x^{N^2}$ implying (\ref{parity-p}). It follows 
that $b_N$ and $c_N$ map to each other under the transformation: $x \mapsto -x$ and $\tau \mapsto -\tau$. After the transformation, interchanging $N$ first columns with $N$ last columns of the determinant returns the original determinant before the transformation, which implies that 
\begin{equation*}
p_N(-x,-t) = (-1)^N p_N(x,t).
\end{equation*}
In view of (\ref{parity-p}), this implies that $p_N$ is even in $\tau = {\rm i} t$, hence it has real-valued coefficients in powers of $x^{-4} t^2$ relative to the principal term $a_N^{(0)} x^{N^2}$.
\end{proof}

We assume that $\hat{p}_N$ in (\ref{polyno-upsilon}) admits exactly $N$ real roots. 
Since the polynomial $\hat{p}_N$ has real-valued coefficients by Lemma \ref{lem-prinicipal-polyn}, then $N(N-1)$ roots of $\hat{p}_N$ are complex-conjugate so that $\hat{p}_N$ has $\frac{N(N-1)}{2}$ roots in $\mathbb{C}^+$ and $\mathbb{C}^-$. The following lemma counts roots of the polynomial $P_N(\cdot,t)$ for large $|t|$.

\begin{lemma}
	\label{lemma-polyn}
	Assume that $\hat{p}_N$ in (\ref{polyno-upsilon}) admits exactly $N$ real roots. Then, $P_N(\cdot,t)$ admits $\frac{N(N-1)}{2}$ roots in $\mathbb{C}^+$ and $\frac{N(N+1)}{2}$ roots in $\mathbb{C}^-$ for large $|t|$.
\end{lemma}

\begin{proof}
Returning back to the expression (\ref{expression-f}), we compute 
	the next-order correction in the expansion of $f = P_N$. Substituting 
	$$
	\phi' = -{\rm i} \phi + {\rm i} L \phi, \quad \psi = S \bar{\phi} = \bar{\phi} - \frac{1}{2} L \bar{\phi} + \mathcal{O}(L^2 \bar{\phi})
	$$
	into (\ref{expression-f}) yields the expansion 
	\begin{align*}
	f &= (-{\rm i})^N |\phi,L \phi,\dots, L^{N-1} \phi; \bar{\phi},L \bar{\phi}, \dots, L^{N-1} \bar{\phi} |  \\
	& \qquad - 
	(-{\rm i})^N |\phi,L \phi,\dots, L^{N-2} \phi, L^{N} \phi; \bar{\phi},L \bar{\phi}, \dots, L^{N-1} \bar{\phi} | \\
	& \qquad - \frac{1}{2} (-{\rm i})^N |\phi,L \phi,\dots, L^{N-1} \phi; \bar{\phi},L \bar{\phi}, \dots, L^{N-2} \bar{\phi}, L^N \bar{\phi} |
	+ \mathcal{O}(z^{N^2-2}).
	\end{align*}
We recall from (\ref{f-expansion}), (\ref{f-expansion-2}), and (\ref{det-expression-2}) that 
\begin{align*}
(-{\rm i})^N |\phi,L \phi,\dots, L^{N-1} \phi; \bar{\phi},L \bar{\phi}, \dots, L^{N-1} \bar{\phi} | = a_N^{(0)} x^{N^2} \left[1 + \mathcal{O}(z^{-1}) \right],
\end{align*}
where $a_N^{(0)}$ is given by (\ref{a_N}). The correction of the order of $\mathcal{O}(x^{N^2-1})$ is due to the correction terms $\eta p_{j-3}(\xi-\eta,\eta)$ in the representation (\ref{phi-poly}). However, as shown in Lemma \ref{lem-prinicipal-polyn}, 
the correction terms produces powers of $(\xi - \eta)^{-4} \eta^2 = \mathcal{O}(z^{-2})$ relative to the principal term $(\xi - \eta)^{N^2} = \mathcal{O}(z^{N^2})$, which improves the previous expansion in the form:
\begin{align}
(-{\rm i})^N |\phi,L \phi,\dots, L^{N-1} \phi; \bar{\phi},L \bar{\phi}, \dots, L^{N-1} \bar{\phi} | = a_N^{(0)} x^{N^2} \left[1 + \mathcal{O}(z^{-2}) \right]. \label{computaton-1}
\end{align}
Similarly to (\ref{f-expansion}), we write
\begin{align}
\label{computaton-2}
|\phi,L \phi,\dots, L^{N-2} \phi, L^{N} \phi; \bar{\phi},L \bar{\phi}, \dots, L^{N-1} \bar{\phi} | &= |M_N^{(1)}(z)| \left[ 1 + \mathcal{O}(z^{-1}) \right], \\
\label{computaton-3}
|\phi,L \phi,\dots, L^{N-1} \phi; \bar{\phi},L \bar{\phi}, \dots, L^{N-2} \bar{\phi}, L^N \bar{\phi} | &= |M_N^{(2)}(z)| \left[ 1 + \mathcal{O}(z^{-1}) \right], 
\end{align}
where 
\begin{align*}
M_N^{(1)}(z) &:= \left( b_N(z),L b_N(z),\dots, L^{N-2} b_N(z), L^{N} b_N(z); b_N(-z),L b_N(-z), \dots, L^{N-1} b_N(-z) \right), \\
M_N^{(2)}(z) &:= \left( b_N(z),L b_N(z),\dots, L^{N-1} b_N(z); b_N(-z),L b_N(-z), \dots, L^{N-2} b_N(-z), L^{N} b_N(-z) \right). 
\end{align*}
We follow the factorization formula (\ref{M-transformation}) and represent
\begin{align*}
|M_N^{(1)}(z)| = |D_-(z)| |M_N^{(1)}(1)| |D_+^{(1)}(z)|, \quad 
|M_N^{(2)}(z)|  = |D_-(z)| |M_N^{(2)}(1)| |D_+^{(2)}(z)|, 
\end{align*}
where $D_-(z)$ is the same as in (\ref{M-transformation}) but $D_+^{(1)}(z)$, $D_+^{(2)}(z)$ are given by the following modifications of $D_+(z)$:
\begin{align*}
D_+^{(1)}(z) &:= {\rm diag}(z^{N-1},z^{N-2}, \dots, z, z^{-1}; z^{N-1}, z^{N-2}, \dots, 1), \\
D_+^{(2)}(z) &:= {\rm diag}(z^{N-1},z^{N-2}, \dots, 1; z^{N-1}, z^{N-2}, \dots, z, z^{-1}).
\end{align*}
By using the properties of determinants, we obtain 
\begin{align}
\label{computaton-4}
|M_N^{(1)}(z)| =  z^{N^2-1} |M_N^{(1)}(1)|, \quad 
|M_N^{(2)}(z)| =  z^{N^2-1} |M_N^{(2)}(1)|.
\end{align}
We show in Appendix \ref{app-C} that
\begin{equation}
\label{det-expression-4}
|M_N^{(1)}(1)| = - |M_N^{(2)}(1)| = -\frac{2^{N-1} (-1)^{N-1} N^2}{1^{2N-1} 3^{2N-3} 5^{2N-5} 7^{2N-7} \dots (2N-3)^3 (2N-1)^{1}}.
\end{equation}
Combining (\ref{det-expression-4}) with (\ref{computaton-1}), (\ref{computaton-2}), (\ref{computaton-3}), (\ref{computaton-4})  and using $z = \frac{{\rm i}}{2} x$, we 
obtain the expansion
\begin{align}
f = a_N^{(0)} x^{N^2} + \frac{{\rm i}}{2} N^2 a_N^{(0)} x^{N^2-1} + \mathcal{O}(x^{N^2-2}) = a_N^{(0)} \left( x + \frac{{\rm i}}{2} \right)^{N^2} + \mathcal{O}(x^{N^2-2}).
\label{computation-2}
\end{align}
The principal part of $P_N(\cdot,t)$ given by $p_N(\cdot,t)$ in (\ref{polyno-principal}) admits exactly $N$ real roots for large $|t|$, which are scaled as $x = \mathcal{O}(\sqrt{|t|})$. These roots are shifted to $\mathbb{C}^-$ due to (\ref{computation-2}). On the other hand, the $\frac{N(N-1)}{2}$ roots of $p_N(\cdot,t)$ in either $\mathbb{C}^+$ or $\mathbb{C}^-$ stay in $\mathbb{C}^+$ and $\mathbb{C}^-$ for large $|t|$ since the representation (\ref{computation-2}) suggests $\mathcal{O}(1)$ correction to the imaginary parts of complex roots, which are of the order of $\mathcal{O}(\sqrt{|t|})$. 
This yields $\frac{N(N-1)}{2}$ roots in $\mathbb{C}^+$ and 
$\frac{N(N+1)}{2}$ roots in $\mathbb{C}^-$ for $P_N(x,t)$ in $x$ for large $|t|$.
\end{proof}

\begin{remark}
	In the proof of Lemma \ref{lemma-polyn}, we used $\phi = \Phi$ defined by (\ref{der-eigen}) similarly to Lemma \ref{lem-leading-order}. If we use the more general expression (\ref{der-eigen-lin}), then the corrections terms from coefficients $c_2,c_3,\ldots,c_{2N}$ do not appear in the two leading orders of the expansion (\ref{computation-2}) due to the hierarchical structure of the double-Wronskian determinants.
\end{remark}

Lemma \ref{lemma-polyn} gives the first assertion of Theorem \ref{theorem-polyn}. The second assertion is proven with the following lemma. 

\begin{lemma}
	\label{lemma-mass}
Under the same assumption as in Lemma \ref{lemma-polyn}, we have $\mathcal{M}_N(u,v) = 4 \pi N$, where $\mathcal{M}_N(u,v)$ is the mass integral (\ref{mass-intro}) computed at the 
rational solution (\ref{rat-sol}).
\end{lemma}

\begin{proof}
	It follows from (\ref{MTM-bil}) and (\ref{bilinear}) that 
	$$
	|u|^2 + |v|^2 = \frac{|g|^2 + |h|^2}{|f|^2} = 2 {\rm i} \left( \frac{f_x}{f} - 
	\frac{\bar{f}_x}{\bar{f}} \right),
	$$
	where $f(x,t) = P_N(x,t)$. By Lemma \ref{lemma-polyn}, $P_N$ has no roots in $x$ on $\mathbb{R}$ for large $|t|$ and it admits  $\frac{N(N-1)}{2}$ roots in $\mathbb{C}^+$ and $\frac{N(N+1)}{2}$ roots in $\mathbb{C}^-$.  By using (\ref{computation-2}), we have 
	$$
	\frac{f_x}{f} - 
	\frac{\bar{f}_x}{\bar{f}} = 
	\frac{N^2}{x + \frac{{\rm i}}{2}} - \frac{N^2}{x - \frac{{\rm i}}{2}} 
	+ \mathcal{O}\left(\frac{1}{|x|^2}\right) = 
	\mathcal{O}\left(\frac{1}{|x|^2}\right) \quad \mbox{\rm as} \;\; |x| \to \infty.
	$$
Therefore, the Jordan lemma of complex analysis is satisfied with 
	$$
	\lim_{R \to \infty} \int_0^{\pi}  \left( \frac{f_x}{f} - 
	\frac{\bar{f}_x}{\bar{f}} \right) \biggr|_{x = R e^{{\rm i} \theta}} {\rm i} R e^{{\rm i} \theta} d \theta = 0.
	$$
	By adding this integral to the integral on $[-R,R]$, we compute the mass integral $\mathcal{M}_N(u,v)$ by using the argument principle:
	\begin{align*}
\mathcal{M}_N(u,v) &= \lim_{R \to \infty} 
	\int_{[-R,R]} 2 \mathrm{i} \left( \frac{f_x}{f} - 
	\frac{\bar{f}_x}{\bar{f}} \right) dx \\
	&= -4 \pi \sum_{x_j \in \mathbb{C}^+ : f(x_j) = 0} {\rm Res}_{x = x_j} \frac{f_x}{f} + 4 \pi \sum_{x_j^* \in \mathbb{C}^+ : \bar{f}(x_j^*) = 0}
	{\rm Res}_{x = x_j^*} \frac{\bar{f}_x}{\bar{f}}   \\
	&= 4 \pi \left( \frac{N(N+1)}{2} - \frac{N (N-1)}{2} \right) = 4 \pi N,
	\end{align*}
which gives the assertion due to the conservation of the mass integral (\ref{mass-intro}) in $t \in \mathbb{R}$.
\end{proof}

Lemmas \ref{lemma-polyn} and \ref{lemma-mass} complete the proof of Theorem \ref{theorem-polyn}.

\section{Examples of the rational solutions}
\label{sec-5}

We illustrate the hierarchy of rational solutions constructed in Theorem \ref{theorem-rat-sol} with explicit examples for $N = 2,\dots,6$. We will 
show that the assumptions of Theorem \ref{theorem-polyn} are satisfied. 
By plotting the solution surfaces, we show that the 
corresponding rational solution describes the slow scattering of $N$ 
algebraic solitons on the time scale $\mathcal{O}(\sqrt{t})$.

\subsection{Double algebraic soliton}
\label{sec1}

We use the double-Wronskian solutions (\ref{Wronskian}) with $N = 2$ and 
define in (\ref{rel-1}):
\begin{equation*}
A =\begin{pmatrix}
-1 & 0 & 0 & 0 \\
1 & -1 & 0 & 0 \\
0 & 1 & -1 & 0 \\
0 & 0 & 1 & -1
\end{pmatrix}
\qquad
\mbox{\rm and}
\qquad
S=\begin{pmatrix}
1 & 0 & 0 & 0 \\
-\frac{1}{2} & 1 & 0 & 0 \\
-\frac{1}{8} & -\frac{1}{2} & 1 & 0 \\
-\frac{1}{16} & -\frac{1}{8} & -\frac{1}{2} & 1
\end{pmatrix}.
\end{equation*}
Since $C=(-{\rm i})^{2}/|S| =-1$, we obtain from (\ref{Wronskian}):
\begin{equation*}
\begin{cases}
f = \vert \phi', \phi''; \psi, \psi' \vert, \\
g = \vert \phi, \phi', \phi''; \psi' \vert, \\
h = -{\rm i} \vert \phi, \phi', \phi''; \psi \vert,
\end{cases}
\end{equation*}
which generates the exact solution due to (\ref{MTM-bil}) and (\ref{transformation}) 
\begin{equation}
\label{rat-2}
u_{2}(x,t) = \frac{Q_2(x,t)}{\bar{P}_2(x,t)} {\mathrm e}^{-{\rm i} t}, \quad v_{2}(x,t) = \frac{R_2(x,t)}{P_2(x,t)} {\mathrm e}^{-{\rm i} t},
\end{equation}
where
\begin{align*}
P_2(x,t) &= -32 {\rm i} \,x^{3}-16 x^{4}-24 {\rm i} x +48 t^{2}-24 x^{2}+3, \\
Q_2(x,t) &= -4 \left(-8 {\rm i} \,x^{3}+12 {\rm i} t +6 {\rm i} x -24 t x -12 x^{2}-3\right), \\
R_2(x,t) &= -4 \left(8 {\rm i} \,x^{3}+12 {\rm i} t -6 {\rm i} x +24 t x -12 x^{2}-3\right).
\end{align*}

\begin{remark}
	Compared to the definition of $P_N$, $Q_N$, and $R_N$ in (\ref{rat-sol}), 
	we divide these polynomials by the common factor 
	$$
		2^{N(2 N-1)} 3^{2N-3} 5^{2N-5}  \ldots  (2N-1),
	$$
	which does not change the rational function. In addition, we use the fundamental solutions $\phi = \Phi$ given by (\ref{der-eigen}) without the arbitrary parameters obtained from the linear superposition (\ref{der-eigen-lin}). 
	This remark applies to all consequent examples.
\end{remark}

\begin{remark}
	The exact solution (\ref{rat-2}) was derived in \cite{Han-2024} and studied in \cite{LiPelin2025}, where it was written in the equivalent form 
	obtained after the transformation $(u,v) \mapsto -(u,v)$, which does not affect 
	the MTM system (\ref{MTM}).
\end{remark}

Figure \ref{FIG:1} shows the solution surface (left) and the density plot 
(right). The principal part of the polynomial $P_2$ is given by 
\begin{equation*}
p_2(x,t) = - 16 (x^{4} - 3 t^{2}).
\end{equation*}
There are only two real roots of $p_2$ in $x$ found from $x^2 = \sqrt{3}~|t|$ and shown by the red curves on the right panel. The maximum of the solution surface 
is given by 
\begin{equation*}
\max\limits_{(x,t) \in \mathbb{R}^{2}}(|u_{2}(x,t)|^2 + |v_{2}(x,t)|^2)=|u_{2}(0,0)|^2 + |v_{2}(0,0)|^2=2\cdot 4^{2}=32,
\end{equation*}   
which defines the magnification factor $2^2$ compared to the algebraic soliton (\ref{1-soliton-alg-expr}).

\begin{figure}[htb!]
	\centering
	\includegraphics[width=0.5\columnwidth,height=6cm]{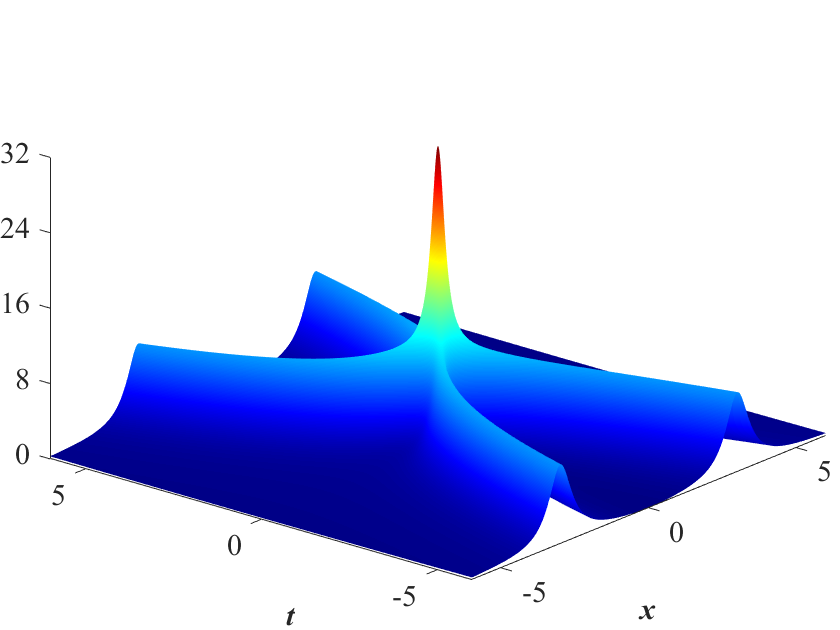}
	\includegraphics[width=0.45\columnwidth,height=6cm]{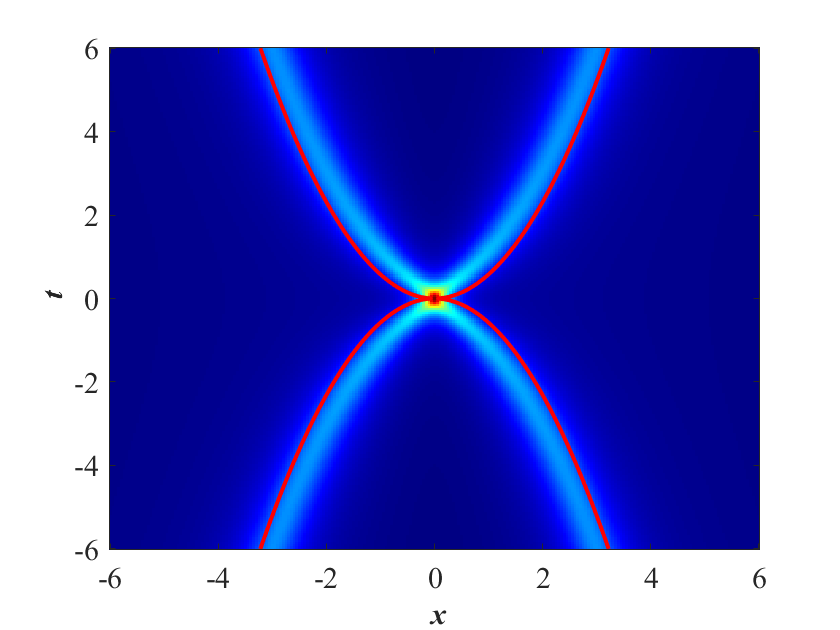}
	\caption{Double algebraic solitons given by (\ref{rat-2}): the solution surface for $|u(x,t)|^2 + |v(x,t)|^2$ (left) and the density plot (right) together with the roots of the principal polynomial $p_2(x,t)$ (red curves).}
	\label{FIG:1}
\end{figure}

The polynomial $P_2$ can be rewritten in the equivalent form:
\[
P_2(x,t) =-\left(2 x +{\rm i} \right)^{4}-12 \left(2 x +{\rm i} \right)^{2}+8 {\rm i} \left(2 x +{\rm i} \right)+48 t^{2},
\]
which agrees with the representation (\ref{polynom-dom}). 
Since $p_2(\cdot,t)$ has only two real roots,
one root in $\mathbb{C}^+$, and one root in $\mathbb{C}^-$, 
this representation implies that $P_2(\cdot,t)$ has $\frac{N(N-1)}{2} = 1$ root in $\mathbb{C}^+$ and $\frac{N(N+1)}{2} = 1$ root in $\mathbb{C}^-$ for large $|t|$. 
It was shown in \cite{Han-2024} that $P_2(\cdot,t)$ admits no roots on $\mathbb{R}$ for all $t \in \mathbb{R}$.

\subsection{Triple algebraic soliton}
\label{sec2}

We use the double-Wronskian solutions (\ref{Wronskian}) with $N = 3$ and 
define in (\ref{rel-1}):
\begin{equation*}
A=\begin{pmatrix}
-1 & 0 & 0 & 0 & 0 & 0 \\
1 & -1 & 0 & 0 & 0 & 0 \\
0 & 1 & -1 & 0 & 0 & 0 \\
0 & 0 & 1 & -1 & 0 & 0 \\
0 & 0 & 0 & 1 & -1 & 0 \\
0 & 0 & 0 & 0 & 1 & -1
\end{pmatrix}
\quad \mbox{\rm and} \quad 
S=\begin{pmatrix}
1 & 0 & 0 & 0 & 0 & 0 \\
-\frac{1}{2} & 1 & 0 & 0 & 0 & 0 \\
-\frac{1}{8} & -\frac{1}{2} & 1 & 0 & 0 & 0 \\
-\frac{1}{16} & -\frac{1}{8} & -\frac{1}{2} & 1 & 0 & 0 \\
-\frac{5}{128} & -\frac{1}{16} & -\frac{1}{8} & -\frac{1}{2} & 1 & 0 \\
-\frac{7}{256} & -\frac{5}{128} & -\frac{1}{16} & -\frac{1}{8} & -\frac{1}{2} & 1
\end{pmatrix}.
\end{equation*}
Since $C = (-{\rm i})^{3}/|S| = {\rm i}$, we obtain from (\ref{Wronskian}):
\begin{eqnarray*}
\begin{cases}
f &= \begin{vmatrix} \phi', \phi'', \phi''' ; \psi, \psi', \psi'' \end{vmatrix}, \\
g &= \begin{vmatrix} \phi, \phi', \phi'', \phi''' ; \psi', \psi'' \end{vmatrix}, \\
h &= \begin{vmatrix} \phi, \phi', \phi'', \phi''' ; \psi, \psi' \end{vmatrix},
\end{cases}
\end{eqnarray*}
which generates the exact solution due to (\ref{MTM-bil}) and (\ref{transformation}):
\begin{equation}
\label{rat-3}
u_{3}(x,t) = \frac{Q_3(x,t)}{\bar{P}_3(x,t)} {\mathrm e}^{-{\rm i} t}, \quad v_{3}(x,t) = \frac{R_3(x,t)}{P_3(x,t)} {\mathrm e}^{-{\rm i} t},
\end{equation}
where
{\small \begin{align*}
P_3(x,t) &= 
512 x^{9} + 2304 {\rm i} \,x^{8} + 4608 x^{7} + 23808 {\rm i}\,x^{6} + \left(-9216 t^{2} - 12096\right) x^{5} + \left(-23040 {\rm i} \,t^{2} \right.\\
& \quad\left.+ 21600 {\rm i} \right) x^{4} + \left(69120 t^{2}  -27360 \right) x^{3} + \left(34560 {\rm i}\,t^{2} + 6480 {\rm i}\right) x^{2} + \left(-69120 t^{4}+  \right.\\
& \quad\left. 8640 t^{2}- 2430\right) x - 34560 {\rm i} \,t^{4} - 12960 {\rm i} \,t^{2} - 135 {\rm i},\\
Q_3(x,t) &= -6 {\rm i} \left( 256 x^{8} - 1024 {\rm i} \,x^{7} + \left(-2048 {\rm i} t + 768\right) x^{6} + \left(-5376 {\rm i} - 6144 t \right) x^{5}+ (7680 {\rm i} t \right. \\
& \quad \left. - 7680 t^{2} + 4320) x^{4} + \left(15360 {\rm i} \,t^{2} - 8640 {\rm i}\right) x^{3} + \left(5760 {\rm i} t - 11520 t^{2} - 3600\right) x^{2} +\right. \\
& \quad \left. \left(11520 {\rm i} \,t^{2} + 720 {\rm i}- 5760 t \right) x + 11520 t^{4}  + 1440 t^{2} + 1440 {\rm i} t - 135 \right), \\
R_3(x,t) &= 6 {\rm i} \left(256 x^{8} + 1024 {\rm i} \,x^{7} + \left(-2048 {\rm i} t + 768\right) x^{6} + \left(5376 {\rm i} + 6144 t \right) x^{5} + (7680 {\rm i} t \right. \\
& \quad \left. - 7680 t^{2} + 4320) x^{4} + \left(-15360 {\rm i} \,t^{2} + 8640 {\rm i} \right) x^{3} + \left(5760 {\rm i} t - 11520 t^{2} - 3600\right) x^{2} \right. \\
& \quad \left. + \left(-11520 {\rm i} \,t^{2} - 720{\rm i} + 5760 t \right) x + 11520 t^{4} + 1440 t^{2} + 1440 {\rm i} t - 135\right).
\end{align*}} 

\begin{figure}[htb!]
	\centering
	\includegraphics[width=0.5\columnwidth,height=6cm]{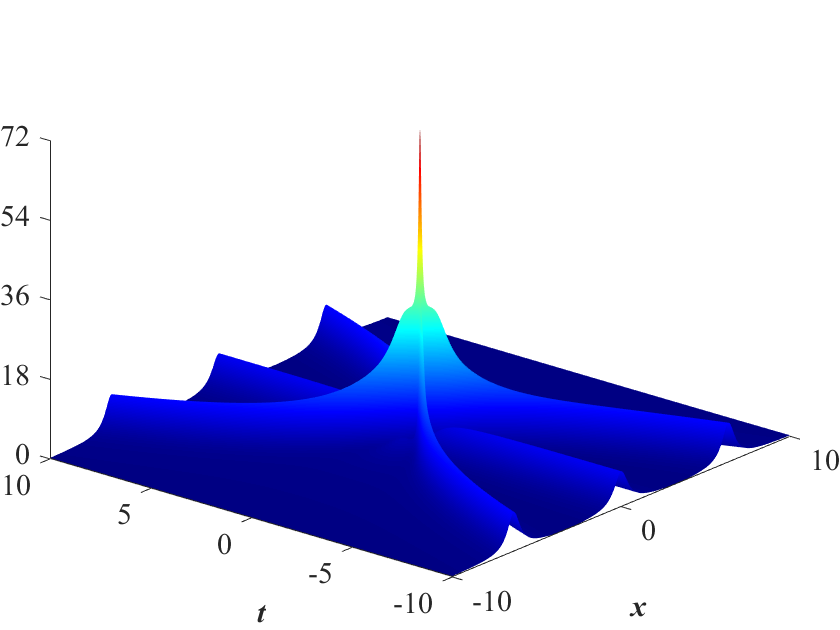}
	\includegraphics[width=0.45\columnwidth,height=6cm]{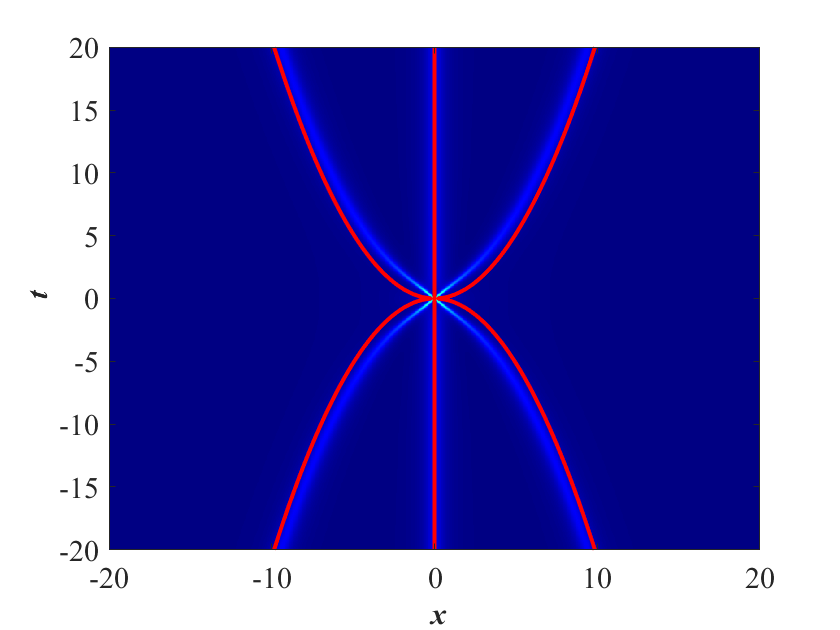}
	\caption{Triple algebraic solitons given by (\ref{rat-3}): the solution surface for $|u(x,t)|^2 + |v(x,t)|^2$ (left) and the density plot (right) together with the roots of the principal polynomial $p_3(x,t)$ (red curves).}
	\label{FIG:2}
\end{figure}

Figure \ref{FIG:2} shows the solution surface (left) and the density plot 
(right). The principal part of the polynomial $P_3$ is given by 
\begin{equation*}
p_3(x,t) = 512 x (x^{8}-18 x^{4} t^{2}-135 t^{4}).
\end{equation*}
The roots of $p_3$ in $x$ are given by 
\begin{equation*}
x^4=9+6 \sqrt{6}~ t^2,\quad x^4=9-6 \sqrt{6} ~t^2, \quad x=0.
\end{equation*}  
Since $9-6 \sqrt{6}<0$, there exist only three real roots found from 
$x = 0$, $x^{2} = \sqrt{9+6 \sqrt{6}} |t|$ and shown by the red curves on the right panel. The maximum of the solution surface is given by 
\begin{equation*}
\max_{(x,t) \in \mathbb{R}^{2}} (|u_{3}(x,t)|^2 + |v_{3}(x,t)|^2) =|u_{3}(0,0)|^2 + |v_{3}(0,0)|^2=2\cdot 6^{2}=72,
\end{equation*}   
which defines the magnification factor $3^2$ compared to the algebraic soliton (\ref{1-soliton-alg-expr}).

The polynomial $P_3$ can be rewritten in the equivalent form:
\begin{align*}
P_3(x,t) &= \left(2 x +{\rm i} \right)^{9}+72 \left(2 x +{\rm i} \right)^{7}-48{\rm i} \left(2 x +{\rm i} \right)^{6}+\left(-288 t^{2}+720\right) \left(2 x +{\rm i} \right)^{5}\\
& \quad -576 {\rm i} \left(2 x +{\rm i}\right)^{4} +5760 t^{2} \left(2 x +{\rm i} \right)^{3}+\left(-11520 {\rm i} \,t^{2}+4608 {\rm i} \right) \left(2 x +{\rm i} \right)^{2} \\
& \quad +\left(-34560 t^{4}+6912\right)\left(2 x +{\rm i} \right) -2560 {\rm i} -18432 {\rm i} \,t^{2},
\end{align*}
which agrees with the representation (\ref{polynom-dom}). 
Since $p_3(\cdot,t)$ has only three real roots,
three roots in $\mathbb{C}^+$, and three roots in $\mathbb{C}^-$, 
this representation implies that $P_3(\cdot,t)$ has $\frac{N(N-1)}{2} = 3$ roots in $\mathbb{C}^+$ and $\frac{N(N+1)}{2} = 6$ roots in $\mathbb{C}^-$ for large $|t|$. 
We also confirm numerically that $P_3(\cdot,t)$ admit no roots on $\mathbb{R}$ for all $t \in \mathbb{R}$.

\subsection{Quadruple algebraic soliton}
\label{sec3}

The rational solution for $N =4$ can be written explicitly:
\begin{equation}
\label{rat-4}
u_{4}(x,t) = \frac{Q_4(x,t)}{\bar{P}_4(x,t)} {\mathrm e}^{-{\rm i} t}, \quad v_{4}(x,t) = \frac{R_4(x,t)}{P_4(x,t)} {\mathrm e}^{-{\rm i} t},
\end{equation}
where
{\small 
\begin{align*}
P_4(x,t) &= 
\begin{aligned}[t]
&65536 x^{16}+524288 {\rm i} \,x^{15}+1966080 x^{14}+21626880 {\rm i} \,x^{13}+\left(-3932160 t^{2}-8601600\right) \\
&x^{12}+\left(-23592960 {\rm i} \,t^{2}+207912960 {\rm i}\right) x^{11}+\left(29491200 t^{2}-268001280\right) x^{10}+\\
&\left(-88473600 {\rm i} \,t^{2}+628531200 {\rm i} \right) x^{9}+\left(-29491200 t^{4}-873676800 t^{2}-571968000\right) \\
&x^{8}+\left(-117964800 {\rm i} \,t^{4}-3140812800 {\rm i}\,t^{2}+1251072000 {\rm i} \right) x^{7}+\left(3096576000 t^{4}-\right.\\
&\left.1935360000 t^{2}-2484518400\right) x^{6}+\left(6812467200{\rm i} \,t^{4}-5806080000 {\rm i}\,t^{2}-631411200 {\rm i}\right) \\
&x^{5}+\left(-2064384000 t^{6}+3483648000 t^{4}+2685312000 t^{2}-358344000\right) x^{4}+\left(-\right.\\
&\left. 4128768000 {\rm i}\,t^{6}+6967296000 {\rm i}\,t^{4}+3048192000 {\rm i} \,t^{2}-281232000 {\rm i} \right) x^{3}+\left(-3096576000\right.\\
&\left. t^{6}-580608000 t^{4}-762048000 t^{2}-34020000\right) x^{2}+\left(-3096576000 {\rm i} \,t^{6}-2903040000 \right.\\
&\left.{\rm i} \,t^{4}+108864000{\rm i} \,t^{2}-6804000 {\rm i} \right) x +1548288000 t^{8}+387072000 t^{6}+762048000 t^{4}\\
&+68040000 t^{2}+212625,
\end{aligned} 
\end{align*}
\begin{align*}
Q_4(x,t) &= 
\begin{aligned}[t]
& 8 (-32768 {\rm i}\,x^{15}-245760 x^{14}+\left(-614400 {\rm i} -491520 t \right) x^{13}+\left(3194880 {\rm i} t -7311360\right)\\
&x^{12}+\left(4423680 {\rm i} \,t^{2} +1382400 {\rm i} -737280 t \right) x^{11}+\left(38707200 {\rm i} t +24330240 t^{2}-\right.\\
& \left.65940480\right)x^{10}+\left(-16588800 {\rm i} \,t^{2}+17203200 t^{3}+138355200 {\rm i} -37324800 t \right) x^{9}\\
&+\left(-77414400 {\rm i} \,t^{3}+402969600 {\rm i} t +118886400 t^{2}+19526400\right) x^{8}+\left(-22118400 {\rm i} \,t^{4}\right.\\
& \left.-505958400 {\rm i} \,t^{2}+154828800 t^{3}+263952000 {\rm i} +31795200 t \right) x^{7}+\left(-799948800 {\rm i} \,t^{3}\right.\\
& \left. -77414400 t^{4} +846720000 {\rm i} t -493516800 t^{2}+89208000\right) x^{6}+\left(580608000 {\rm i} \,t^{4}\right.\\
&\left. -154828800 t^{5}-217728000 {\rm i} \,t^{2}-406425600 t^{3}+298015200 {\rm i} +988848000 t \right) x^{5}\\
&+\left(387072000 {\rm i} \,t^{5}-725760000 {\rm i} \,t^{3}+ 870912000 t^{4}-258552000 {\rm i} t +108864000 t^{2}\right.\\
&\left.-91854000\right) x^{4}+\left(-258048000 {\rm i} \,t^{6}-435456000 {\rm i} \,t^{4}+ 1161216000 t^{5}-789264000 \right.\\
&\left.i \,t^{2}-919296000 t^{3}+92421000 {\rm i} +40824000 t \right) x^{3}+\left(-580608000 {\rm i} \,t^{5}-387072000 t^{6}\right.\\
&\left. -217728000 {\rm i} \,t^{3}+217728000 t^{4}-74844000 {\rm i} t +231336000 t^{2}+19561500\right) x^{2}+\\
& (193536000 {\rm i} \,t^{6}-774144000 t^{7}+108864000 {\rm i} \,t^{4}+145152000 t^{5}-142884000 {\rm i}\,t^{2}\\
& -81648000 t^{3}-2126250 i +28917000 t ) x +387072000 {\rm i} \,t^{7} -96768000 t^{6}+\\
&217728000 {\rm i} \,t^{5}-272160000 t^{4}+4536000 i \,t^{3}-4252500 i t -10206000 t^{2}+212625),
\end{aligned} 
\end{align*}
\begin{align*}\nonumber
R_4(x,t) &= 8  (32768 {\rm i} \,x^{15}-245760 x^{14}+(614400 {\rm i}+491520 t ) x^{13}+(3194880 {\rm i} t -7311360)\\
&\quad x^{12}+(-4423680 {\rm i} \,t^{2}-1382400 {\rm i} +737280 t ) x^{11}+(38707200 {\rm i} t +24330240 t^{2}- \\
&\quad 65940480) x^{10}+(16588800 {\rm i} \,t^{2}-17203200 t^{3}- 138355200 i+37324800 t ) x^{9}+\\
&\quad (-77414400 {\rm i} \,t^{3}+402969600 {\rm i} t +118886400 t^{2}+19526400) x^{8}+(22118400 {\rm i} \,t^{4} \\
&\quad +505958400 {\rm i}\,t^{2}-154828800 t^{3}-263952000 {\rm i} -31795200 t ) x^{7}+(-799948800 {\rm i} \,t^{3}  \\
&\quad -77414400 t^{4}+846720000 {\rm i}t-493516800 t^{2}+89208000) x^{6}+(-580608000 {\rm i} \,t^{4}+\\
&\quad 154828800 t^{5}+217728000 {\rm i}\,t^{2}+406425600 t^{3}-298015200 i-988848000 t ) x^{5}\\
& \quad  +(387072000 {\rm i} \,t^{5}-725760000 {\rm i} \,t^{3}+870912000 t^{4}-258552000 {\rm i} t +108864000 t^{2}\\
&\quad  -91854000)x^{4}+(258048000 {\rm i} \,t^{6}+435456000 {\rm i} \,t^{4}-1161216000 t^{5}+789264000 {\rm i} \,t^{2} \\
&\quad +919296000 t^{3}-92421000 {\rm i} -40824000 ) x^{3}+(-580608000 {\rm i} \,t^{5}-387072000 t^{6}-\\
&\quad 217728000 {\rm i} \,t^{3}+217728000 t^{4}-74844000 {\rm i} t +231336000 t^{2}+19561500) x^{2}+\\
&\quad (-193536000 {\rm i} \,t^{6}+774144000 t^{7}-108864000 {\rm i} \,t^{4}-145152000 t^{5}+142884000 {\rm i} \,t^{2}\\
&\quad +81648000 t^{3}+2126250 {\rm i}-28917000 t ) x +387072000 {\rm i} \,t^{7}+217728000 {\rm i} \,t^{5}\\
&\quad  -96768000 t^{6}+4536000 {\rm i} \,t^{3}-272160000 t^{4}-4252500 {\rm i} t -10206000 t^{2}+212625).
 \end{align*}}

\begin{figure}[htb!]
	\centering
	\includegraphics[width=0.5\columnwidth,height=6cm]{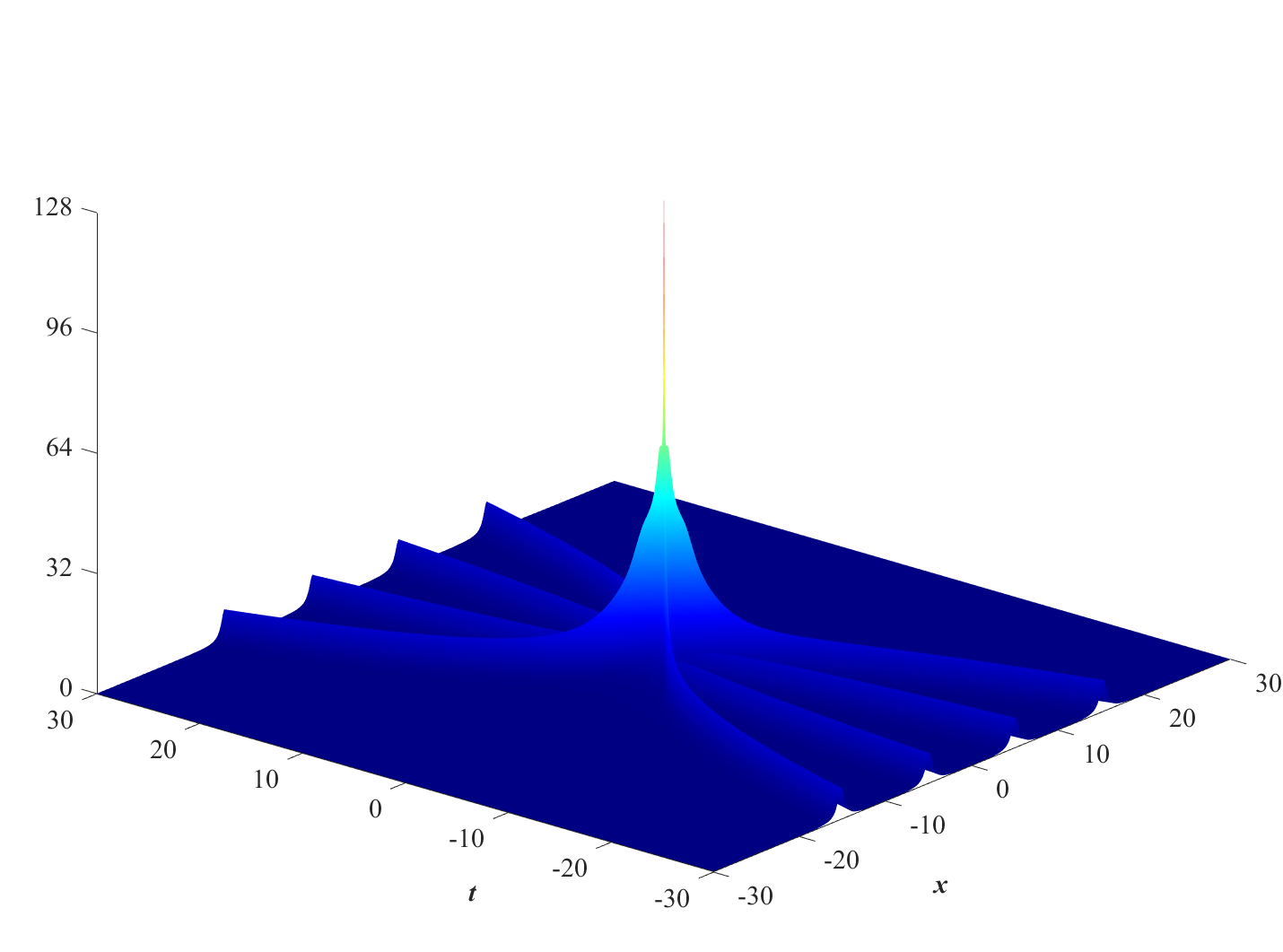}
	\includegraphics[width=0.45\columnwidth,height=6cm]{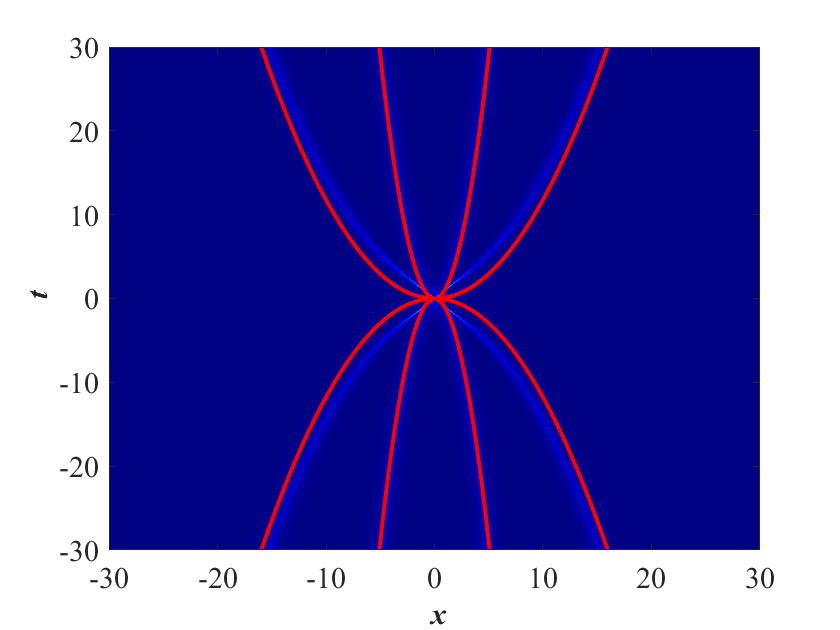}
	\caption{Quadruple algebraic solitons given by (\ref{rat-4}): the solution surface for $|u(x,t)|^2 + |v(x,t)|^2$ (left) and the density plot (right) together with the roots of the principal polynomial $p_4(x,t)$ (red curves).}
	\label{FIG:3}
\end{figure}

Figure \ref{FIG:3} shows the solution surface (left) and the density plot 
(right). The principal part of the polynomial $P_4$ is given by 
\begin{equation*}
p_4(x,t) = 65536 (x^{16}-60 t^{2} x^{12}-450 t^{4} x^{8}-31500 t^{6} x^{4}+23625 t^{8}).
\end{equation*}
The roots of $p_4$ in $x$ are given by 
\begin{align*}
x^4 = \left( 15+2 \sqrt{15}\, s_1 \lambda + 2 \sqrt{15}\, s_2  \sqrt{\frac{10 \sqrt{15}s_1+15  \lambda - \lambda^{3}}{\lambda}} \right) t^2,
\end{align*}
where $s_1 = \pm 1$, $s_2 = \pm 1$ are two independent sign combinations and 
$$
\lambda=\sqrt{\frac{\left(50+10 \sqrt{35}\right)^{\frac{2}{3}}+5 \left(50+10 \sqrt{35}\right)^{\frac{1}{3}}-10}{\left(50+10 \sqrt{35}\right)^{\frac{1}{3}}}}.
$$ 
There exist only four real roots shown by the red curves on the right panel. The maximum of the solution surface is given by 
\begin{equation*}
\max\limits_{(x,t) \in \mathbb{R}^{2}} \left(|u_{4}(x,t)|^2 + |v_{4}(x,t)|^2 \right) =|u_{4}(0,0)|^2 + |v_{4}(0,0)|^2=2\cdot 8^{2}=128,
\end{equation*}   
which defines the magnification factor $4^2$ compared to the algebraic soliton (\ref{1-soliton-alg-expr}).

\begin{figure}[htb!]
	\centering
	\includegraphics[width=0.5\columnwidth,height=6cm]{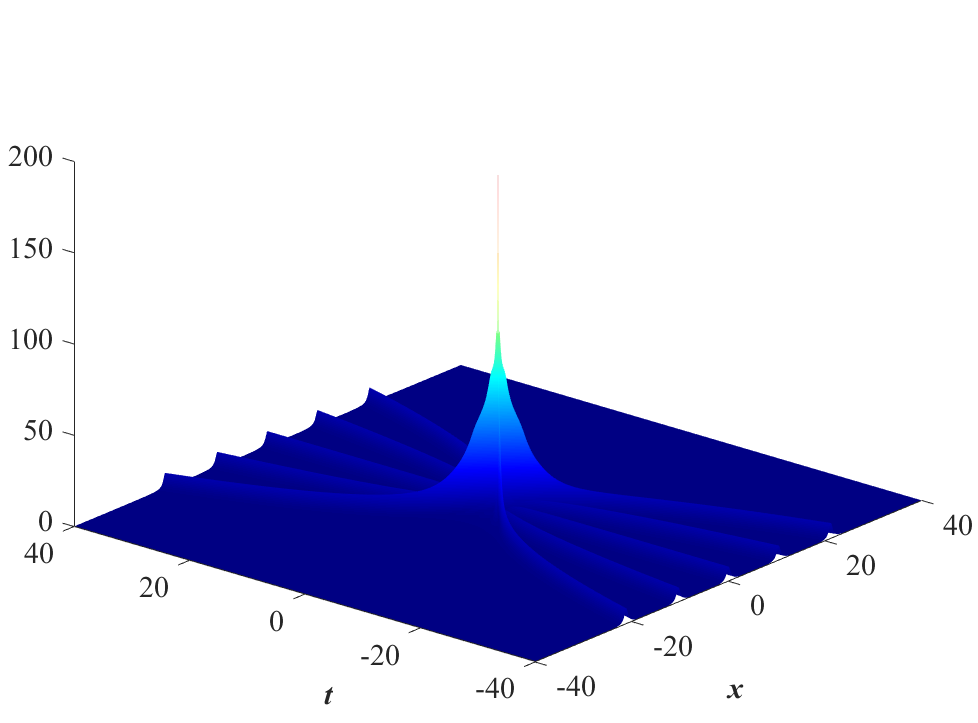}
	\includegraphics[width=0.45\columnwidth,height=6cm]{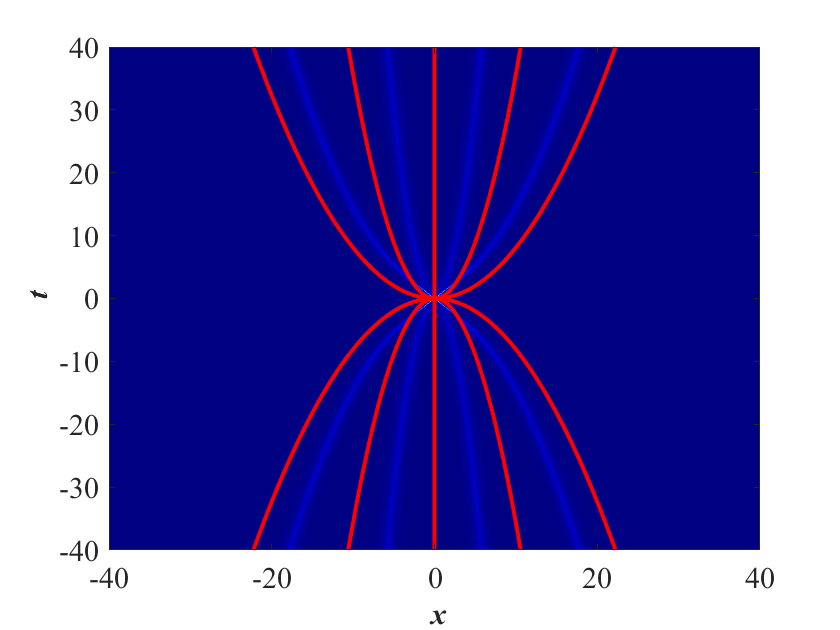}
	\caption{Quintuple algebraic solitons: the solution surface for $|u(x,t)|^2 + |v(x,t)|^2$ (left) and the density plot (right) together with the roots of the principal polynomial $p_5(x,t)$ (red curves).}
	\label{FIG:4}
\end{figure}

\begin{figure}[htb!]
	\centering
	\includegraphics[width=0.5\columnwidth,height=6cm]{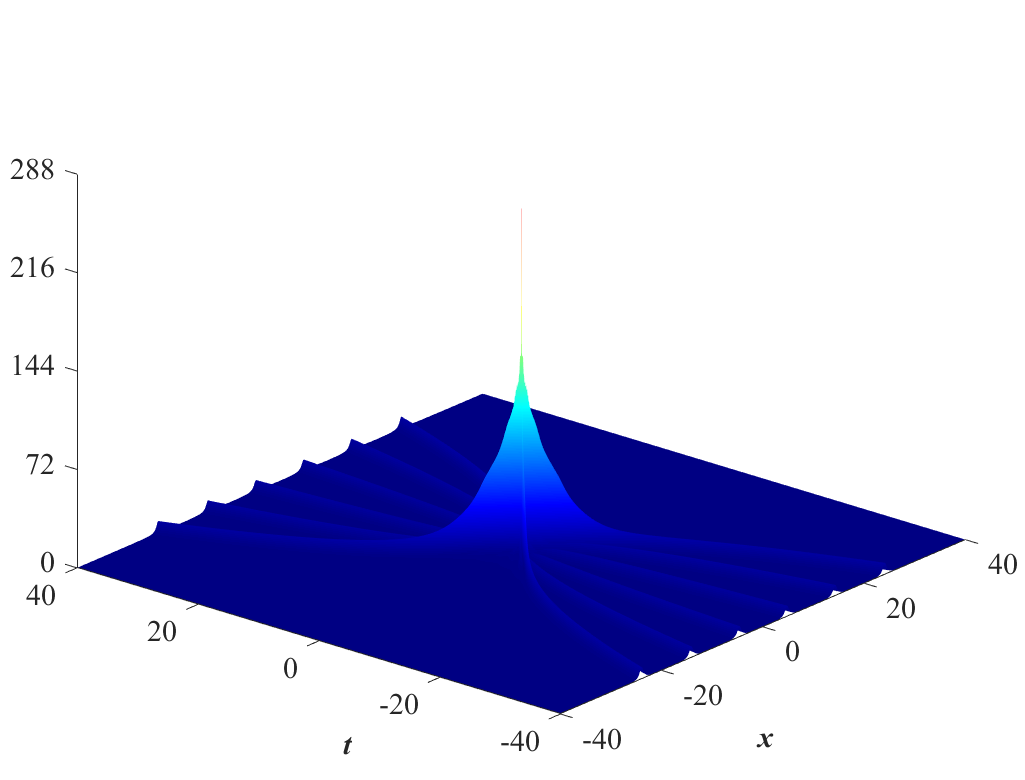}
	\includegraphics[width=0.45\columnwidth,height=6cm]{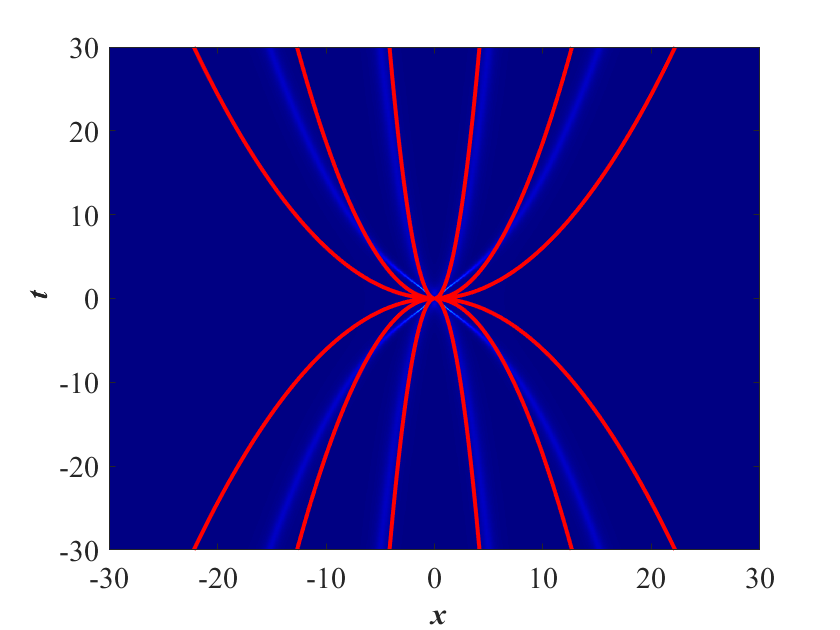}
	\caption{Sextuple algebraic solitons: the solution surface for $|u(x,t)|^2 + |v(x,t)|^2$ (left) and the density plot (right) together with the roots of the principal polynomial $p_6(x,t)$ (red curves).}
	\label{FIG:5}
\end{figure}

\subsection{Quintuple and sextuple algebraic solitons}
\label{sec4}

The rational solutions for $N = 5$ and $N = 6$ are too long to be written explicitly. Figures \ref{FIG:4} and \ref{FIG:5} show the solution surfaces (left) and the density plots (right). We only give the principal part of the polynomials $P_5$ and $P_6$ in the explicit form:
\begin{align*}
p_5(x,t) =& 33554432 x \left( x^{24}-150 x^{20} t^{2}+1575 x^{16} t^{4}-220500 x^{12} t^{6}-15710625 x^{8} t^{8} \right.\\
& \left.+104186250 x^{4} t^{10}+260465625 t^{12}\right)
\end{align*}
and 
\begin{align*}
p_6(x,t) =& 68719476736( x^{36}-315 t^{2} x^{32}+18900 t^{4} x^{28}-1190700 t^{6} x^{24}-120856050 t^{8} x^{20}\\
& -12639875850 t^{10} x^{16}+396449518500 t^{12} x^{12}+1299619282500 t^{14} x^{8} \\
& + 34115006165625 t^{16} x^{4}-11371668721875 t^{18}).
\end{align*}
Roots of $p_5$ and $p_6$ in $x$ are found numerically and shown on the right panels.
There exist only five real roots of $p_5$ and six real roots of $p_6$. The maximum values of the solution surfaces are given by 
\begin{equation*}
\max\limits_{(x,t) \in \mathbb{R}^{2}} \left( |u_{5}(x,t)|^2 + |v_{5}(x,t)|^2 \right) =|u_{5}(0,0)|^2 + |v_{5}(0,0)|^2=2\cdot 10^{2}=200
\end{equation*} 
and 
\begin{equation*}
\max\limits_{(x,t) \in \mathbb{R}^{2}}|u_{6}(x,t)|^2 + |v_{6}(x,t)|^2=|u_{6}(0,0)|^2 + |v_{6}(0,0)|^2=2\cdot 12^{2}=288,
\end{equation*}   
which define the magnification factors $5^2$ and $6^2$, respectively, compared to the algebraic soliton (\ref{1-soliton-alg-expr}).

{\bf Acknowledgement.} The authors thank Jiaqi Han for the productive collaboration during the initial stage of this project. The work of Z. Zhao and C. He was conducted during their PhD studies while visiting McMaster University with the financial support from the China Scholarship Council. The work of C. He is partially supported by the National Natural Science Foundation of China under grant No. 12431008. The work of B. Feng is partially supported by the U.S. Department of Defense (DoD) and Air Force for Scientific Research (AFOSR) under grant No. W911NF2010276. The work of D. E. Pelinovsky is supported by the NSERC Discovery grant.

\appendix
\renewcommand{\thesubsection}{\Alph{section}.\arabic{subsection}} 
\numberwithin{equation}{section} 
\renewcommand{\theequation}{\Alph{section}\arabic{equation}}   
\begingroup
\let\oldsubsection\subsection
\renewcommand{\subsection}[1]{%
\oldsubsection{#1}\par\noindent}
\let\oldsubsubsection\subsubsection
\renewcommand{\subsubsection}[1]{%
	\oldsubsubsection{#1}\par\noindent}

\section{Proof of (\ref{det-expression-2})}
\label{app-A}

We use elementary column operations and transform $|M_N(1)|$ to the form:

\begin{align*}
|M_N(1)| &=
	\left|
	\begin{array}{c c c c : c c c c}
	1 & 0 & 0  & \dots & 1 & 0 & 0 &  \dots \\
	1  & 1 & 0 &  \dots  & -1 & 1 & 0 &  \dots \\
	\frac{1}{2!} & 1  & 1 &  \dots  & \frac{1}{2!} & -1 & 1 & \dots \\
	\frac{1}{3!} & \frac{1}{2!} & 1 &  \dots  & -\frac{1}{3!} & \frac{1}{2!} & -1 &  \dots \\
	\vdots & \vdots & \vdots & \ddots & \vdots & \vdots & \vdots & \ddots \\
	\frac{1}{(2N-1)!} & \frac{1}{(2N-2)!} & \frac{1}{(2N-3)!} &  \dots  & -\frac{1}{(2N-1)!} & \frac{1}{(2N-2)!} & -\frac{1}{(2N-3)!} &  \dots 
	\end{array}
	\right|\\
	&=2^N (-1)^N	\left|\begin{array}{c c c c : c c c c}
	1 & 0 & 0 &  \dots & 0 & 0 & 0 &  \dots \\
	0  & 1 & 0 &  \dots  & 1 & 0 & 0 & \dots \\
	\frac{1}{2!} & 0  & 1 & \dots  & 0 & 1 & 0 &  \dots \\
	0 & \frac{1}{2!} & 0 &  \dots  & \frac{1}{3!} & 0 & 1 &  \dots \\
	\vdots & \vdots &  \vdots & \ddots & \vdots & \vdots & \vdots &  \ddots \\
	0 & \frac{1}{(2N-2)!} & 0 & \dots  & \frac{1}{(2N-1)!} & 0 & \frac{1}{(2N-3)!} & \dots 
	\end{array}
	\right|,
\end{align*}
where we first added the $(N+j)$-th column to the $j$-th column for $1 \leq j \leq N$, then extracted the factor of $2$ from the first $N$ columns, then subtracted the updated $j$-th column from the $(N+j)$-th column for $1 \leq j \leq N$, and finally, multiplied the last $N$ columns by the negative signs. To further simplify the determinant denoted by $A_N$, we expand it trivially along the first row and reduce it to the $(2N-1)\times(2N-1)$ determinant:
\begin{align}
\label{A}
A_N = \left| 	\begin{array}{c c c c : c c c c c}
1 & 0 & 0 & \dots  & 1 & 0 & 0 & 0 & \dots \\
0  & 1 & 0 & \dots  & 0 & 1 & 0 & 0 & \dots \\
\frac{1}{2!} & 0 & 1 & \dots  & \frac{1}{3!} & 0 & 1 & 0 & \dots \\
0 & \frac{1}{2!} & 0 & \dots  & 0 & \frac{1}{3!} & 0 & 1 & \dots \\
\frac{1}{4!} & 0 & \frac{1}{2!} & \dots  & \frac{1}{5!} & 0 & \frac{1}{3!} & 0 & \dots \\
\vdots & \vdots & \vdots & \ddots & \vdots & \vdots & \vdots & \ddots & \vdots \\
\frac{1}{(2N-2)!} & 0 & \frac{1}{(2N-4)!} & \dots & \frac{1}{(2N-1)!} & 0 & \frac{1}{(2N-3)!} & 0 & \dots
	\end{array}\right|.
\end{align}
The following proposition gives the proof of (\ref{det-expression-2}).

\begin{prop}
	\label{pro1}
For any $N \in \mathbb{N}$, we have 
	\begin{equation}
	\label{eq:A1} 
		A_{N} =\frac{1}{1^{2N-1} 3^{2N-3} 5^{2N-5} 7^{2N-7} \ldots (2N-3)^3 (2N-1)}. 
	\end{equation}
\end{prop}

\begin{proof}
When $N=1$, we obtain $A_1 = 1$, which agrees with (\ref{eq:A1}). When $N=2$, we obtain
\begin{equation*}
	A_{2} =	\begin{vmatrix}
		1 & 1 & 0\\
		0 & 0 & 1\\
		\frac{1}{2!} & \frac{1}{3!} & 0
	\end{vmatrix}
	\xlongequal{\,C_1-C_2\,}
	\begin{vmatrix}
		0 & 1 & 0\\
		0 & 0 & 1\\
		\frac{1}{3} & \frac{1}{3!} & 0
	\end{vmatrix}	
	=\frac{1}{3} \begin{vmatrix}
		0 & 1 & 0\\
		0 & 0 & 1\\
		1 & \frac{1}{3!} & 0
	\end{vmatrix}	
=\frac{1}{3},
\end{equation*}
which agrees with (\ref{eq:A1}). When $N=3$, we obtain
\begin{align*}
	A_{3} &=
	\begin{vmatrix}
		1&0&1&0&0\\
		0&1&0&1&0\\
		\frac{1}{2!}&0&\frac{1}{3!}&0&1\\
		0&\frac{1}{2!}&0&\frac{1}{3!}&0\\
		\frac{1}{4!}&0&\frac{1}{5!}&0&\frac{1}{3!}
	\end{vmatrix}
	\xlongequal[C_2-C_4]{C_1-C_3}
		\begin{vmatrix}
	0&0&1&0&0\\
	0&0&0&1&0\\
	\frac{1}{3}&0&\frac{1}{3!}&0&1\\
	0&\frac{1}{3}&0&\frac{1}{3!}&0\\
	\frac{1}{5 \cdot 3!}&0&\frac{1}{5!}&0&\frac{1}{3!}
	\end{vmatrix} 
	=	
	\begin{vmatrix}
		\frac{1}{3}&0&1\\
		0&\frac{1}{3}&0\\
		\frac{1}{5 \cdot 3!}&0&\frac{1}{3!}
	\end{vmatrix}=
	\frac{1}{3^2}
	\begin{vmatrix}
		1&0&1\\
		0&1&0\\
		\frac{1}{5 \cdot 2}&0&\frac{1}{3!}
	\end{vmatrix}\nonumber\\
	&\xlongequal{C_3-C_1}\frac{1}{3^2}
	\begin{vmatrix}
		1&0&0\\
		0&1&0\\
		\frac{1}{5 \cdot 2}&0&\frac{1}{5 \cdot 3}
	\end{vmatrix}
	= 	\frac{1}{3^2\cdot (3 \cdot 5)}
	\begin{vmatrix}
1&0&0\\
0&1&0\\
\frac{1}{5 \cdot 2}&0&1
\end{vmatrix}
	= \frac{1}{3^3 \cdot 5^1},
\end{align*}
which agrees with (\ref{eq:A1}). This sets up an algorithm, where we 
successively eliminate two columns from the left and right blocks of the 
determinant $A_N$ after elementary column operations and factorizations. 
This algorithm is described next according to the following steps. For illustration, we use transformation of $A_4$ for $N = 4$.

\underline{\bf Step 1.} The determinant $A_N$ in (\ref{A}) contains two blocks of $(N-1)$ left columns and $N$ right columns. We subtract the first $(N-1)$ right columns from $(N-1)$ left columns, entries of which change to 
\begin{equation}
\label{step-1}
	\frac{1}{(2 \ell)!} - \frac{1}{(2\ell+1)!} = \frac{1}{(2\ell+1) \cdot (2\ell-1)!}, \qquad \ell=1,2,\dots,N-1,
\end{equation}
where $\ell$ corresponds to the $(2\ell)$-th lower diagonal of the left $(2N-1) \times (N-1)$ block. Since $\frac{1}{3}$ is a common factor in each left column, we extract the numerical factor  $\frac{1}{3^{N-1}}$ in front of the determinant and modify the entries 
of $(N-1)$ left columns  to 
\begin{equation*}
\frac{3}{(2\ell + 1) \cdot (2\ell-1)!}, \qquad \ell = 1,2,\dots,N-1.
\end{equation*}
For $N = 4$, this step yields 
\begin{equation*}
A_4 = \begin{vmatrix}
0&0&0&1&0&0&0\\
0&0&0&0&1&0&0\\
\frac{1}{3}&0&0&\frac{1}{3!}&0&1&0\\
0&\frac{1}{3}&0&0&\frac{1}{3!}&0&1\\
\frac{1}{5 \cdot 3!}&0&\frac{1}{3} & \frac{1}{5!}&0&\frac{1}{3!}&0 \\
0&\frac{1}{5 \cdot 3!}&0&0&\frac{1}{5!}&0& \frac{1}{3!} \\
\frac{1}{7 \cdot 5!}&0&\frac{1}{5 \cdot 3!} & \frac{1}{7!}&0&\frac{1}{5!}&0 
\end{vmatrix} 
= \frac{1}{3^3} \begin{vmatrix}
0&0&0&1&0&0&0\\
0&0&0&0&1&0&0\\
1&0&0&\frac{1}{3!}&0&1&0\\
0&1&0&0&\frac{1}{3!}&0&1\\
\frac{1}{5 \cdot 2}&0& 1 & \frac{1}{5!}&0&\frac{1}{3!}&0 \\
0&\frac{1}{5 \cdot 2}&0&0&\frac{1}{5!}&0& \frac{1}{3!} \\
\frac{3}{7 \cdot 5!}&0&\frac{1}{5 \cdot 2} & \frac{1}{7!}&0&\frac{1}{5!}&0 
\end{vmatrix} 
\end{equation*}

\underline{\bf Step 2.} We eliminate the first two rows and the first two right columns. This does not change the sign of the determinant, independently whether $N$ is even or odd. After the two-column elimination, we obtain a determinant which contains two blocks of $(N-1)$ left columns and $(N-2)$ right columns. We subtract the first $(N-2)$ left columns 
from $(N-2)$ right columns, entries of which change to 
\begin{equation}
\label{step-2}
\frac{1}{(2 \ell-1)!} - \frac{3}{(2\ell+1) \cdot (2\ell-1)!} = \frac{1}{(2\ell+1) \cdot (2\ell-1) \cdot (2\ell-3)!}, \qquad \ell = 2,\dots,N-1,
\end{equation}
where $\ell$ corresponds to the $2(\ell-1)$-th lower diagonal of the right 
$(2N-3) \times (N-2)$ block. Since $\frac{1}{3 \cdot 5}$ is a common factor in each right column, we extract the numerical factor  $\frac{1}{(3 \cdot 5)^{N-2}}$ in front of the determinant and modify the entries of $(N-2)$ right columns to 
\begin{equation*}
\frac{3 \cdot 5}{(2\ell+1) \cdot (2\ell-1) \cdot (2\ell-3)!}, \qquad \ell = 2,\dots,N-1.
\end{equation*}
For $N = 4$, this step yields 
\begin{align*}
A_4 &= \frac{1}{3^3} \begin{vmatrix}
1&0&0&1&0\\
0&1&0&0&1\\
\frac{1}{5 \cdot 2}&0& 1 & \frac{1}{3!}&0 \\
0&\frac{1}{5 \cdot 2}&0&0& \frac{1}{3!} \\
\frac{3}{7 \cdot 5!}&0&\frac{1}{5 \cdot 2} & \frac{1}{5!}&0 
\end{vmatrix} 
= \frac{1}{3^3} \begin{vmatrix}
1&0&0&0&0\\
0&1&0&0&0\\
\frac{1}{5 \cdot 2}&0& 1 & \frac{1}{5 \cdot 3}& 0 \\
0&\frac{1}{5 \cdot 2}&0&0& \frac{1}{5 \cdot 3} \\
\frac{3}{7 \cdot 5!}&0&\frac{1}{5 \cdot 2} & \frac{1}{7 \cdot 5 \cdot 3!}& 0 
\end{vmatrix} \\
&= \frac{1}{3^3 \cdot (3 \cdot 5)^2} \begin{vmatrix}
1&0&0&0&0\\
0&1&0&0&0\\
\frac{1}{5 \cdot 2}&0& 1 & 1& 0 \\
0&\frac{1}{5 \cdot 2}&0&0& 1 \\
\frac{3}{7 \cdot 5!}&0&\frac{1}{5 \cdot 2} & \frac{1}{7 \cdot 2}& 0 
\end{vmatrix} 
\end{align*}

\underline{\bf Step 3.} We eliminate the first two rows and the first two left columns and obtain a determinant which contains two blocks of $(N-3)$ left columns and $(N-2)$ right columns. We subtract the first $(N-3)$ right columns from $(N-3)$ left columns, entries of which change to 
\begin{align}
& \frac{3}{(2 \ell - 1) \cdot (2 \ell- 3)!} - \frac{3 \cdot 5}{(2\ell+1) \cdot (2 \ell - 1) \cdot (2\ell-3)!} \notag \\
&= \frac{3}{(2\ell+1) \cdot (2\ell-1) \cdot (2 \ell - 3) \cdot (2\ell-5)!}, \qquad \ell = 3,\dots,N-1, \label{step-3}
\end{align}
where $\ell$ corresponds to the $2(\ell-2)$-th lower diagonal of the left $(2N-5) \times (N-3)$ block. Since $\frac{1}{5 \cdot 7}$ is a common factor in each left column, we extract the numerical factor  $\frac{1}{(5 \cdot 7)^{N-3}}$ in front of the determinant and modify the entries of $(N-3)$ left columns to 
\begin{equation*}
\frac{3 \cdot 5 \cdot 7}{(2\ell+1) \cdot (2\ell-1) \cdot (2 \ell - 3) \cdot (2\ell-5)!}, \qquad \ell = 3,\dots,N-1.
\end{equation*}
For $N = 4$, this step yields 
\begin{align*}
A_4 &= \frac{1}{3^3 \cdot (3 \cdot 5)^2} \begin{vmatrix}
1 & 1& 0 \\
0&0& 1 \\
\frac{1}{5 \cdot 2} & \frac{1}{7 \cdot 2}& 0 
\end{vmatrix} 
= \frac{1}{3^3 \cdot (3 \cdot 5)^2} \begin{vmatrix}
0 & 1& 0 \\
0&0& 1 \\
\frac{1}{5 \cdot 7} & \frac{1}{7 \cdot 2}& 0 
\end{vmatrix} \\
&= \frac{1}{3^3 \cdot (3 \cdot 5)^2 \cdot (5 \cdot 7)^1} \begin{vmatrix}
0 & 1& 0 \\
0&0& 1 \\
1 & \frac{1}{7 \cdot 2}& 0 
\end{vmatrix} 
= \frac{1}{3^5 \cdot 5^3 \cdot 7^1},
\end{align*}
which completes the algorithm for $N = 4$. 

\underline{\bf Step $k$.} The algorithm is continued for any $2 \leq k \leq N-2$ 
by using the same alternation between the first and second blocks. 
After the subtraction of the first $(N-k)$ columns of the larger block from $(N-k)$ columns of the smaller block, entries of which are modified as
\begin{align}
& \frac{3 \cdot 5 \cdot \ldots \cdot (2k-3)}{(2 \ell - 1) \cdot (2\ell - 3) \cdot \ldots \cdot (2 \ell - 2k + 5) \cdot (2 \ell- 2 k +3)!} \notag \\
& - \frac{3 \cdot 5 \cdot \ldots \cdot (2k-1)}{(2\ell+1) \cdot (2 \ell - 1) \cdot \ldots \cdot (2\ell - 2 k + 5) \cdot (2\ell- 2k + 3)!} \notag \\
&= \frac{3 \cdot 5 \cdot \ldots \cdot (2k-3)}{(2\ell+1) \cdot (2\ell-1) \cdot \ldots \cdot  (2 \ell - 2k + 3) \cdot (2\ell - 2k + 1)!}, \qquad \ell = k,\dots,N-1,
\label{induction}
\end{align}
where $\ell$ corresponds to the $2 (\ell-k+1)$-th lower diagonal of the smaller $(2N-2k+1) \times (N-k)$ block. Since $\frac{1}{(2k-1) \cdot (2k+1)}$ is a common factor in each column, we extract the numerical factor $\frac{1}{((2k-1) \cdot (2k+1))^{N-k}}$ in front of the determinant and modify the entries of $(N-k)$ columns of the smaller block to 
\begin{equation}
\label{recursion}
\frac{3 \cdot 5 \cdot \ldots \cdot (2k+1)}{(2\ell+1)\cdot  (2\ell-1) \cdot \ldots \cdot  (2 \ell - 2k + 3) \cdot (2\ell - 2k + 1)!}, \qquad \ell = k,\dots,N-1.
\end{equation}
After eliminating the first two rows and the first two columns of the larger block, 
$k$ is incremented by $1$ until we reach $k = N-1$. 
When we recompute the entries in (\ref{induction}), we substitute (\ref{recursion}) with $k \to k-2$ and $\ell \to \ell - 1$ into the first term of (\ref{induction}) 
because the entries of the columns were not changed but 
the size of the block was reduced by two columns in the previous iteration and we substitute (\ref{recursion}) with $k \to k-1$ and $\ell \to \ell$ into the second term of (\ref{induction}) because the entries were changed but the size of the block contains the same number of columns as in the previous iteration. This proves the recursion formula (\ref{recursion}) by the method of mathematical induction.

\underline{\bf Step $k = N-1$.} At this last step, we have a $3 \times 3$ deteterminant. We perform (\ref{induction}) for the remaining column 
of a smaller block, extract the factor $\frac{1}{(2N-3) \cdot (2N-1)}$ and obtain a trivial $3\times 3$ determinant, which is equal to $1$, see examples for $N = 2,3,4$. Combining together all numerical factors, we obtain 
$$
A_N = \frac{1}{3^{N-1} \cdot (3 \cdot 5)^{N-2} \cdot (5 \cdot 7)^{N-3} \cdot 
	\ldots \cdot ((2N-5) \cdot(2N-3))^2 \cdot ((2N-3) \cdot (2N-1))},
$$
which recovers the numerical value (\ref{eq:A1}).
\end{proof}

\section{The proof of (\ref{det-expression-3})}
\label{app-B}
 \setcounter{equation}{0}
 
We first simplify the determinant by performing the same operations as in Appendix \ref{app-A}. However, since $\tilde{M}_N(1)$ consists of two blocks of $(N+1)$ and $(N-1)$ columns, we only add the $(N+1+j)$-th column to the $j$-th column for $1 \leq j \leq N-1$, then extract the factor of $2$ from the first $(N-1)$ columns,  then subtract the updated $j$-th column from the $(N+1+j)$-th column for $1 \leq j \leq N-1$ and finally, multiply the last $(N-1)$ columns by the negative signs. This leads to the block structure of $|\tilde{M}_N(1)|$:

\begin{align*}
|\tilde{M}_N(1)| &= 2^{N-1} (-1)^{N-1}\\
	&\left|
	\begin{array}{c c c c : c c:c c c}
	1 & 0 & 0 &  \dots & 0 & 0 & 0 & 0 & \dots \\
	0  & 1 & 0 &  \dots & 0 & 0 & 1 & 0 & \dots \\
	\frac{1}{2!} & 0  & 1 & \dots & 0 & 0 & 0 & 1 & \dots \\
	0 & \frac{1}{2!} & 0 &  \dots & 0 & 0 & \frac{1}{3!} & 0 & \dots \\
	\vdots & \vdots &  \vdots & \ddots & \vdots & \vdots & \vdots & \vdots &  \ddots \\
	\frac{1}{(N-1)!} & 0  & \frac{1}{(N-3)!} & \dots & 1 & 0 & 0 & \frac{1}{(N-2)!} &  \dots \\
	0 & \frac{1}{(N-1)!}  & 0 & \dots & 1 & 1 & \frac{1}{N!} & 0 &  \dots \\
	\vdots & \vdots &  \vdots & \ddots & \vdots & \vdots & \vdots & \vdots &  \ddots \\
	\frac{1}{(2N-2)!} & 0 & \frac{1}{(2N-4)!} & \dots & \frac{1}{(N-1)!} & \frac{1}{(N-2)!} & 0 & \frac{1}{(2N-3)!} & \dots \\
	0 & \frac{1}{(2N-2)!} & 0 & \dots & \frac{1}{N!} & \frac{1}{(N-1)!} & \frac{1}{(2N-1)!} & 0 & \dots 
	\end{array}
	\right|,
\end{align*}
where the $N$-th and $(N+1)$-th columns are unchanged and where we also show the $N$-th and $(N+1)$-th rows in the case when $N$ is odd. To further simplify the determinant, we expand it trivially along the first row and obtain the $(2N-1)\times(2N-1)$ determinant:
\begin{align}
\label{B}
B_N :=\left|
\begin{array}{c c c : c c:c c c}
1 & 0 &  \dots & 0 & 0 & 1 & 0 & \dots \\
0  & 1 & \dots & 0 & 0 & 0 & 1 & \dots \\
\frac{1}{2!} & 0 &  \dots & 0 & 0 & \frac{1}{3!} & 0 & \dots \\
\vdots &  \vdots & \ddots & \vdots & \vdots & \vdots & \vdots & \ddots \\
0  & \frac{1}{(N-3)!} & \dots & 1 & 0 & 0 & \frac{1}{(N-2)!} & \dots \\
\frac{1}{(N-1)!}  & 0 & \dots & 1 & 1 & \frac{1}{N!} & 0 & \dots \\
\vdots &  \vdots & \ddots & \vdots & \vdots & \vdots & \vdots & \ddots \\
0 & \frac{1}{(2N-4)!} & \dots & \frac{1}{(N-1)!} & \frac{1}{(N-2)!} & 0 & \frac{1}{(2N-3)!} & \dots \\
\frac{1}{(2N-2)!} & 0 & \dots & \frac{1}{N!} & \frac{1}{(N-1)!} & \frac{1}{(2N-1)!} & 0 & \dots 
	\end{array}
\right|,
\end{align}
where the two middle rows and columns are now located at $(N-1)$-th and $N$-th positions. The following proposition gives the proof of (\ref{det-expression-3}).

\begin{prop}
	\label{pro2}
	For any $N \in \mathbb{N}$, we have 
	\begin{equation}
	\label{eq:B1} 
	B_{N} = \frac{N}{1^{2N-1} 3^{2N-3} 5^{2N-5} 7^{2N-7} \ldots (2N-3)^3 (2N-1)}. 
	\end{equation}
\end{prop}
	
\begin{proof}
When $N=1$, we obtain $B_1 = 1$, which agrees with (\ref{eq:B1}). 
When $N=2$, we obtain 
\begin{align*}
B_{2} =
\begin{vmatrix}
1 & 0 & 1\\
1 & 1 & 0\\
\frac{1}{2!} & 1&\frac{1}{3!}
\end{vmatrix}
\xlongequal{\,C_1-C_3\,}
\begin{vmatrix}
0 & 0 & 1\\
1 & 1 & 0\\
\frac{1}{3} & 1&\frac{1}{3!}
\end{vmatrix}
\xlongequal{\,C_2-C_1\,}
\begin{vmatrix}
0 & 0 & 1\\
1 & 0 & 0\\
\frac{1}{3} & \frac{2}{3}&\frac{1}{3!}
\end{vmatrix}
=  \frac{2}{3},
\end{align*}
which agrees with (\ref{eq:B1}). When $N=3$, we obtain 
\begin{align*}
B_{3} &=
\begin{vmatrix}
1 & 0 & 0 & 1 & 0\\
0 & 1 & 0 & 0 & 1\\
\frac{1}{2!} & 1 & 1 & \frac{1}{3!}&0\\
0 & \frac{1}{2!} &1 &0 & \frac{1}{3!}\\
\frac{1}{4!} & \frac{1}{3!} &\frac{1}{2!} &\frac{1}{5!} & 0 
\end{vmatrix} 
\xlongequal[C_2-C_5]{C_1-C_4}
\begin{vmatrix}
0 & 0 & 0 & 1 & 0\\
0 & 0 & 0 & 0  & 1\\
\frac{1}{3} & 1 & 1 & \frac{1}{3!}&0\\
0 & \frac{1}{3} &1 &0 & \frac{1}{3!}\\
\frac{1}{5 \cdot 3!} & \frac{1}{3!} &\frac{1}{2!} &\frac{1}{5!} & 0 
\end{vmatrix} 
\xlongequal{C_3-C_2}	
\frac{1}{3}
\begin{vmatrix}
0 & 0 & 0 & 1 & 0\\
0 & 0 & 0 & 0  & 1\\
1 & 1 & 0 & \frac{1}{3!}&0\\
0 & \frac{1}{3} & \frac{2}{3} &0 & \frac{1}{3!}\\
\frac{1}{5 \cdot 2} & \frac{1}{3!} &\frac{1}{3} &\frac{1}{5!} & 0 
\end{vmatrix} 
	\\
&=	\frac{1}{3}
	\begin{vmatrix}
	1 & 1 & 0 \\
	0 & \frac{1}{3} & \frac{2}{3}\\
	\frac{1}{5 \cdot 2} & \frac{1}{3!} &\frac{1}{3} 
	\end{vmatrix}
\xlongequal{C_2-C_1}
\frac{1}{3}
\begin{vmatrix}
1 & 0 & 0\\
0 & \frac{1}{3} & \frac{2}{3}\\
\frac{1}{5 \cdot 2}& \frac{1}{5 \cdot 3}&\frac{1}{3}
\end{vmatrix} 
=\frac{1}{3 \cdot 3^2}
\begin{vmatrix}
1 & 0 & 0\\
0 & 1 & 2\\
\frac{1}{5 \cdot 2}& \frac{1}{5}& 1
\end{vmatrix} 
\xlongequal{C_3-2 C_2}
\frac{1}{3 \cdot 3^2}
\begin{vmatrix}
1 & 0 & 0\\
0 & 1 & 0\\
\frac{1}{5 \cdot 2} & \frac{1}{5}&\frac{3}{5}
\end{vmatrix} \\
&=\frac{3}{3^3 \cdot 5},
\end{align*}
which agrees with (\ref{eq:B1}).

We further adapt the algorithm based on the two-column elimination from the proof of Proposition \ref{pro1} to the determinant $B_N$, where we focus on the transformations of the two middle columns at each step of the algorithm. 
We label the first middle column as $C_1$ and the second middle column as $C_2$. As previously, we illustrate the algorithm in the case of $N = 4$.
	
\underline{\bf Step 1.} The determinant $B_N$ in (\ref{B}) contains three blocks of $(N-2)$ left columns, $2$ middle columns, and $(N-1)$ right columns. We subtract the first $(N-2)$ right columns from $(N-1)$ left columns. With the same recursive formula (\ref{step-1}),
we extract the numerical factor $\frac{1}{3^{N-2}}$ in front of the determinant and modify the entries of $(N-2)$ left columns to 
\begin{equation}
\label{B-step1}
		\frac{3}{(2\ell + 1) \cdot (2\ell-1)!}, \qquad \ell = 1,\dots,N-1,
	\end{equation}
where $\ell$ corresponds to the $(2\ell)$-th lower diagonal of the left $(2N-1) \times (N-2)$ block. 

In addition, we subtract the last right column from the first middle column. 
The nonzero entries at $j = N + 2 \ell$ are unchanged as $\frac{1}{(2 \ell +1)!}$
for $\ell = 0,1,\dots, \lfloor \frac{N-1}{2} \rfloor$ but the nonzero entries at $j = N + 2 \ell - 1$ are modified as $\frac{1}{(2\ell+1) \cdot (2\ell - 1)!}$ similarly to (\ref{step-1}) for $\ell = 1,\dots,\lfloor \frac{N}{2} \rfloor$, where $\lfloor \cdot \rfloor$ denotes the integer floor. For transparency, we record the nonzero entries of the first middle column $C_1$ as 
\begin{align}
\label{C1-step1}
\left\{ \begin{array}{ll} (C_1)_{N + 2 \ell -1} = \frac{1}{(2\ell+1) \cdot (2\ell - 1)!},
\quad & \ell = 1,\dots,\lfloor \frac{N}{2} \rfloor, \\
(C_1)_{N + 2 \ell} = \frac{1}{(2 \ell +1)!}, \quad & \ell = 0,\dots,\lfloor \frac{N-1}{2} \rfloor. \end{array} \right.
\end{align}

After this transformation, we subtract the first middle column from the second middle column. This changes the nonzero entries at $j = N+2\ell$ to $\frac{1}{(2\ell + 1) \cdot (2\ell -1)!}$ similarly to (\ref{step-1}) for $\ell = 1,\dots, \lfloor \frac{N-1}{2} \rfloor$ and changes the nonzero entries at $j = N+2\ell - 1$ as 
$$
\frac{1}{(2 \ell - 1)!} - \frac{1}{(2 \ell + 1) \cdot (2 \ell - 1)!} = \frac{(2 \ell)}{(2 \ell + 1) \cdot (2 \ell - 1)!}, \quad \ell = 1,\dots,\lfloor \frac{N}{2} \rfloor.
$$
Again, we record the nonzero entries of the second middle column $C_2$ as 
\begin{align}
\label{C2-step1}
\left\{ \begin{array}{ll} (C_2)_{N + 2 \ell -1} = \frac{(2 \ell)}{(2 \ell + 1) \cdot (2 \ell - 1)!},
\quad & \ell = 1,\dots,\lfloor \frac{N}{2} \rfloor, \\
(C_2)_{N + 2 \ell} = \frac{1}{(2\ell + 1) \cdot (2\ell -1)!}, \quad & \ell = 1,\dots,\lfloor \frac{N-1}{2} \rfloor. \end{array} \right.
\end{align}

For $N=4$, this step yields
	\begin{align*}
		B_4 &= \begin{vmatrix}
			1&0&0&0&1&0&0\\
			0&1&0&0&0&1&0\\
			\frac{1}{3}&0&1&0&\frac{1}{3!}&0&1\\
			0&\frac{1}{3}&1&1&0&\frac{1}{3!}&0\\
			\frac{1}{5 \cdot 3!}&0&\frac{1}{2!} &1& \frac{1}{5!}&0&\frac{1}{3!} \\
			0&\frac{1}{5 \cdot 3!}&\frac{1}{3!}&\frac{1}{2!}&0& \frac{1}{5!}&0 \\
			\frac{1}{7 \cdot 5!}&0&\frac{1}{4!} & \frac{1}{3!}&\frac{1}{7!}&0 &\frac{1}{5!}
		\end{vmatrix} 
		= \begin{vmatrix}
		0&0&0&0&1&0&0\\
		0&0&0&0&0&1&0\\
		\frac{1}{3}&0&0&0&\frac{1}{3!}&0&1\\
		0&\frac{1}{3}&1&1&0&\frac{1}{3!}&0\\
		\frac{1}{5 \cdot 3!}&0&\frac{1}{3} &1& \frac{1}{5!}&0&\frac{1}{3!} \\
		0&\frac{1}{5 \cdot 3!}&\frac{1}{3!}&\frac{1}{2!}&0& \frac{1}{5!}&0 \\
		\frac{1}{7 \cdot 5!}&0&\frac{1}{5 \cdot 3!} & \frac{1}{3!}&\frac{1}{7!}&0 &\frac{1}{5!}
		\end{vmatrix}\\
		&= \frac{1}{3^2} \begin{vmatrix}
			0&0&0&0&1&0&0\\
			0&0&0&0&0&1&0\\
			1&0&0&0&\frac{1}{3!}&0&1\\
			0&1&1&0&0&\frac{1}{3!}&0\\
			\frac{1}{5 \cdot 2}&0& \frac{1}{3}& \frac{2}{3} & \frac{1}{5!}&0&\frac{1}{3!} \\
			0&\frac{1}{5 \cdot 2}&\frac{1}{3!}&\frac{1}{3}&0&\frac{1}{5!}&0 \\
			\frac{3}{7 \cdot 5!}&0&\frac{1}{5 \cdot 2} &\frac{4}{5 \cdot 3!}& \frac{1}{7!}&0&\frac{1}{5!}
		\end{vmatrix} 
	\end{align*}
	
	\underline{\bf Step 2.} We eliminate the first rows and the first two right columns and obtain a determinant which contains three blocks of $(N-2)$ left columns, $2$ middle columns, and $(N-3)$ right columns. We subtract the first 
	$(N-3)$ left columns from $(N-3)$ right columns. With the same recursive formula (\ref{step-2}), we extract the numerical factor $\frac{1}{(3 \cdot 5)^{N-3}}$ in front of the determinant and modify the entries of $(N-3)$ right columns to 
\begin{equation}
\label{B-step2}
\frac{3 \cdot 5}{(2\ell + 1) \cdot (2\ell-1) \cdot (2\ell - 3)!}, \qquad \ell = 2,\dots,N-1,
\end{equation}
where $\ell$ corresponds to the $2(\ell - 1)$-th lower diagonal of the right $(2N-3) \times (N-3)$ block. Next, we subtract the last left column from the first middle column. 
The nonzero entries at $j = N + 2 \ell - 3$ are unchanged as $\frac{1}{(2 \ell + 1) \cdot (2 \ell - 1)!}$ obtained after Step 1 for $\ell = 1,\dots,\lfloor \frac{N}{2} \rfloor$ but the nonzero entries at $j = N + 2 \ell - 2$ are modified as $\frac{1}{(2\ell + 3) \cdot (2 \ell + 1) \cdot (2 \ell - 1)!}$ similarly to (\ref{step-2}) for $\ell = 1,\dots,\lfloor \frac{N-1}{2} \rfloor$. Since $\frac{1}{3}$ is a common factor in both middle columns, we extract it from both middle columns. Entries to the two columns are written explicitly as
\begin{align}
\label{C1-step2}
\left\{ \begin{array}{ll}  
(C_1)_{N+2\ell - 3} = \frac{3}{(2 \ell + 1) \cdot (2 \ell - 1)!}, \quad & \ell = 1,\dots,\lfloor \frac{N}{2} \rfloor, \\
(C_1)_{N+2\ell - 2} = \frac{3}{(2\ell + 3) \cdot (2 \ell + 1) \cdot (2 \ell - 1)!},  \quad & \ell = 1,\dots,\lfloor \frac{N-1}{2} \rfloor 
\end{array} \right.
\end{align}
and 
\begin{align*}
\left\{ \begin{array}{ll}  
(C_2)_{N+2 \ell -3} = \frac{(2 \ell) \cdot 3}{(2 \ell + 1) \cdot (2 \ell - 1)!}, \quad & \ell = 1,\dots,\lfloor \frac{N}{2} \rfloor, \\
(C_2)_{N+2 \ell -2} = \frac{3}{(2 \ell + 1) \cdot (2 \ell - 1)!},  \quad & \ell = 1,\dots,\lfloor \frac{N-1}{2} \rfloor 
\end{array} \right.
\end{align*}
After this operation, we subtract the first middle column multiplied by $2$ from the second middle column and recompute nonzero entries of the second column as 
\begin{align}
\label{C2-step2}
\left\{ \begin{array}{ll} (C_2)_{N+2 \ell -3} = \frac{3}{(2 \ell + 1) \cdot (2 \ell - 1) \cdot (2\ell - 3)!},
\quad & \ell = 2,\dots,\lfloor \frac{N}{2} \rfloor, \\
(C_2)_{N+2 \ell -2} = \frac{3}{(2 \ell + 3) \cdot (2 \ell - 1)!}, \quad & \ell = 1,\dots,\lfloor \frac{N-1}{2} \rfloor.  \end{array} \right.
\end{align}
	
For $N=4$, this step yields
	\begin{align*}
		B_{4} &=\frac{1}{3^2}\begin{vmatrix}
			1&0&0&0&1\\
			0&1&1&0&0\\
			\frac{1}{5 \cdot 2}&0&\frac{1}{3}& \frac{2}{3} & \frac{1}{3!}\\
			0&\frac{1}{5 \cdot 2}&\frac{1}{3!}&\frac{1}{3}&0\\
			\frac{3}{7 \cdot 5!}&0&\frac{1}{5\cdot 3!}&\frac{4}{5 \cdot 3!}&\frac{1}{5!}
		\end{vmatrix}=\frac{1}{3^2}
		\begin{vmatrix}
			1&0&0&0&0\\
			0&1&0&0&0\\
			\frac{1}{5 \cdot 2}&0&\frac{1}{3}&\frac{2}{3}&\frac{1}{5 \cdot 3}\\
			0&\frac{1}{5 \cdot 2}&\frac{1}{5 \cdot 3}&\frac{1}{3}&0\\
			\frac{3}{7 \cdot 5!}&0&\frac{1}{5\cdot 3!}&\frac{4}{5 \cdot3!}&\frac{1}{7 \cdot 5 \cdot 3!}
		\end{vmatrix} \\
		&=\frac{1}{3^2 \cdot (3 \cdot 5)^1 \cdot 3^2}
		\begin{vmatrix}
			1&0&0&0&0\\
			0&1&0&0&0\\
			\frac{1}{5 \cdot 2}&0&1&2&1\\
			0&\frac{1}{5 \cdot 2}&\frac{1}{5}&1&0\\
			\frac{3}{7 \cdot 5!}&0&\frac{1}{5\cdot 2}&\frac{4}{5 \cdot 2}&\frac{1}{7 \cdot 2}
		\end{vmatrix} 
		=\frac{1}{3^2 \cdot (3 \cdot 5)^1 \cdot 3^2}
		\begin{vmatrix}
		1&0&0&0&0\\
		0&1&0&0&0\\
		\frac{1}{5 \cdot 2}&0&1&0&1\\
		0&\frac{1}{5 \cdot 2}&\frac{1}{5}&\frac{3}{5}&0\\
		\frac{3}{7 \cdot 5!}&0&\frac{1}{5\cdot 2}&\frac{1}{5}&\frac{1}{7 \cdot 2}.
		\end{vmatrix}
	\end{align*}

	\underline{\bf Step 3.} We eliminate the first two rows and first two left columns and obtain a determinant which contains three blocks of $(N-4)$ left columns, $2$ middle columns, and $(N-3)$ right columns. If $N \geq 5$, we subtract the first $(N-4)$ right columns from $(N-4)$ left columns. With the same recursive formula (\ref{step-3}), we extract the numerical factor $\frac{1}{(5 \cdot 7)^{N-4}}$ in front of the determinant and modify the entries of $(N-4)$ left columns to 
\begin{equation}
\label{B-step3}
	\frac{3 \cdot 5 \cdot 7}{(2\ell + 1) \cdot (2\ell-1) \cdot (2\ell - 3) \cdot (2\ell - 5)!}, \qquad \ell = 3,\dots,N-1,
\end{equation}
	where $\ell$ corresponds to the $2(\ell - 2)$-th lower diagonal of the left $(2N-5) \times (N-4)$ block. Next, we subtract the last right column from the first middle column. 
	The nonzero entries at $j = N+2 \ell - 4$ are unchanged as $\frac{3}{(2\ell + 3) \cdot (2 \ell + 1) \cdot (2 \ell - 1)!}$ obtained after Step 2 for $\ell = 1,\dots, \lfloor \frac{N-1}{2} \rfloor$ but the nonzero entries for 
	$j = N+2\ell - 5$ are modified as $\frac{3}{(2\ell + 3) \cdot (2 \ell + 1) \cdot (2 \ell - 1) \cdot (2 \ell - 3)!}$ similarly to (\ref{step-3}) for $\ell = 2, \dots, \lfloor \frac{N}{2} \rfloor$. Since $\frac{1}{5}$ is a common factor in both middle columns, we extract it from both middle columns. Entries to the two columns are written explicitly as
\begin{align}
\label{C1-step3}
\left\{ \begin{array}{ll}  
(C_1)_{N+2\ell - 5} = \frac{3 \cdot 5}{(2\ell + 3) \cdot (2 \ell + 1) \cdot (2 \ell - 1) \cdot (2 \ell - 3)!}, \quad & \ell = 2,\dots,\lfloor \frac{N}{2} \rfloor, \\
(C_1)_{N+2\ell - 4} = \frac{3 \cdot 5}{(2\ell + 3) \cdot (2 \ell + 1) \cdot (2 \ell - 1)!},  \quad & \ell = 1,\dots,\lfloor \frac{N-1}{2} \rfloor 
\end{array} \right.
\end{align}
and 
\begin{align*}
\left\{ \begin{array}{ll}  
(C_2)_{N+2 \ell - 5} = \frac{3 \cdot 5}{(2 \ell + 1) \cdot (2 \ell - 1) \cdot (2\ell - 3)!}, \quad & \ell = 2,\dots,\lfloor \frac{N}{2} \rfloor, \\
(C_2)_{N+2 \ell - 4} = \frac{3 \cdot 5}{(2 \ell + 3) \cdot (2 \ell - 1)!}, \quad & \ell = 1,\dots,\lfloor \frac{N-1}{2} \rfloor 
\end{array} \right. 
\end{align*}	
After this operation, we subtract the first middle column multiplied by $3$ from the second middle column and recompute nonzero entries of the second column as 
\begin{align}
\label{C2-step3}
\left\{ \begin{array}{ll} (C_2)_{N+2 \ell -5} = \frac{(2 \ell) \cdot 3 \cdot 5}{(2 \ell + 3) \cdot (2 \ell + 1) \cdot (2 \ell - 1) \cdot (2\ell - 3)!},
\quad & \ell = 2,\dots,\lfloor \frac{N}{2} \rfloor, \\
(C_2)_{N+2 \ell - 4} = \frac{3 \cdot 5}{(2\ell + 3) \cdot (2 \ell + 1) \cdot (2 \ell - 1) \cdot (2 \ell - 3)!}, \quad & \ell = 2,\dots,\lfloor \frac{N-1}{2} \rfloor.  \end{array} \right.
\end{align}	

For $N=4$, this step yields
\begin{align*}
B_4 &= \frac{1}{3^2 \cdot (3 \cdot 5)^1 \cdot 3^2}
\begin{vmatrix}
1&0&1\\
\frac{1}{5}&\frac{3}{5}&0\\
\frac{1}{5\cdot 2}&\frac{1}{5}&\frac{1}{7 \cdot 2}
\end{vmatrix}
= \frac{1}{3^2 \cdot (3 \cdot 5)^1 \cdot 3^2}
\begin{vmatrix}
0&0&1\\
\frac{1}{5}&\frac{3}{5}&0\\
\frac{1}{7 \cdot 5}&\frac{1}{5}&\frac{1}{7 \cdot 2}
\end{vmatrix} \\
&= \frac{1}{3^2 \cdot (3 \cdot 5)^1 \cdot 3^2 \cdot 5^2}
\begin{vmatrix}
0&0&1\\
1&3&0\\
\frac{1}{7}&1&\frac{1}{7 \cdot 2}
\end{vmatrix}
= \frac{1}{3^2 \cdot (3 \cdot 5)^1 \cdot 3^2 \cdot 5^2}
\begin{vmatrix}
0&0&1\\
1&0&0\\
\frac{1}{7}& \frac{4}{7}&\frac{1}{7 \cdot 2}
\end{vmatrix}
= \frac{4}{3^5 \cdot 5^3 \cdot 7},
\end{align*}
which completes the algorithm for $N = 4$.

\underline{\bf Step $k$.} The algorithm is continued for any $2 \leq k \leq N-2$ 
by using the same alternation between the left and right blocks. With the same recursive formula (\ref{induction}) but for $\ell = k,\dots,N-2$, we extract the numerical factor $\frac{1}{((2k-1) \cdot (2k+1))^{N-k-1}}$ and modify the entries of $(N-k)$ columns of a smaller block to 
\begin{equation}
\label{B-step-k}
\frac{3 \cdot 5 \cdot \ldots \cdot (2k+1)}{(2\ell+1)\cdot  (2\ell-1) \cdot \ldots \cdot  (2 \ell - 2k + 3) \cdot (2\ell - 2k + 1)!}, \qquad \ell = k,\dots,N-1,
\end{equation}
where $\ell$ corresponds to the $2 (\ell-k+1)$-th lower diagonal of the smaller $(2N-2k+1) \times (N-k-1)$ block. Next, we subtract the last column of the larger block from the first middle column and extract the numerical factor $\frac{1}{(2k-1)^2}$ from both middle columns. Entries of the first middle columns are now written in the form: 
\begin{align}
\label{C1-step-k-even}
\left\{ \begin{array}{ll}  
(C_1)_{N+2\ell - 2k + 1} = \frac{3 \cdot 5 \cdot \dots \cdot (2k-1)}{(2 \ell + 2m-1) \cdot (2\ell + 2m - 3) \cdot \ldots \cdot (2\ell - 2m + 3) \cdot (2 \ell - 2m + 1)!}, \; & \ell = m,\dots,\lfloor \frac{N}{2} \rfloor, \\
(C_1)_{N+2\ell - 2k + 2} = \frac{3 \cdot 5 \cdot \dots \cdot (2k-1)}{(2\ell + 2m + 1) \cdot (2 \ell + 2m - 1) \cdot \ldots \cdot (2 \ell -2m + 3) \cdot (2 \ell - 2m + 1)!},  \; & \ell = m,\dots,\lfloor \frac{N-1}{2} \rfloor ,
\end{array} \right.
\end{align}
if $k = 2m$ is even, and in the form:
\begin{align}
\label{C1-step3-k-odd}
\left\{ \begin{array}{ll}  
(C_1)_{N+2\ell - 2k + 1} = \frac{3 \cdot 5 \cdot \dots \cdot (2k-1)}{(2\ell + 2m - 1) \cdot (2 \ell + 2m-3) \cdot \ldots \cdot (2 \ell - 2m + 3) \cdot (2 \ell - 2m+1)!}, \; & \ell = m,\dots,\lfloor \frac{N}{2} \rfloor, \\
(C_1)_{N+2\ell - 2k + 2} = \frac{3 \cdot 5 \cdot \dots \cdot (2k-1)}{(2\ell + 2m - 1) \cdot (2 \ell + 2 m - 3) \cdot \ldots \cdot (2\ell - 2 m + 5) \cdot (2 \ell - 2m + 3)!},  \; & \ell = m-1,\dots,\lfloor \frac{N-1}{2} \rfloor ,
\end{array} \right.
\end{align}
if $k = 2m-1$ is odd, see (\ref{C1-step1}), (\ref{C1-step2}), and (\ref{C1-step3}). 

Finally, we subtract the first middle column multiplied by $k$ from the second middle column. Entries of the second middle columns are now written in the form: 
\begin{align}
\label{C2-step-k-even}
\left\{ \begin{array}{ll} (C_2)_{N+2 \ell - 2k + 1} = \frac{3 \cdot 5 \cdot \dots \cdot (2k-1)}{(2 \ell + 2m-1) \cdot (2\ell + 2m - 3) \cdot \ldots \cdot (2\ell - 2m + 1) \cdot (2 \ell - 2m - 1)!},
\quad & \ell = m+1,\dots,\lfloor \frac{N}{2} \rfloor, \\
(C_2)_{N+2 \ell - 2k + 2} = \frac{3 \cdot 5 \cdot \dots \cdot (2k-1)}{{\tiny \underbrace{(2\ell + 2m + 1) \cdot (2 \ell + 2m - 1) \cdot \ldots \cdot (2 \ell -2m + 3)}_{\neq (2\ell + 1)} }\cdot (2 \ell - 2m + 1)!}, \quad & \ell = m,\dots,\lfloor \frac{N-1}{2} \rfloor.  \end{array} \right.
\end{align}
if $k = 2m$ is even, and in the form:
\begin{align}
\label{C2-step-k-odd}
\left\{ \begin{array}{ll} (C_2)_{N+2 \ell - 2k+1} = \frac{(2 \ell) \cdot 3 \cdot 5 \cdot \dots \cdot (2k-1)}{(2\ell + 2m - 1) \cdot (2 \ell + 2m-3) \cdot \ldots \cdot (2 \ell - 2m + 3) \cdot (2 \ell - 2m+1)!},
\; & \ell = m,\dots,\lfloor \frac{N}{2} \rfloor, \\
(C_2)_{N+2 \ell - 2k + 2} = \frac{3 \cdot 5 \cdot \dots \cdot (2k-1)}{(2\ell + 2m - 1) \cdot (2 \ell + 2m-3) \cdot \ldots \cdot (2 \ell - 2m + 3) \cdot (2 \ell - 2m+1)!}, \; & \ell = m,\dots,\lfloor \frac{N-1}{2} \rfloor, 
 \end{array} \right.
\end{align}	
if $k = 2m-1$ is odd, see (\ref{C2-step1}), (\ref{C2-step2}), and (\ref{C2-step3}).

After eliminating the first two rows and the first two columns of the larger block, 
$k$ is incremented by $1$ until we reach $k = N-1$. 
The recursion formulas (\ref{B-step-k}), (\ref{C1-step-k-even}), (\ref{C1-step3-k-odd}), (\ref{C2-step-k-even}), and (\ref{C2-step-k-odd}) are proven by the method of mathematical induction.

\underline{\bf Step $k = N-1$.} At this last step, 
we have a $3\times 3$ determinant. We subtract the last column of either left or right block from the first middle column and extract the numerical factor $\frac{1}{(2N-3)^2}$ from both middle columns. Then, we subtract the first middle column multiplied by $(N-1)$ from the second middle column. After this operation, we obtain a trivial $3 \times 3$ determinant, which is equal to $\frac{N}{(2N-1)}$, see examples for $N = 2,3,4$. Combining together all numerical factors, we obtain 
$$
B_N = \frac{N}{3^{N-2} \cdot (3 \cdot 5)^{N-3} \cdot \ldots \cdot ((2N-5) \cdot(2N-3))^1 \cdot 3^2 \cdot 5^2 \cdot \ldots \cdot (2N-3)^2 \cdot (2N-1)},
$$
which recovers the numerical value (\ref{eq:B1}).
\end{proof}

\section{The proof of (\ref{det-expression-4})}
\label{app-C}

Both $M^{(1)}_N(1)$ and $M^{(2)}_N(1)$ consist of two equal blocks of $N$ columns. 
Compared to $M_N(1)$ in Appendix \ref{app-A}, the last column of the left block 
is modified in $M_N^{(1)}(1)$ by an additional shift down 
and the last column of the right block is modified in $M_N^{(2)}(1)$ by an additional shift down. We perform the same operations with the first $(N-1)$ columns in the left and right blocks as in Appendix \ref{app-B}, expand it trivially along the first row, and obtain  obtain                                                                                                                                                                   
\begin{align*}
	\left| M^{(1)}_{N}(1) \right| = 2^{N-1} (-1)^{N-1} C_N^{(1)}, \qquad 
	\left| M^{(2)}_{N}(1) \right| = 2^{N-1} (-1)^{N-1} C_N^{(2)}, 
\end{align*}
where the $(2N-1) \times (2N-1)$ determinants $C_N^{(1)}$ and $C_N^{(2)}$ are given by 
\begin{align}
	C^{(1)}_{N} =\left|\begin{array}{c c c c : c c c c}
		1&0&\dots&0&1&0&\dots&0\\
		0&1&\dots&0&0&1&\dots&0\\
		\frac{1}{2!}&0&\dots&0&\frac{1}{3!}&0&\dots&0\\
		\vdots&\vdots&\ddots&\vdots&\vdots&\vdots&\ddots&\vdots\\
		0&\frac{1}{(N-3)!}&\cdots&0&0&\frac{1}{(N-2)!}&\cdots&1\\
		\frac{1}{(N-1)!}&0&\cdots&1&\frac{1}{N!}&0&\cdots&-1\\
		\vdots&\vdots&\ddots&\vdots&\vdots&\vdots&\ddots&\vdots\\
		0&\frac{1}{(2N-4)!}&\dots&\frac{1}{(N-2)!}&0&\frac{1}{(2N-3)!}&\dots&\frac{1}{(N-1)!}\\
		\frac{1}{(2N-2)!}&0&\dots&\frac{1}{(N-1)!}&\frac{1}{(2N-1)!}&0&\dots&-\frac{1}{N!}
	\end{array}\right|, \label{C1}
	\end{align}
	and
	\begin{align}
	C^{(2)}_{N}:=\left|\begin{array}{c c c c : c c c c}
		1&0&\dots&0&1&0&\dots&0\\
		0&1&\dots&0&0&1&\dots&0\\
		\frac{1}{2!}&0&\dots&0&\frac{1}{3!}&0&\dots&0\\
		\vdots&\vdots&\ddots&\vdots&\vdots&\vdots&\ddots&\vdots\\
		0&\frac{1}{(N-3)!}&\dots&1&0&\frac{1}{(N-2)!}&\cdots&0\\
		\frac{1}{(N-1)!}&0&\dots&1&\frac{1}{N!}&0&\cdots&1\\
		\vdots&\vdots&\ddots&\vdots&\vdots&\vdots&\ddots&\vdots\\
		0&\frac{1}{(2N-4)!}&\dots&\frac{1}{(N-1)!}&0&\frac{1}{(2N-3)!}&\dots&-\frac{1}{(N-2)!}\\
		\frac{1}{(2N-2)!}&0&\dots&\frac{1}{N!}&\frac{1}{(2N-1)!}&0&\dots&\frac{1}{(N-1)!}
	\end{array}\right|. 
	\label{C2}
\end{align}
We note that the $(N-1)$-th and $(2N-1)$-th columns in $C_N^{(1)}$ and $C_N^{(2)}$ 
are unchanged and we show entries of the $(N-1)$-th and $N$-th rows in the case when $N$ is odd. The following proposition gives the proof of (\ref{det-expression-4})

\begin{prop}\label{pro3}
	For any $N \in \mathbb{N}$, we have 
	\begin{align} 
C_N^{(1)}= - C_N^{(2)} = - \frac{N^2}{1^{2N-1}3^{2N-3}5^{2N-5}7^{2N-7}\ldots(2N-3)^{3}(2N-1)^1}.
		\label{eq:C4}
	\end{align}
\end{prop}

\begin{proof}
The case $N=1$ is exceptional, for which we obtain directly
$$
|M^{(1)}_1(1)| = \left| \begin{matrix} 0 & 1 \\ 1 & -1 \end{matrix} \right| = -1, \quad  |M^{(2)}_1(1)| = \left| \begin{matrix} 1 & 0 \\ 1 & 1 \end{matrix} \right| = 1,
$$ 
which still agrees with (\ref{eq:C4}). When $N=2$, we obtain
	\begin{align*}
		&C^{(1)}_{2}
		=
		\begin{vmatrix}
			0 & 1 & 1 \\
			1 & 0 & -1 \\
			1 & \frac{1}{3!} & \frac{1}{2!}
		\end{vmatrix}
		\xlongequal{\,C_3-C_2\,}
		\begin{vmatrix}
			0 & 1 & 0 \\
			1 & 0 & -1 \\
			1 & \frac{1}{3!} & \frac{1}{3}
		\end{vmatrix}
		\xlongequal{\,C_1+C_3\,}
		\begin{vmatrix}
			0 & 1 & 0 \\
			0 & 0 & -1 \\
			\frac{4}{3} & \frac{1}{3!} & \frac{1}{3}
		\end{vmatrix}
		=-\frac{2^2}{3}
		, \\
		&C^{(2)}_{2}=
		\begin{vmatrix}
			1 & 1 & 0 \\
			1 & 0 & 1 \\
			\frac{1}{2!} & \frac{1}{3!} & -1
		\end{vmatrix}
		\xlongequal{\,C_1-C_2\,}
		\begin{vmatrix}
			0 & 1 & 0 \\
			1 & 0 & 1 \\
			\frac{1}{3} & \frac{1}{3!} & -1
		\end{vmatrix}
		\xlongequal{\,C_1-C_3\,}
		\begin{vmatrix}
			0 & 1 & 0 \\
			0 & 0 & 1 \\
			\frac{4}{3} & \frac{1}{3!} & -1
		\end{vmatrix}=\frac{2^2}{3},
	\end{align*}
	which agrees with (\ref{eq:C4}). When $N=3$, we obtain 
	\begin{align*}
C^{(1)}_{3} &=\begin{vmatrix}
			1&0&1&0&0\\
			0&0&0&1&1\\
			\frac{1}{2!}&1&\frac{1}{3!}&0&-1\\
			0&1&0&\frac{1}{3!}&\frac{1}{2!}\\
			\frac{1}{4!}&\frac{1}{2!}&\frac{1}{5!}&0&-\frac{1}{3!}
		\end{vmatrix}
		\xlongequal[C_5-C_4]{C_1-C_3}
		\begin{vmatrix}
			0&0&1&0&0\\
			0&0&0&1&0\\
			\frac{1}{3}&1&\frac{1}{3!}&0&-1\\
			0&1&0&\frac{1}{3!}&\frac{1}{3}\\
			\frac{1}{5\cdot 3!}&\frac{1}{2!}&\frac{1}{5!}&0&-\frac{1}{3!}
		\end{vmatrix} 
		= \begin{vmatrix}
			\frac{1}{3}&1&-1\\
			0&1&\frac{1}{3}\\
			\frac{1}{5\cdot 3!}&\frac{1}{2!}&-\frac{1}{3!}
		\end{vmatrix} = \frac{1}{3}\begin{vmatrix}
		1&1&-1\\
		0&1&\frac{1}{3}\\
		\frac{1}{5\cdot 2}&\frac{1}{2!}&-\frac{1}{3!}
		\end{vmatrix} \\
		&
		\xlongequal[C_3+C_1]{C_2+C_3}
		\frac{1}{3}\begin{vmatrix}
			1&0&0\\
			0&\frac{4}{3}&\frac{1}{3}\\
			\frac{1}{5\cdot 2}&\frac{1}{3}&-\frac{1}{5\cdot 3}
		\end{vmatrix}
		=\frac{1}{3^3}
		\begin{vmatrix}
			1&0&0\\
			0&4&1\\
			\frac{1}{5\cdot 2}&1&-\frac{1}{5}
		\end{vmatrix}
		\xlongequal{\,C_2- 4 C_3\,}
		\frac{1}{3^3}
		\begin{vmatrix}
			1&0&0\\
			0&0&1\\
			\frac{1}{5\cdot 2}&\frac{9}{5}&-\frac{1}{5}
		\end{vmatrix} = -\frac{3^2}{3^3\cdot 5^1},\\
C^{(2)}_{3} &=\begin{vmatrix}
			1&0&1&0&0\\
			0&1&0&1&0\\
			\frac{1}{2!}&1&\frac{1}{3!}&0&1\\
			0&\frac{1}{2!}&0&\frac{1}{3!}&-1\\
			\frac{1}{4!}&\frac{1}{3!}&\frac{1}{5!}&0&\frac{1}{2!}
		\end{vmatrix}
		\xlongequal[C_2-C_4]{C_1-C_3}
		\begin{vmatrix}
			0&0&1&0&0\\
			0&0&0&1&0\\
			\frac{1}{3}&1&\frac{1}{3!}&0&1\\
			0&\frac{1}{3}&0&\frac{1}{3!}&-1\\
			\frac{1}{5\cdot 3!}&\frac{1}{3!}&\frac{1}{5!}&0&\frac{1}{2!}
		\end{vmatrix}
			= \begin{vmatrix}
		\frac{1}{3} &1&1\\
		0&\frac{1}{3}&-1\\
		\frac{1}{5\cdot 3!} &\frac{1}{3!}&\frac{1}{2!}
		\end{vmatrix}
		= \frac{1}{3}\begin{vmatrix}
			1&1&1\\
			0&\frac{1}{3}&-1\\
			\frac{1}{5\cdot 2}&\frac{1}{3!}&\frac{1}{2!}
		\end{vmatrix}\\
		&\xlongequal[C_2-C_1]{C_3-C_2}
		\frac{1}{3}\begin{vmatrix}
			1&0&0\\
			0&\frac{1}{3}&-\frac{4}{3}\\
			\frac{1}{5\cdot 2}&\frac{1}{5\cdot 3}&\frac{1}{3}
		\end{vmatrix}
		=\frac{1}{3^3}
		\begin{vmatrix}
			1&0&0\\
			0&1&-4\\
			\frac{1}{5\cdot 2}&\frac{1}{5}&1
		\end{vmatrix}
		\xlongequal{\,C_3 + 4 C_2\,}
		\frac{1}{3^3}
		\begin{vmatrix}
			1&0&0\\
			0&1&0\\
			\frac{1}{5\cdot 2}&\frac{1}{5}&\frac{9}{5}
		\end{vmatrix} = \frac{3^2}{3^3\cdot 5^1},
	\end{align*}
	which agrees with (\ref{eq:C4}). 
	
	\underline{\bf We first prove that $C_N^{(1)} = -C_N^{(2)}$.}
	The determinants $C^{(1)}_{N}$ and $C^{(2)}_{N}$ contain two blocks of $(N-1)$ left columns and $N$ right columns, where the first $(N-2)$ left columns and 
	the first $(N-1)$ right columns of $C^{(1)}_{N}$ are identical to the corresponding columns of $C^{(2)}_{N}$. 
	Starting with $C_N^{(1)}$, we swap the last left column with the last right column, which changes the sign of the determinant to the opposite. If $N$ is odd, we multiply all odd-numbered rows and odd-numbered columns by $-1$. This transformation does not change the sign of the determinant. The final determinant coincides with $C_N^{(2)}$, which proves that $C_N^{(1)} = -C_N^{(2)}$. If $N$ is even, then we multiply all even-numbered rows, all even-numbered left columns and all odd-numbered right columns by $-1$. Again, this transformation does not change the sign of the determinant and since 
	it recovers $C_N^{(2)}$, we again obtain $C_N^{(1)} = -C_N^{(2)}$.
	For $N=4$ (even), we illustrate this transformation by
	\begin{align*}
	C^{(1)}_{4} &=\begin{vmatrix}
	1&0&0&1&0&0&0\\
	0&1&0&0&1&0&0\\
	\frac{1}{2!}&0&0&\frac{1}{3!}&0&1&1\\
	0&\frac{1}{2!}&1&0&\frac{1}{3!}&0&-1\\
	\frac{1}{4!}&0&1&\frac{1}{5!}&0&\frac{1}{3!}&\frac{1}{2!}\\
	0&\frac{1}{4!}&\frac{1}{2!}&0&\frac{1}{5!}&0&-\frac{1}{3!}\\
	\frac{1}{6!}&0&\frac{1}{3!}&\frac{1}{7!}&0&\frac{1}{5!}&\frac{1}{4!}
	\end{vmatrix}
	=-\begin{vmatrix}
	1&0&0&1&0&0&0\\
	0&1&0&0&1&0&0\\
	\frac{1}{2!}&0&1&\frac{1}{3!}&0&1&0\\
	0&\frac{1}{2!}&-1&0&\frac{1}{3!}&0&1\\
	\frac{1}{4!}&0&\frac{1}{2!}&\frac{1}{5!}&0&\frac{1}{3!}&1\\
	0&\frac{1}{4!}&-\frac{1}{3!}&0&\frac{1}{5!}&0&\frac{1}{2!}\\
	\frac{1}{6!}&0&\frac{1}{4!}&\frac{1}{7!}&0&\frac{1}{5!}&\frac{1}{3!}
	\end{vmatrix}\\
	&=-\begin{vmatrix}
	1&0&0&1&0&0&0\\
	0&1&0&0&1&0&0\\
	\frac{1}{2!}&0&1&\frac{1}{3!}&0&1&0\\
	0&\frac{1}{2!}&1&0&\frac{1}{3!}&0&1\\
	\frac{1}{4!}&0&\frac{1}{2!}&\frac{1}{5!}&0&\frac{1}{3!}&-1\\
	0&\frac{1}{4!}&\frac{1}{3!}&0&\frac{1}{5!}&0&\frac{1}{2!}\\
	\frac{1}{6!}&0&\frac{1}{4!}&\frac{1}{7!}&0&\frac{1}{5!}&-\frac{1}{3!}
	\end{vmatrix} = -C^{(2)}_{4}.
	\end{align*}
	Therefore, it is sufficient to prove (\ref{eq:C4}) for $C_N^{(1)}$ only. 
	We adapt the algorithm based on the two-column elimination from the proof of Proposition \ref{pro1} for the determinant  $C^{(1)}_{N}$, where we focus on the transformations of the last left and right columns at each step of the algorithm. As previously, we illustrate the algorithm in the case of $N = 4$.

	\underline{\bf Step $1$.} We subtract the first $(N-2)$ right columns from $(N-2)$ left columns. With the same recursive fomula (\ref{step-1}), we extract the numerical fator $\frac{1}{3^{N-2}}$ in front of the determinant and modify the entries of $(N-2)$ left columns to 
	\begin{equation*}
		\frac{3}{(2\ell + 1) \cdot (2\ell-1)!}, \qquad \ell = 1,2,\dots,N-1.
	\end{equation*}
	where $\ell$ corresponds to the $(2\ell)$-th lower diagonal of the $(2N-1)\times (N-2)$ block.                                                                                                                                                                                                                                                                                                                                                                                                                                                                                                                                                                                                                                                                                                                                                                                                                                                                                                                                                                                                                                                                                                                                                                                       
	We substract the $(N-1)$-nd column of the right columns from the last right column, entries of which change to $(-1)^{j+1} \frac{1}{j!\left(\ell+1+\frac{1-(-1)^{j}}{2}\right)}$ for $j=0,1,\dots,N-1$. The nonzero entries in the first $(N-2)$ columns of the right columns at $j=N+2\ell$ remain unchanged as $-\frac{1}{(2\ell+1)!}$ for $\ell=0,1,\dots,\left\lfloor\frac{N-1}{2}\right\rfloor$ but the nonzero entries in the first $(N-2)$ columns of the left columns at $j=N+2\ell-1$ are modified $\frac{1}{(2\ell+1)\cdot (2\ell-1)!}$ similary to (\ref{step-1}) for $\ell=1,2,\dots,\left\lfloor \frac{N}{2} \right\rfloor$. For transparency, we record the nonzero entries of the last right column $C_{2}$ as
	\begin{equation}
		\begin{dcases}
			\left(C_{2}\right)_{N+2\ell-1}=\frac{1}{(2\ell+1)\cdot(2\ell-1)!}, \qquad &\ell=1,2,\dots,\left\lfloor\frac{N}{2}\right\rfloor,\\
			\left(C_{2}\right)_{N+2\ell}=-\frac{1}{(2\ell+1)!},\qquad &\ell=0,1,\dots,\left\lfloor\frac{N-1}{2}\right\rfloor.
		\end{dcases}
	\end{equation}

	For $N=4$, this step yields
	\begin{align*}
		C^{(1)}_{4}&=\begin{vmatrix}
			1&0&0&1&0&0&0\\
			0&1&0&0&1&0&0\\
			\frac{1}{2!}&0&0&\frac{1}{3!}&0&1&1\\
			0&\frac{1}{2!}&1&0&\frac{1}{3!}&0&-1\\
			\frac{1}{4!}&0&1&\frac{1}{5!}&0&\frac{1}{3!}&\frac{1}{2!}\\
			0&\frac{1}{4!}&\frac{1}{2!}&0&\frac{1}{5!}&0&-\frac{1}{3!}\\
			\frac{1}{6!}&0&\frac{1}{3!}&\frac{1}{7!}&0&\frac{1}{5!}&\frac{1}{4!}
		\end{vmatrix}=
		\begin{vmatrix}
			0&0&0&1&0&0&0\\
			0&0&0&0&1&0&0\\
			\frac{1}{3}&0&0&\frac{1}{3!}&0&1&1\\
			0&\frac{1}{3}&1&0&\frac{1}{3!}&0&-1\\
			\frac{1}{5 \cdot 3!}&0&1&\frac{1}{5!}&0&\frac{1}{3!}&\frac{1}{2!}\\
			0&\frac{1}{5 \cdot 3!}&\frac{1}{2!}&0&\frac{1}{5!}&0&-\frac{1}{3!}\\
			\frac{1}{7\cdot 5!}&0&\frac{1}{3!}&\frac{1}{7!}&0&\frac{1}{5!}&\frac{1}{4!}
		\end{vmatrix}\\=
		&\frac{1}{3^2}
		\begin{vmatrix}
			0&0&0&1&0&0&0\\
			0&0&0&0&1&0&0\\
			1&0&0&\frac{1}{3!}&0&1&1\\
			0&1&1&0&\frac{1}{3!}&0&-1\\
			\frac{1}{5 \cdot 2}&0&1&\frac{1}{5!}&0&\frac{1}{3!}&\frac{1}{2!}\\
			0&\frac{1}{5 \cdot 2}&\frac{1}{2!}&0&\frac{1}{5!}&0&-\frac{1}{3!}\\
			\frac{3}{7\cdot 5!}&0&\frac{1}{3!}&\frac{1}{7!}&0&\frac{1}{5!}&\frac{1}{4!}
		\end{vmatrix}
		=\frac{1}{3^2}\begin{vmatrix}
			0&0&0&1&0&0&0\\
			0&0&0&0&1&0&0\\
			1&0&0&\frac{1}{3!}&0&1&0\\
			0&1&1&0&\frac{1}{3!}&0&-1\\
			\frac{1}{5 \cdot 2}&0&1&\frac{1}{5!}&0&\frac{1}{3!}&\frac{1}{3}\\
			0&\frac{1}{5 \cdot 2}&\frac{1}{2!}&0&\frac{1}{5!}&0&-\frac{1}{3!}\\
			\frac{3}{7\cdot 5!}&0&\frac{1}{3!}&\frac{1}{7!}&0&\frac{1}{5!}&\frac{1}{5 \cdot 3!}.
		\end{vmatrix}
	\end{align*}

	\underline{\bf Step 2.} We eliminate the first two rows and the first two right columns and obtain a determinant which contains two blocks of $(N-1)$ left columns and $(N-2)$ right columns. This does not change the sign factor of the determinant, independently whether $N$ is even or odd. 
	After the two-column elimination, we obtain a determinant which contains two blocks of $(N-1)$ left columns and $(N-2)$ right columns. With the same recursive formula (\ref{step-2}), we extract the numerical factor $\frac{1}{(3 \cdot 5)^{N-3}}$ in front of the determinant and modify the entries 
	of $(N-3)$ left columns to 
	\begin{equation*}
		\frac{3 \cdot 5}{(2\ell+1) (2\ell-1) \cdot (2\ell-3)!}, \qquad \ell = 2,\dots,N-2.
	\end{equation*}
	where $\ell$ corresponds to the $(2\ell-1)$-th lower diagonal of the right $(2N-5)\times(N-4)$ block.
	
	In addition, we add last right column to the last left column. This changes the nonzero entries in the last left column at $j=N+2\ell-1$ to $\frac{2\ell+2}{(2\ell+1)\cdot (2\ell-1)!}$ and changes the nonzero entries in the last left column as
	\begin{equation*}
		\frac{1}{(2\ell-1)!}+\frac{1}{(2\ell+1)\cdot(2\ell-1)!}=\frac{2\ell+2}{(2\ell+1)\cdot(2\ell-1)!},\qquad \ell=1,2,\dots,\left\lfloor\frac{N}{2}\right\rfloor.
	\end{equation*}
	This changes the nonzero entries in the last left column at $j=N+2\ell$ to $\frac{1}{(2\ell+1)\cdot (2\ell-1)!}$ and changes the nonzero entries in the last left column as 
	\begin{equation*}
		\frac{1}{(2\ell)!}-\frac{1}{(2\ell+1)!}=\frac{1}{(2\ell+1)\cdot(2\ell-1)!},\qquad \ell=1,2,\dots,\left\lfloor\frac{N-1}{2}\right\rfloor.
	\end{equation*}
	
	After this transformation, we add the $(N-2)$-th column of the first columns to the last right column. The nonzero entries at $j=N+2\ell-1$ are unchanged as $\frac{1}{(2\ell+1)\cdot(2\ell-1)!}$ for $\ell=1,2,\dots,\left\lfloor\frac{N}{2}\right\rfloor$ but the nonzero entries at $\ell=N+2\ell$ are modified as
	\begin{equation*}
		\frac{3}{(2\ell+3)\cdot(2\ell+1)!}-\frac{1}{(2\ell+1)!}=-\frac{1}{(2\ell+3)\cdot(2\ell+1)\cdot(2\ell-1)!},\quad \ell=0,1,\dots,\left\lfloor\frac{N-1}{2}\right\rfloor.
	\end{equation*}
	
	Since $\frac{1}{3}$ is a common factor in the last left column and the last right column, we extract it from the last left column and the last right column.	For transparency, we record the nonzero entries of the last left column $C_{1}$ as
	\begin{equation}
		\begin{dcases}
			\left(C_{1}\right)_{N+2\ell-1}=\frac{3\cdot(2\ell+2)}{(2\ell+1)\cdot(2\ell-1)!}, \qquad &\ell=1,2,\dots,\left\lfloor\frac{N}{2}\right\rfloor,\\
			\left(C_{1}\right)_{N+2\ell}=\frac{3}{(2\ell+1)\cdot(2\ell-1)!},\qquad &\ell=1,2,\dots,\left\lfloor\frac{N-1}{2}\right\rfloor.
		\end{dcases}
	\end{equation}
	Again, we record the nonzero entries of the last right column $C_{2}$ as
	\begin{equation}
		\begin{dcases}
			\left(C_{2}\right)_{N+2\ell-1}=\frac{3}{(2\ell+1)\cdot(2\ell-1)!}, \quad &\ell=1,2,\dots,\left\lfloor\frac{N}{2}\right\rfloor,\\
			\left(C_{2}\right)_{N+2\ell}=-\frac{3}{(2\ell+3)\cdot(2\ell+1)\cdot (2\ell-1)!},\quad &\ell=1,2,\dots,\left\lfloor\frac{N-1}{2}\right\rfloor.
		\end{dcases}
	\end{equation}
	For $N=4$, this step yields
	\begin{align*}
		C^{(1)}_{4}&=
		\frac{1}{3^2}\begin{vmatrix}
			1&0&0&1&0\\
			0&1&1&0&-1\\
			\frac{1}{5\cdot 2}&0&1&\frac{1}{3!}&\frac{1}{3}\\
			0&\frac{1}{5\cdot 2}&\frac{1}{2!}&0&-\frac{1}{3!}\\
			\frac{3}{7\cdot 5!}&0&\frac{1}{3!}&\frac{1}{5!}&\frac{1}{5 \cdot 3!}
		\end{vmatrix}
		=\frac{1}{3^2}\begin{vmatrix}
			1&0&0&0&0\\
			0&1&1&0&-1\\
			\frac{1}{5\cdot 2}&0&1&\frac{1}{5\cdot 3}&\frac{1}{3}\\
			0&\frac{1}{5\cdot 2}&\frac{1}{2!}&0&-\frac{1}{3!}\\
			\frac{3}{7\cdot 5!}&0&\frac{1}{3!}&\frac{1}{7\cdot 5\cdot 3!}&\frac{1}{5\cdot 3!}
		\end{vmatrix}\\
		&=\frac{1}{3^2\cdot (5\cdot 3)^1}\begin{vmatrix}
			1&0&0&0&0\\
			0&1&1&0&-1\\
			\frac{1}{5\cdot 2}&0&1&1&\frac{1}{3}\\
			0&\frac{1}{5\cdot 2}&\frac{1}{2!}&0&-\frac{1}{3!}\\
			\frac{3}{7\cdot 5!}&0&\frac{1}{3!}&\frac{1}{7\cdot 2}&\frac{1}{5 \cdot 3!}
		\end{vmatrix}
		=\frac{1}{3^2 \cdot(5\cdot 3)^1}\begin{vmatrix}
			1&0&0&0&0\\
			0&1&0&0&0\\
			\frac{1}{5\cdot 2}&0&\frac{4}{3}&1&\frac{1}{3}\\
			0&\frac{1}{5\cdot 2}&\frac{1}{3}&0&-\frac{1}{5\cdot 3}\\
			\frac{3}{7\cdot 5!}&0&\frac{1}{5}&\frac{1}{7\cdot 2}&\frac{1}{5 \cdot 3!}
		\end{vmatrix}\\
		&=
		\frac{1}{3^2 \cdot(5\cdot 3)^1\cdot 3^2}\begin{vmatrix}
			1&0&0&0&0\\
			0&1&0&0&0\\
			\frac{1}{5\cdot 2}&0&4&1&1\\
			0&\frac{1}{5\cdot 2}&1&0&-\frac{1}{5}\\
			\frac{3}{7\cdot 5!}&0&\frac{3}{5}&\frac{1}{7\cdot 2}&\frac{1}{5 \cdot 2}
		\end{vmatrix}.
	\end{align*}
	
	\underline{\bf Step 3.} We eliminate the first two rows and the first two left columns and obtain a determinant which contains two blocks of $(N-3)$ left columns and $(N-2)$ right columns. If $N\geq 5$, We subtract the first $(N-4)$ rights columns of the second block from $(N-2)$ columns of the  first block. With the same recursive formula (\ref{step-3}), we extract the numerical factor  $\frac{1}{(5 \cdot 7)^{N-4}}$ in front of the determinant and modify the entries of $(N-4)$ left columns to 
	\begin{equation*}
		\frac{3 \cdot 5 \cdot 7}{(2\ell+1) \cdot(2\ell-1) \cdot(2 \ell - 3) \cdot (2\ell-5)!}, \qquad \ell = 3,\dots,N-1.
	\end{equation*}
	where $\ell$ correspnds to the $2(\ell-2)$-th lower diagonal of the left $(2N-5)\times (N-4)$
	In addition, we subtract four times the last right column from the last left column. This changes the nonzero entries in the last left column at $j=N+2\ell-3$ to 
	\begin{align*}
		&\frac{3}{(2\ell+1)\cdot(2\ell-1)!}+\frac{4\cdot 3}{(2\ell+3)\cdot(2\ell+1)\cdot (2\ell-1)!}\\
		&=\frac{3\cdot(2\ell+7)}{(2\ell+3)\cdot(2\ell+1)\cdot(2\ell-1)!},\quad \ell=1,2,\dots,\left\lfloor\frac{N}{2}\right\rfloor.
	\end{align*}
	
	This changes the nonzero entries in the last left column at $j=N+2\ell-2$ to \begin{align*}
		&\frac{3\cdot(2\ell+2)}{(2\ell+1)\cdot (2\ell-1)!}-\frac{4\cdot 3}{(2\ell+1)\cdot(2\ell-1)!}\\
		&=\frac{3}{(2\ell+1)\cdot(2\ell-1)\cdot(2\ell-3)!},\quad \ell=2,3,\dots,\left\lfloor\frac{N-1}{2}\right\rfloor.
	\end{align*}
	After, we subtract the last right column from the $(N-5)$-th column of the right columns. The nonzero entries at $j=N+2\ell-3$ are unchanged as $-\frac{3}{(2\ell+3)\cdot(2\ell+1)\cdot (2\ell-1)!}$ obtain after Step 3 for $\ell=1,2,\dots,\lfloor\frac{N}{2}\rfloor$ but the nonzero entries at $j=N+2\ell-2$ as
	\begin{align*}
		&\frac{3}{(2\ell+1)\cdot(2\ell-1)!}-\frac{5\cdot 3}{(2\ell+3)\cdot(2\ell+1)\cdot(2\ell-1)!}\\
		&=\frac{3}{(2\ell+3)\cdot(2\ell+1)\cdot(2\ell-1)\cdot(2\ell-3)!},\quad \ell=2,3,\dots,\left\lfloor\frac{N-1}{2} \right\rfloor.
	\end{align*}
	Since $\frac{1}{5}$ is a common factor in the last left column and the last right column, we extract it from the last left column and the last right column. Entries to the two columns are written explicitly as
	\small{\begin{equation}
			\begin{dcases}
				(C_{1})_{N+2\ell-3}=\frac{3\cdot 5\cdot(2\ell+7)}{(2\ell+3)\cdot(2\ell+1)\cdot(2\ell-1)!},\quad &\ell=1,2,\dots,\left\lfloor\frac{N}{2} \right\rfloor,\\
				(C_{1})_{N+2\ell-2}=\frac{3\cdot 5}{(2\ell+1)\cdot(2\ell-1)\cdot(2\ell-3)!}\quad &\ell=2,3,\dots,\left\lfloor\frac{N-1}{2} \right\rfloor
			\end{dcases}
	\end{equation}}
	and
	\small{\begin{equation}
			\begin{dcases}
				(C_{2})_{N+2\ell-3}=-\frac{3\cdot 5}{(2\ell+3)\cdot(2\ell+1)\cdot (2\ell-1)!},\quad &\ell=1,2,\dots,\left\lfloor\frac{N}{2} \right\rfloor,\\
				(C_{2})_{N+2\ell-2}=\frac{3\cdot 5}{(2\ell+3)\cdot(2\ell+1)\cdot(2\ell-1)\cdot(2\ell-3)!}\quad &\ell=2,3,\dots,\left\lfloor\frac{N-1}{2} \right\rfloor.
			\end{dcases}
	\end{equation}}

	After this operation, we subtract the last right column multiplied by $9$ from the last left column.
	This changes the nonzero entries in the last left column at $j=N+2\ell-5$ to 
	\begin{align*}
		&\frac{3\cdot5}{(2\ell+1)\cdot (2\ell-1)\cdot(2\ell-3)!}+\frac{9\cdot 3\cdot 5}{(2\ell+3)\cdot(2\ell+1)\cdot(2\ell-1)\cdot(2\ell-3)!}\\
		&=\frac{3\cdot 5\cdot (2\ell+12)}{(2\ell+3)\cdot(2\ell+1)\cdot(2\ell-1)\cdot(2\ell-3)!},\quad \ell=3,4,\dots,\left\lfloor\frac{N}{2}\right\rfloor.
	\end{align*}
	This changes the nonzero entries in the last left column at $j=N+2\ell-4$ to \begin{align*}
		&\frac{3\cdot5\cdot(2\ell+7)}{(2\ell+3)\cdot(2\ell+1)\cdot(2\ell-1)!}-\frac{9\cdot 3\cdot 5}{(2\ell+3)\cdot(2\ell+1)\cdot (2\ell-1)!}\\
		&=\frac{3\cdot 5}{(2\ell+3)\cdot(2\ell+1)\cdot(2\ell-1)\cdot(2\ell-3)!},\quad \ell=2,3,\dots,\left\lfloor\frac{N-1}{2}\right\rfloor.
	\end{align*}
	Recompute nonzero entries of the last left column as
	\small{\begin{equation}
			\begin{dcases}
				(C_{1})_{N+2\ell-5}=\frac{3\cdot 5\cdot7\cdot(2\ell+12)}{(2\ell+3)\cdot(2\ell+1)\cdot (2\ell-1)\cdot(2\ell-3)!},\quad &\ell=3,4,\dots,\left\lfloor\frac{N}{2} \right\rfloor,\\
				(C_{1})_{N+2\ell-4}=\frac{3\cdot 5\cdot 7}{(2\ell+3)\cdot(2\ell+1)\cdot(2\ell-1)\cdot(2\ell-3)!}\quad &\ell=2,3,\dots,\left\lfloor\frac{N-1}{2} \right\rfloor.
			\end{dcases}
	\end{equation}}

	For $N=4$, this step yields
	\begin{align*}
		C^{(1)}_{4}&=
		\frac{1}{3^2 \cdot (5\cdot 3)^1\cdot3^2}\begin{vmatrix}
			4&1&1\\
			1&0&-\frac{1}{5}\\
			\frac{3}{5}&\frac{1}{7\cdot 2}&\frac{1}{5 \cdot 2}
		\end{vmatrix}=
		\frac{1}{3^2 \cdot(5\cdot 3)^1\cdot3^2}\begin{vmatrix}
			0&1&0\\
			\frac{9}{5}&0&-\frac{1}{5}\\
			\frac{1}{5}&\frac{1}{7\cdot 2}&\frac{1}{7 \cdot 5}
		\end{vmatrix}\\
		&=
		\frac{1}{3^2 \cdot (5\cdot 3)^1\cdot3^2\cdot 5^2}\begin{vmatrix}
			0&1&0\\
			9&0&-1\\
			1&\frac{1}{7\cdot 2}&\frac{1}{7}
		\end{vmatrix}
		=
		\frac{1}{3^2 \cdot (5\cdot 3)^1\cdot3^2\cdot 5^2}\begin{vmatrix}
			0&1&0\\
			0&0&-1\\
			\frac{16}{7}&\frac{1}{7\cdot 2}&\frac{1}{7}
		\end{vmatrix}\\
		&	=\frac{1}{3^2 \cdot (5\cdot 3)^1\cdot3^2\cdot 5^2\cdot 7^1}\begin{vmatrix}
			0&1&0\\
			0&0&-1\\
			16&\frac{1}{7\cdot 1}&1
		\end{vmatrix}
		=-\frac{4^2}{ 3^5\cdot 5^3\cdot 7^1},
	\end{align*}
	which completes the algorithm for $N=4$.
	
	\underline{\bf Step $k$.} This algorithm can be continued for any $2\leq k\leq N-2$ by alternately extracting the product common factor, square common factor, and order reduction between the left columns and right solumns. With the same recursive formula (\ref{induction}) but for $\ell=k,\dots,N-2$, we extract the numerical factor $\frac{1}{((2k-1) \cdot (2k+1))^{N-k-1}}$ in front of the determinant and modify the entries of $(N-k-1)$ columns of the first block to 
	\begin{equation*}
		\frac{3 \cdot 5 \cdot \ldots \cdot (2k+1)}{(2\ell+1) (2\ell-1)\cdots  (2 \ell - 2k + 3) \cdot (2\ell - 2k + 1)!}, \qquad \ell = k,\dots,N-1.
	\end{equation*}
	Then we eliminate the first two rows and the first two columns of the larger block. Afterwards, we perform operations on the last left column and the last right column. We subtract the last right column multiplied by $k^2$ from the last left column. Entries of the last left column are written in the form:
	\begin{tiny}\begin{equation}
			\begin{dcases}
				(C_{1})_{N+2\ell-2k+3}=\frac{3\cdot 5\cdot\ldots\cdot (2k-1)\cdot(2\ell+5k-8)}{(2\ell+2m-1)\cdot(2\ell+2m-3)\cdot\ldots\cdot(2\ell-2m+3)\cdot(2\ell-2m+1)!} ,\quad &\ell=m,\dots,\left\lfloor\frac{N}{2} \right\rfloor,\\
				(C_{1})_{N+2\ell-2k+4}=\frac{3\cdot 5\cdot\ldots\cdot (2k-1)}{(2\ell+2m-1)\cdot(2\ell+2m-3)\cdot\ldots\cdot(2\ell-2m+3)\cdot(2\ell-2m+1)!} ,\quad &\ell=m,\dots,\left\lfloor\frac{N-1}{2} \right\rfloor
			\end{dcases}
	\end{equation}\end{tiny}
	if $k=2m$ is even, and in the form:
	\begin{tiny}\begin{equation}
			\begin{dcases}
				(C_{2})_{N+2\ell-2k+3}=\frac{3\cdot 5\cdot\ldots\cdot (2k-1)\cdot(2\ell+5k-8)}{(2\ell+2m-1)\cdot(2\ell+2m-3)\cdot\ldots\cdot(2\ell-2m+3)\cdot(2\ell-2m+1)!} ,\quad &\ell=m-1,\dots,\left\lfloor\frac{N}{2} \right\rfloor,\\
				(C_{2})_{N+2\ell-2k+4}=\frac{3\cdot 5\cdot\ldots\cdot (2k-1)}{(2\ell+2m-3)\cdot(2\ell+2m-5)\cdot\ldots\cdot(2\ell-2m+3)\cdot(2\ell-2m+1)!} ,\quad &\ell=m,\dots,\left\lfloor\frac{N-1}{2} \right\rfloor
			\end{dcases}
	\end{equation}\end{tiny}
	if $k=2m-1$ is odd.
	
	Entries of the last right column are written in the form:
	
	\begin{tiny}\begin{equation}
			\begin{dcases}
				(C_{1})_{N+2\ell-2k+3}=\frac{3\cdot 5\cdot\ldots\cdot (2k-1)\cdot(2\ell+5k-8)}{(2\ell+2m-1)\cdot(2\ell+2m-3)\cdot\ldots\cdot(2\ell-2m+3)\cdot(2\ell-2m+1)!} ,\quad &\ell=m,\dots,\left\lfloor\frac{N}{2} \right\rfloor,\\
				(C_{1})_{N+2\ell-2k+4}=\frac{3\cdot 5\cdot\ldots\cdot (2k-1)}{(2\ell+2m-1)\cdot(2\ell+2m-3)\cdot\ldots\cdot(2\ell-2m+3)\cdot(2\ell-2m+1)!} ,\quad &\ell=m,\dots,\left\lfloor\frac{N-1}{2} \right\rfloor
			\end{dcases}
	\end{equation}\end{tiny}
	if $k=2m$ is even, and in the form:
	\begin{tiny}
		\begin{equation}
			\begin{dcases}
				(C_{2})_{N+2\ell-2k+3}=\frac{3\cdot 5\cdot\ldots\cdot (2k-1)\cdot(2\ell+5k-8)}{(2\ell+2m-1)\cdot(2\ell+2m-3)\cdot\ldots\cdot(2\ell-2m+3)\cdot(2\ell-2m+1)!} ,\quad &\ell=m-1,\dots,\left\lfloor\frac{N}{2} \right\rfloor,\\
				(C_{2})_{N+2\ell-2k+4}=\frac{3\cdot 5\cdot\ldots\cdot (2k-1)}{(2\ell+2m-3)\cdot(2\ell+2m-5)\cdot\ldots\cdot(2\ell-2m+3)\cdot(2\ell-2m+1)!} ,\quad &\ell=m,\dots,\left\lfloor\frac{N-1}{2} \right\rfloor
			\end{dcases}
		\end{equation}
	\end{tiny}
	if $k=2m-1$ is odd.

	When we recompute the entries in the above equations, we substitute the recursive formula from the previous step with $k \to k-1$ and $\ell \to \ell$ in the first term of the difference because the entries of the columns were not changed but the size of the block was reduced in the previous iteration, and we use the current parameters $k$ and $\ell$ in the second term because the entries were changed by the column operations of the current iteration. The algebraic simplification confirms that the results satisfy the general recursive formula. The process is then repeated cyclically, incrementing $k$ by $1$ until we reach $k = N-1$.
	
	\underline{\bf Step $k=N-1$.} In this last step, we perform (\ref{induction}) for the remaining columns, extract the remaining factor $\frac{1}{(2N-3)\cdot(2N-1)}$. We apply the recursive formula summarized in step $k$ to the last left column and the last right column, and obtain a trivial $3\times 3$ determinant, which is equal is  $-\frac{N^2}{2N-1}$, see examples for $N=2,3,4$. Combining together all numerical factors, we obtain 
	$$
	C^{(1)}_{N} = -\frac{N^2}{3^{N-2} \cdot (3 \cdot 5)^{N-3} \cdot\ldots\cdot ((2N-5) \cdot(2N-3))^1\cdot 3^2 \cdot 5^2\cdot\ldots\cdot(2N-3)^2 \cdot (2N-1)^1},
	$$
which recovers the numerical value (\ref{eq:C4}).
\end{proof}

\bibliographystyle{plain}

\begin{thebibliography}{99}
	
\bibitem{bilman1} D. Bilman and R. Buchingham, ``Large-order asymptotics for multiple-pole solitons of the focusing nonlinear Schr\"{o}dinger equation",
J. Nonlinear Sci. {\bf 29} (2019) 2185--2229.

\bibitem{bilman2} D. Bilman, R. Buchingham, and D.S. Wang, ``Far-field asymptotics for multiple-pole solitons in the large-order limit",
J. Differential Equations {\bf 297} (2021) 320--369.

\bibitem{bilman3} D. Bilman and P. D. Miller, ``Broader universality of rogue waves of infinite order", Phys. D {\bf 435} (2022)  133289 (12 pp).

\bibitem{bilman4} D. Bilman and P. D. Miller, ``Extreme superposition: high-order fundamental rogue waves in the far-field regime", 
Mem. Amer. Math. Soc. {\bf 300} (2024) 1505 (90 pp).
	
\bibitem{BC19} N. Boussaid and A. Comech, 
{\em Nonlinear Dirac equations: Spectral stability of solitary waves}, Mathematical Surveys and Monographs {\bf 244} (American Mathematical Society, Providence, RI, 2019).

\bibitem{Candy} T. Candy, ``Global existence for an $L^2$   critical nonlinear Dirac equation in one dimension",
Adv. Differential Equations {\bf 16} (2011) 643--666.
	
\bibitem{Chen-SAPM-2023}
J. Chen and B.  Feng, ``Tau-function formulation for bright, dark soliton and breather solutions to
the massive Thirring model", Stud. Appl. Math. \textbf{150} (2023) 35--68.

\bibitem{Feng2}
J. Chen, B. Yang, and B. Feng, ``Rogue waves in the massive Thirring model", Stud Appl Math. {\bf 151} (2023) 1020--1052.


\bibitem{CPS16} A. Contreras, D.E. Pelinovsky, and Y. Shimabukuro, ``$L^2$ orbital stability of Dirac solitons in the massive Thirring model", Commun. PDEs {\bf 41} (2016) 227--255.



\bibitem{Degasperis} A. Degasperis, S. Wabnitz, and A. B. Aceves, 
``Bragg grating rogue wave", Phys. Lett. A {\bf 379} (2015)  1067--1070.

\bibitem{Han-2024}
J. Han, C. He, and D. E. Pelinovsky, ``Algebraic solitons in the massive Thirring model", Phys. Rev. E {\bf 110} (2024) 034202 (11 pages)


\bibitem{Cheng23} C. He, J. Liu, and C. Qu, ``Massive Thirring model: inverse scattering and soliton resolution", 
Sci. China Math. {\bf 69} (2026)  39--92.

\bibitem{Huh} H.   Huh, ``Global strong solution to the Thirring model in critical space",  J. Math. Anal. Appl. {\bf 381} (2011) 513--520.

\bibitem{Guo-SAPM-2013}
B. Guo, L. Ling, and Q. P. Liu, ``High-order solutions and generalized Darboux transformations of derivative nonlinear Schr\"{o}dinger equations", Stud. Appl. Math. \textbf{130} (2013) 317--344.

\bibitem{He}
L. Guo, L. Wang, Y. Cheng, and J. He, ``High-order rogue wave solutions of the classical massive Thirring model equations", Commun. Nonlinear Sci. Numer. Simulat. {\bf 52}  (2017) 11--23.

\bibitem{Kaup-ANC-1977}
D.J. Kaup and A.C. Newell,  ``On the Coleman correspondence and the solution of the Massive Thirring model", Lett. Nuovo Cimento, \textbf{20} (1977) 325--331.

\bibitem{KPR06} M. Klaus, D.E. Pelinovsky, and V.M. Rothos,  ``Evans function for Lax operators with algebraically decaying potentials", J. Nonlin. Science {\bf 16} (2006) 1--44.

\bibitem{Kuznetsov-TMP-1977}
E.A. Kuznetsov and A.V. Mikhailov, ``On the complete integrability of the two-dimensional
classical Thirring model", Theor. Math. Phys. \textbf{30} (1977) 193--200.

\bibitem{Wu18} S. Kwon and Y. Wu, ``Orbital stability of solitary waves for derivative nonlinear Schr\"{o}dinger equation", 
J. Anal. Math. {\bf 135} (2018)  473--486.

\bibitem{LiPelin2025} Z. Li, D. E. Pelinovsky, and S. Tian, ``Exponential and algebraic double-soliton solutions of the massive Thirring model", J. Math. Phys. {\bf 66} (2025) 101529 (24 pages)

\bibitem{HirotaOhtaSatsuma1988} R. Hirota, Y. Ohta, and J. Satsuma,  ``Wronskian structures of solutions for soliton equations", Prog. Theor. Phys. Suppl.  {\bf 94} (1988) 59--72.

\bibitem{Kakei1988} S. Kakei, N. Sasa, and J. Satsuma,  ``Bilinearization of generalized derivative Schr\"{o}dinger equation", J. Phys. Soc. Jpn.  {\bf 64} (1995) 1519--1523.

\bibitem{Wang1} S. Z. Liu, J. Wang, and D. J. Zhang, ``The Fokas-Lenells equations: Bilinear approach", Stud. Appl. Math. {\bf 148} (2022) 651--688.

\bibitem{Mikhailov-JETP-1976}
A.V. Mikhailov, ``Integrability of the two-dimensional Thirring model", JETP Lett. \textbf{23} (1976) 320--323.

\bibitem{Orfanidis-PRD-1976}
S. J. Orfanidis,  ``Soliton solutions of the massive Thirring model and the inverse scattering transform", Phys. Rev. D \textbf{14} (1976) 472--478.

 \bibitem{Pel11}
D. E. Pelinovsky {\it Localization in Periodic Potentials; from Schr\"{o}dinger Operators to the Gross--Pitaevskii Equation}, London Mathematical Society Lecture Note Series, 2011. 

\bibitem{PS19} D.E. Pelinovsky and A. Saalmann, ``Inverse scattering for the massive Thirring model" in {\em Nonlinear Dispersive Partial Differential Equations and Inverse Scattering} (Editors: P. Miller, P. Perry, J.C. Saut, and C. Sulem), Fields Institute Communications {\bf 83} (Springer, New York, NY, 2019) 497--528.
 
\bibitem{PS14} D.E. Pelinovsky and Y. Shimabukuro, ``Orbital stability of Dirac solitons", Lett. Math. Phys. {\bf 104} (2014) 21--41.

\bibitem{PS18} D. E. Pelinovsky and Y. Shimabukuro, ``Existence of global solutions to the derivative NLS equation with the inverse scattering transform method", Inter. Math. Res. Notices {\bf 2018} (2019) 5663--5728.

\bibitem{Thirring-AOP-1958}
W. Thirring,  ``A soluble relativistic field theory", Ann. Phys. \textbf{3} (1958) 91--112.

\bibitem{Wang2} J. Wang and H. Wu, ``Rational solutions with zero background and algebraic solitons of three derivative nonlinear Schr\"{o}dinger equations:
bilinear approach", Nonlinear Dyn. {\bf 109} (2022) 3101--3111. 

\bibitem{Xu-JPA-2011}
S. Xu, J. He, and L. Wang, ``The Darboux transformation of the derivative nonlinear Schr\"{o}dinger equation", J. Phys. A: Math. Theor. \textbf{44} (2011) 305203 (22 pages).

\bibitem{YY1} B. Yang and J. Yang, ``Rogue wave patterns in the nonlinear
Schr\"{o}dinger equation", Physica D {\bf 419} (2021) 132850.

\bibitem{YY2} B. Yang and J. Yang, ``Universal rogue wave patterns associated
with the Yablonskii–Vorob\'ev polynomial hierarchy", Physica D
{\bf 425} (2021) 132958.

\bibitem{YY3} B. Yang and J. Yang, ``Rogue wave patterns associated with
Okamoto polynomial hierarchies", Stud. Appl. Math. {\bf 151} (2023) 60--115.

\bibitem{YY} B. Yang and J. Yang, {\em Rogue Waves in Integrable Systems} 
(Springer, 2024).

\bibitem{Ye}
Y. Ye, L. Bu, C. Pan, S. Chen, D. Mihalache, and F. Baronio, ``Super rogue wave states in the classical massive Thirring model system", Rom. Rep. Phys. 73 (2021) 117 (16 pages).

\end{thebibliography}

\end{document}